\documentclass{llncs}
\usepackage{pta,stackrel}

\title{On The Reachability Problem for Recursive Hybrid Automata with One and Two Players}
\author{
  Shankara Narayanan Krishna, 
  Lakshmi Manasa, and
  Ashutosh Trivedi}
\institute{
  Indian Institute of Technology Bombay, India,
  \email{krishnas,manasa,trivedi@cse.iitb.ac.in}}

\begin{document}
\maketitle
\begin{abstract}
  Motivated by the success of bounded model checking framework for finite state
  machines, Ouaknine and Worrell proposed a time-bounded theory of real-time
  verification by claiming that restriction to bounded-time recovers
  decidability for several key decision problem related to real-time
  verification. 
  In support of this theory, the list of undecidable problems recently shown
  decidable under time-bounded restriction is rather impressive: language inclusion for timed
  automata, emptiness problem for alternating timed automata, and emptiness
  problem for rectangular hybrid automata.
  The objective of our study was to recover decidability for general recursive
  timed automata---and perhaps for recursive hybrid automata---under time-bounded
  restriction in order to provide an appealing verification framework for
  powerful modeling environments such as Stateflow/Simulink.
  Unfortunately, however, we answer this question in negative by showing that
  time-bounded reachability problem stays undecidable for recursive timed
  automata with five or more clocks.  
  While the bad news continues even when one considers 
  well-behaved subclasses of recursive hybrid automata, 
  we recover decidability by considering recursive hybrid automata 
  with bounded context using a pass-by-reference mechanism, or 
  by restricting the number of variables to two, with rates in $\{0,1\}$. 
\end{abstract}

\section{Introduction}
Recursive state machines (RSMs), as introduced by Alur, Etessami, and
Yannakakis~\cite{AEY01}, are a variation on various visual notations to
represent hierarchical state machines, notably Harel's statecharts~\cite{Har87}
and Object Management Group supported UML diagrams~\cite{RJB99}, that permits
recursion  while disallowing concurrency.  
RSMs closely correspond~\cite{AEY01} to pushdown systems~\cite{BEM97},
context-free grammars, and Boolean programs~\cite{BR00}, and provide a natural
specification and verification framework to reason with sequential programs with
recursive procedure calls.
The two fundamental verification questions for RSM, namely reachability and
B\"uchi emptiness checking, are known to be decidable in polynomial
time~\cite{AEY01,EHRS00}.

Hybrid automata~\cite{AD90,ACHH93} extend finite state machines with continuous
variables that permit a natural modeling of hybrid systems.
In a hybrid automaton the variables continuously flow  according to a given set
of ordinary differential equations within each discrete states, while they are
allowed to have discontinuous jumps during transitions between states that are
guarded by constraints over variables. 
In this paper we study the reachability problem for recursive hybrid automata that
generalize recursive state machines with continuous variables, or equivalently
hybrid automata with recursion. 

In this paper we restrict our attention to so-called \emph{singular hybrid
  automata} where the dynamics of every variable is restricted to
state-dependent constant rates.  
A variable is often called a \emph{clock} if its rate over all the states is
$1$, while it is called a \emph{stopwatch} if its rate is either $0$ (clock is
stopped) or $1$ (clock is ticking) in different states. 
A timed automaton~\cite{AD90} is a hybrid automaton where all the variables are
clocks, while a stopwatch automaton~\cite{HKPV98} is a hybrid automaton where
all the variables are stopwatches. 
It is well known that the reachability problem is decidable (PSPACE-complete)
for timed automata~\cite{AD90} and undecidable for stopwatch 
automata with $3$ stopwatches~\cite{Cer92}. 

Trivedi and Wojtczak~\cite{TW10} introduced \emph{recursive timed automata}
(RTAs) as an extension of timed automata with recursion to model real-time software
systems. 
Formally, an RTA is a finite collection of components where each component is a
timed automaton that in addition to making transitions between various states,
can have transitions to ``boxes'' that are mapped to other components modeling a
potentially recursive call to a subroutine. 
During such invocation a limited information can be  passed through clock values 
from the ``caller'' component to the ``called'' component via two different
mechanism: a) \emph{pass-by-value}, where upon returning from the called
component a clock assumes the value prior to the invocation, and b)
\emph{pass-by-reference}, where upon return clocks reflect any changes to the
value inside the invoked procedure.  
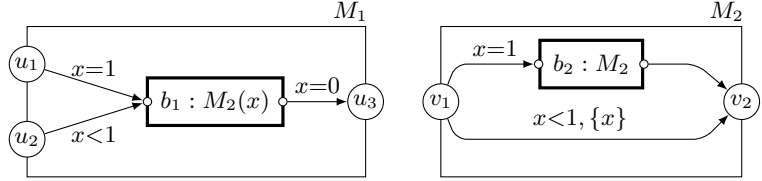
\begin{figure}[t]
  \centering
 \vspace*{-0.4cm}
  \makebox[2cm]{{\small
  \begin{tikzpicture}[node distance=4cm]

    \draw(-2, -1) rectangle (2.5,1);
    \draw (2.3, 1.2) node {$M_1$};
    
    \node[loc](u1) at (-2, 0.5) {$u_1$};
    \node[loc](u2) at (-2, -0.5) {$u_2$}; 

    \node[loc](u3) at (2.5, 0) {$u_3$}; 
    
    \node[boxloc](b1) at (0.5, 0) {$~b_1:M_2(x)~$};
    \node[port](b1v1) at (-0.4, 0) {}; 
    \node[port](b1v2) at (1.4, 0) {}; 

    \draw[trans] (u1)--(b1v1)  node [midway, above]{$x {=} 1$};
    \draw[trans] (u2)--(b1v1)  node [midway, below]{$x {<} 1$};
    \draw[trans] (b1v2)--(u3) node [midway, above]{$x {=} 0$};
    
    %%%%%%%%%%%%%%%%%%%%%%%%%%%%%%%%%%%%%%%%%%%%%%%%%%%%%%%%%%%%%%%%%%%%%%%%%%%%%%%  
    %%%%%%%%%%%%%%%%%%%%%%%%%%%%%%%%%%%%%%%%%%%%%%%%%%%%%%%%%%%%%%%%%%%%%%%%%%%%%%%  
    \draw(3.5, -1) rectangle (7.5,1);
    \draw (7.3, 1.2) node {$M_2$};
    
    \node[loc](v1) at (3.5, 0) {$v_1$};
    \node[loc](v2) at (7.5, 0) {$v_2$};
    
    \node[boxloc](c1) at (5.5, 0.5) {$~b_2:M_2~$};
    \node[port](c1v1) at (4.8, 0.5) {}; 
    \node[port](c1v2) at (6.2, 0.5) {}; 
    
    \draw[trans] (v1)-- +(0.2, 0.5) --(c1v1) node [midway, above]{$x {=} 1$} ;
    \draw[trans] (c1v2)-- +(0.7, 0) -- (v2);
    \draw[trans] (v1)-- + (0.2, -0.5) -- +(3.5, -0.5) node [midway,
    above]{$x {<} 1, \set{x}$} -- (v2);
    %%%%%%%%%%%%%%%%%%%%%%%%%%%%%%%%%%%%%%%%%%%%%%%%%%%%%%%%%%%%%%%%%%%%%%%%%%%%%%% 
    
  \end{tikzpicture}
}}
  \vspace*{-0.4cm}
  \caption{An example of recursive timed automata with one clock and two components}
  \vspace*{-0.4cm}
  \label{fig:example2}
\end{figure}
\begin{example}
  The visual presentation of a recursive timed automaton with
  two components $M_1$ and $M_2$, and one clock variable $x$
  is shown in Figure~\ref{fig:example2} (example taken from~\cite{TW10}), where
  component $M_1$ calls component $M_2$ via box $b_1$ and component $M_2$
  recursively calls itself via box $b_2$.
  Components are shown as thinly framed rectangles with their names 
  written next to upper right corner.
  Various control states, or ``nodes'', of the components are shown as circles
  with their labels written inside them, e.g. see node $u_1$.
  Entry nodes of a component appear on the left of the component
  (see $u_1$), while exit nodes appear on the right (see $u_3$).
  Boxes are shown as thickly framed rectangles inside components
  labeled $b:M(C)$, where $b$ is the label of the box, $M$ is the
  component it is mapped to, and $C$ is the set of clocks passed to $M$ by
  value and the rest of the variables are passed by reference.
  When the set $C$ is empty, we just write $b:M$ for $b:M(\emptyset)$.
%  Call ports of boxes are drawn as small circles on the left of the
%  box, while return ports are on the right.
%  We omit labelling the call and return ports as these labels are clear from
%  their position on the boxes.
%  For example, call port $(b_1, v_1)$ is the top small circle on the
%  left-hand side of box $b_1$, since box $b_1$ is mapped to $M_2$ and $v_1$
%  is the top node on its left-hand side.
  Each transition is labelled with a guard and the set of reset variables,
  (e.g. transition from node $v_1$ to $v_2$ can be taken only
  when variable $x{<}1$, and after taking this transition, variable $x$ is
  reset).
  To minimize clutter we omit empty reset sets.
\end{example}

Trivedi and Wojtczak~\cite{TW10} showed that the reachability and termination
(reachability with empty calling context) problem is undecidable for RTAs with
three or more clocks. 
Moreover, they considered the so-called glitch-free restriction of RTAs---where at
each invocation either all clocks are passed by value or all clocks are passed
by reference--- and showed that the reachability (and termination) is
EXPTIME-complete for RTAs with two or more clocks. 
In the model of~\cite{TW10} it is compulsory to pass all the clocks at every
invocation with either mechanism. 
Abdulla, Atig, and Stenman~\cite{AAS12} studied a related model called
timed pushdown automata where they disallowed passing clocks by value. 
On the other hand, they allowed clocks to be passed either by reference or not
passed at all (in that case they are stored in the call context and continue to
tick with the uniform rate). 
It is shown in~\cite{AAS12} that the reachability problem for this class remains
decidable (EXPTIME-complete).
In this paper we restrict ourselves to the recursive timed automata model as
introduced in~\cite{TW10}.

\noindent{\bf Contributions.}
In this paper we consider time-bounded reachability problem for RTA and show
that the problem stays undecidable for RTA with $5$ or more clocks. 
We also consider the extension of RTAs to recursive hybrid  
automata (RHAs) and show that the  reachability problem stays undecidable even
for glitch-free RHAs with $3$ or more stopwatches (clocks that can be paused),
while we show decidability of glitch-free RHAs with $2$ stopwatches.  
We also show that the reachability problem is undecidable for unrestricted RHA with two or
more stopwatches. For the time-bounded reachability case, we show that the problem stays
undecidable even for glitch-free variant of RHAs with $14$ or more stopwatches. 
On the positive side, we show decidability of time-bounded reachability in glitch-free RHAs where 
with pass-by-reference only mechanism. 
Our results are summarized and compared with known results in
Table~\ref{tab:res}.

We study these problems for two player games on RTA and RHA also. The undecidability saga continues even for games 
with lesser number of clocks than single-player case. The results for games have been summarized in 
Table~\ref{tab:res_game}.

\begin{table}[t]
\begin{tabular}{lcccc}%|l||l|l||l|l|}
\toprule
&\multicolumn{2}{c}{~~~~~~~Recursive Timed
  Automata~~~~~~~~~~}&\multicolumn{2}{c}{
  ~~~~~~~~~Recursive
  Hybrid Automata~~~~~~~~}\\
%\cline{2-5}
 &TUB& TB&TUB&TB\\
\toprule
Pass by reference&{\color{gray} D }& {\color{gray} D} &{\color{blue} U} ($\geq 3$ sw)\cite{Cer92} & {\textbf{D}}(Bounded \\
		  &		    &			&					     & context)\\
\hline
Pass by value& {\color{gray} D} & {\color{gray} D}& \textbf{U} $(\geq 3$ sw) &{\textbf U} ($\geq 14$ sw)\\
 \hline 
 Glitch free&  {\color{gray} D} &  {\color{gray} D} &  \textbf{U} ($\geq 3$ sw)  &{\textbf U}($\geq 14$ sw)\\
 &       &     & \textbf{D} $(\leq 2$ sw)& \textbf{D} $(\leq 2$ sw)\\
%                 &       &     & holds for RSA& holds for RSA\\
      \hline
Unrestricted &  {\color{gray} U ($\geq 3$ clocks)} &  \textbf{U} ($\geq 5$ clocks)&\textbf{U} $(\geq 2$ sw) &\textbf{U} ($\geq 5$ sw) \\
\hline
\end{tabular}
\vspace{1em}
\caption{Summary of results related to RHA - single player game. Results shown in bold are contributions form
  this paper, while results shown in gray color are contributions
  from~\cite{TW10}. Here TUB  and TB stand for time-unbounded and time-bounded
  reachability problems, respectively. $\text{U}$ stands for undecidable, $\text{D}$ for decidable, and $\text{sw}$ for stopwatches.}
\label{tab:res}
\vspace{-2em}
\end{table}

\begin{table}[t]
  \begin{tabular}{lcccc}%|l||l|l||l|l|}
    \toprule
    &\multicolumn{2}{c}{~~~~~~~Recursive Timed Automata~~~~~~~~~~}&\multicolumn{2}{c}{
      ~~~~~~~~~Recursive Stopwatch Automata~~~~~~~~}\\
    &TUB& TB & TUB & TB\\
    \toprule
    Glitch free&  {\color{gray} D} &  {\color{gray} D} &  \textbf{U} ($\geq 3$ sw)  &{\textbf U}($\geq 4$ sw)\\
    &       &     & \textbf{D} $(\leq 2$ sw)& \textbf{D} $(\leq 2$ sw)\\
    \hline
    Unrestricted &  {\color{gray} U ($\geq 2$ clocks)} &  \textbf{U} ($\geq 3$ clocks)&\textbf{U} $(\geq 2$ sw) &\textbf{U} ($\geq 3$ sw) \\
    \hline
  \end{tabular}
  \caption{Summary of results for two-player games. Results shown in bold are
    contributions form this paper, while results shown in gray color are
    from~\cite{TW10}.} 
\vspace{-1em}
\label{tab:res_game}
\end{table}

\noindent{\bf Related work.}
For a survey of models related to recursive timed automata and dense-time
pushdown automata we refer the reader to~\cite{TW10}and~\cite{AAS12}.
Another closely related model is introduced by Benerecetti, Minopoli, and
Peron~\cite{BMP10} where pushdown automata is extended with an additional stack
used to store clock valuations. 
The reachability problem is known to be undecidable for this model.
We do not consider this model in the current paper, but we conjecture that
time-bounded reachability problem for this model is also undecidable. 

Two special kinds of RSMs with restricted recursion are 
hierarchical RSMs and bounded stack RSMs \cite{popl08}.
\cite{popl08} gives efficient algorithms for the reachability analysis of hierarchical and bounded stack RSMs, 
and algorithms for the latter might be useful in the analysis of programs without infinite recursion.  
The language theory of bounded context recursion 
 has been studied recently \cite{bc1}.   
Hierarchical hybrid systems studied by Alur et
al.~\cite{ADEF01} are a  restriction of recursive hybrid
automata where recursion is disallowed but concurrency is allowed.
A number of case-studies with the tool CHARON~\cite{ADEF01} demonstrate the benefit of
hierarchical modeling of hybrid systems. 

\noindent{\bf Organization.} 
In the next section, we begin by reviewing the definition of recursive state
machines followed by  a formal definition of recursive (singular) hybrid
automata. 
We also formally define the termination and the reachability problems for two players on this
model, and present our main results. 
Finally, Sections~\ref{sec:undec-rha} and \ref{sec:undec-rhg} details our main undecidability
results while Sections~\ref{sec:dec-ref} and \ref{sec:dec-rha} discuss our decidability results.

\section{Preliminaries}
\subsection{Reachability Games on Labelled Transition Systems.}
A \emph{labeled transition system} (LTS) is a tuple $\lts = (S, A, \trans)$
where $S$ is the set of \emph{states}, 
$A$ is the set of \emph{actions}, and $\trans : S {\times} A \to S$ is
the \emph{transition function}.
We say that an LTS $\lts$ is \emph{finite} (\emph{discrete}) if both $S$
and $A$ are finite (countable).
We write $A(s)$ for the set of actions available at $s \in S$, i.e., $A(s) =
\set{a \::\: \trans(s, a) \not = \emptyset}$. 
A \emph{game arena} $G$ is a tuple $(\lts, S_\Ach, S_\Tor)$, where $\lts = (S,
A, \trans)$ is an LTS, $S_\Ach \subseteq S$ is the set of states controlled by
player \ach, and $S_\Tor \subseteq S$ is the set of states controlled by $\tort$. 
Moreover, sets  $S_\Ach$ and $S_\Tor$ form a partition of the set
$S$.
In a \emph{reachability game} on $G$ rational players---\ach
and \tort---take turns to move a token along the states of $\lts$.
The decision to choose the successor state is made by the player
controlling the current state.
The objective of \ach is to eventually reach certain states, while the
objective of \tort is to avoid them forever.

We say that $(s, a, s') \in S {\times} A {\times} S$ is a transition of
$\lts$ if $s' = \trans(s, a)$ and a \emph{run} of $\lts$ is a sequence
$\seq{s_0, a_1, s_1, \ldots} \in S {{\times}} (A {{\times}} S)^*$
such that $(s_i, a_{i+1}, s_{i+1})$ is a transition of $\lts$ for all
$i \geq 0$.
We write $\RUNS^\lts$ ($\FRUNS^\lts$) for the sets of infinite (finite)
runs and $\RUNS^\lts(s)$ ($\FRUNS^\lts(s)$) for the
sets of infinite (finite) runs starting from state~$s$.
For a set $F \subseteq S$ and a run $r = \seq{s_0, a_1, \ldots}$ we
define $\STOP(F)(r) = \inf \set{ i\in \Nat \::\: s_i \in F}$.
Given a state $s \in S$ and a set of final states $F \subseteq S$ we
say that a final state is reachable from $s_0$ if there is a run $r \in
\RUNS^\lts(s_0)$ such that $\STOP(F)(r) < \infty$.
A strategy of \ach is a partial function
$\alpha :\FRUNS^\lts \to A$ such that for a run $r \in \FRUNS^\lts$ we
have that $\alpha(r)$ is defined if $\LAST(r) \in S_\Ach$, and
$\alpha(r) \in A(\LAST(r))$ for every such $r$.
A strategy of \tort is defined analogously.
Let $\LSigma_\Ach$ and $\LSigma_\Tor$ be the set of strategies of
\ach and \tort, respectively.
The unique run $\RUN(s, \alpha, \tau)$ from a state $s$ when players use
strategies $\alpha \in \LSigma_\Ach$ and $\tau \in \LSigma_\Tor$ is
defined in a straightforward manner.

For an initial state $s$ and a set of final states $F$, the lower
value $\LVAL^\lts_F(s)$ of the reachability game is defined as the
upper bound on the number of transitions that \tort can ensure before
the game visits a state in $F$ irrespective of the strategy of \ach,
and is equal to
$\sup_{\tau \in \LSigma_\Tor} \inf_{\alpha \in \LSigma_\Ach}
\STOP(F)(\RUN(s, \alpha, \tau))$.
The concept of upper value is $\UVAL^\lts_F(s)$ is analogous and
defined as $\inf_{\alpha \in \LSigma_\Ach} \sup_{\tau \in \LSigma_\Tor}
\STOP(F)(\RUN(s, \alpha, \tau))$.
If ${\LVAL^\lts_F(s) = \UVAL^\lts_F(s)}$ then we say that the reachability
game is determined, or the value  $\VAL^\lts_F(s)$ of the
reachability game exists and it is such that ${\VAL^\lts_F(s) =
\LVAL^\lts_F(s) = \UVAL^\lts_F(s)}$.
We say that \ach wins the reachability game if $\VAL^\lts_F(s) <
\infty$.
A \emph{reachability game problem} is to decide whether in a given
game arena $G$, an initial state $s$ and a set of final states $F$,
\ach has a strategy to win the reachability game.

\subsection{Reachability Games on Recursive state machines}
A \emph{recursive state machine}~\cite{ABEGRY05} $\Mm$ is a tuple 
$(\Mm_1, \Mm_2, \ldots, \Mm_k)$ of components, where each component $\Mm_i =
(N_i, \En_i, \Ex_i, B_i, Y_i, A_i, \trans_i)$ for each $1 \leq i \leq k$  is
such that:
\begin{itemize}
\item
  $N_i$ is a finite set of \emph{nodes} including a distinguished set $\En_i$ of
  \emph{entry nodes} and a set $\Ex_i$ of \emph{exit nodes} such that $\Ex_i$
  and $\En_i$ are disjoint sets; 
\item
  $B_i$ is a finite set of \emph{boxes};
\item
  $Y_i: B_i \to \set{1, 2, \ldots, k}$ is a mapping %the \emph{boxes-to-components mapping}
  that assigns every box to a component.
  We associate a set of \emph{call ports} $\call(b)$ and return ports
  $\return(b)$ to each box $b \in B_i$:
  \begin{eqnarray*}
  \call(b) = \set{(b, en) \::\: en \in \En_{Y_i(b)} } &\text{ and } &
  \return(b) = \set{(b, ex) \::\: ex \in  \Ex_{Y_i(b)}}.
  \end{eqnarray*}
  Let $\call_i = \cup_{b \in B_i} \call(b)$ and
    $\return_i = \cup_{b \in B_i} \return(b)$ be the set of call and
    return ports of component $\Mm_i$.
    We define the set of locations $Q_i$  of component $\Mm_i$ as the 
    union of the set of nodes, call ports and return ports, i.e. 
    $Q_i = N_i  \cup \call_i \cup \return_i$;
  \item
    $A_i$ is a finite set of \emph{actions}; and
  \item
    $\trans_i : Q_i {{\times}} A_i \to Q_i$ is the transition function  with a
    condition that call ports and exit nodes do not have any outgoing transitions.
  \end{itemize}
For the sake of simplicity, we assume that the set of boxes $B_1,
\ldots, B_k$ and set of nodes $N_1, N_2, \ldots, N_k$ are mutually
disjoint.
We use symbols $N, B, A, Q, \trans$, etc. to denote
the union of the corresponding symbols over all components.

\begin{figure}[t]
  \centering
 \vspace*{-0.4cm}
  \makebox[2cm]{{\small
  \begin{tikzpicture}[node distance=4cm]

    \draw(-2, -1) rectangle (2,1);
    \draw (1.8, 1.2) node {$M_1$};

    \node[loc](u1) at (-2, 0.5) {$u_1$};
    \node[loc](u2) at (-2, -0.5) {$u_2$};

    \node[loc](u4) at (2, 0) {$u_4$};

    \node[boxloc](b1) at (-0.4, 0.5) {$~b_1:M_2~$};
    \node[port](b1v1) at (-1.05, 0.6) {};
    \node[port](b1v2) at (-1.05, 0.4) {};
    \node[port](b1v3) at (0.26, 0.6) {};
    \node[port](b1v4) at (0.26, 0.4) {};

    \node[boxloc](b2) at (-0.4, -0.5) {$~b_2:M_3~$};
    \node[port](b2w1) at (0.25, -0.5) {};
    \node[port](b2w2) at (-1.05, -0.5) {};

    \node[loc](u3) at (1.15, -0.5) {$u_3$};

    \draw[trans] (u1)--(b1v1);
    \draw[trans] (u2)--(b2w2);
    \draw[trans] (u3)--(u4);
    \draw[trans] (b2w1)--(u3);
    \draw[trans] (b1v3)--(u4);
    \draw[trans] (b1v4)-- +(0.25, 0) -- +(0.25, -0.4) --+(-1.7, -0.4) --
    +(-1.7, 0) -- +(-1.4, 0);

    %%%%%%%%%%%%%%%%%%%%%%%%%%%%%%%%%%%%%%%%%%%%%%%%%%%%%%%%%%%%%%%%%%%%%%%%%%%%%%%
    %%%%%%%%%%%%%%%%%%%%%%%%%%%%%%%%%%%%%%%%%%%%%%%%%%%%%%%%%%%%%%%%%%%%%%%%%%%%%%%
    \draw(3, -1) rectangle (6,1);
    \draw (5.8, 1.2) node {$M_2$};

    \node[loc](v1) at (3, 0.5) {$v_1$};
    \node[loc](v2) at (3, -0.5) {$v_2$};
    \node[loc](v3) at (6, 0.5) {$v_3$};
    \node[loc](v4) at (6, -0.5) {$v_4$};

    \node[boxloc](c1) at (4.5, 0.5) {$~c_1:M_2~$};
    \node[port](c1v1) at (3.85, 0.6) {};
    \node[port](c1v2) at (3.85, 0.4) {};
    \node[port](c1v3) at (5.14, 0.6) {};
    \node[port](c1v4) at (5.14, 0.4) {};

    \node[boxloc](c2) at (4.5, -0.5) {$~c_2:M_3~$};
    \node[port](c2w1) at (3.85, -0.5) {};
    \node[port](c2w2) at (5.14, -0.5) {};

    \draw[trans] (v1)--(c1v1);
    \draw[trans] (v2)-- + (0.2, 0.6) -- (c1v2);
    \draw[trans] (v2) -- (c2w1);
    \draw[trans] (c1v3) -- +(0.2, 0)--(v4);
    \draw[trans] (c1v4)--(v3);
    \draw[trans] (c2w2) --(v4);
    \draw[trans] (c2w2)-- +(0.2, 0.2) -- +(0.2, 0.5) -- +(-1.6, 0.5) --(c1v2);

    %%%%%%%%%%%%%%%%%%%%%%%%%%%%%%%%%%%%%%%%%%%%%%%%%%%%%%%%%%%%%%%%%%%%%%%%%%%%%%%
    %%%%%%%%%%%%%%%%%%%%%%%%%%%%%%%%%%%%%%%%%%%%%%%%%%%%%%%%%%%%%%%%%%%%%%%%%%%%%%%

    \draw(7, -1) rectangle (9.5,1);
    \draw (9.3, 1.2) node {$M_3$};

    \node[loc](w1) at (7, 0) {$w_1$};
    \node[loc](w2) at (9.5, 0) {$w_2$};

    \node[boxloc](d) at (8.3, 0.5) {$~d:M_1~$};
    \node[port](du1) at (7.7, 0.6) {};
    \node[port](du2) at (7.7, 0.4) {};
    \node[port](du4) at (8.9, 0.5) {};

    \draw[trans] (w1) .. controls +(60:0.5) and +(180:0.8) .. (du1);
    \draw[trans] (du4) -- (w2);
    \draw[trans] (w1) edge[bend right=25] (w2);

  \end{tikzpicture}
}}
  \vspace*{-0.4cm}
  \caption{Example recursive state machine taken from~\cite{ABEGRY05}}
  \vspace*{-0.4cm}
 \label{fig:example}
\end{figure}
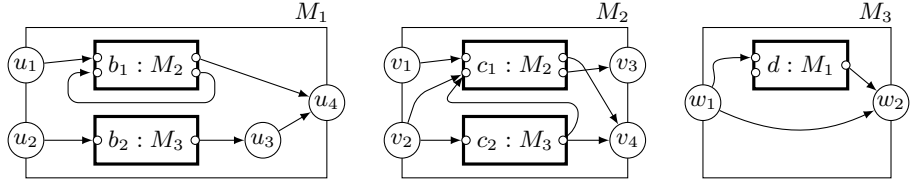
An example of a RSM is shown in Figure~\ref{fig:example} (taken from \cite{TW10}).
An execution of a RSM begins at the entry node of some component and depending
upon the sequence of input actions the state evolves naturally like a labeled
transition system. 
However, when the execution reaches an entry port of a box, this box is stored
 on a stack of pending calls, and the execution continues naturally from the
 corresponding entry node of the component mapped to that box. 
 When an exit node of a component is encountered, and if the stack of pending
 calls is empty then the run terminates; otherwise, it pops the box from the top
 of the stack and jumps to the exit port of the just popped box corresponding
 to the just reached exit of the component. 
We formalize the semantics of a RSM using a discrete LTS, whose states are pairs
consisting of a sequence of boxes, called the context, mimicking the stack of
pending calls and the current location.
%\begin{definition}[RSM semantics]
Let $\Mm = (\Mm_1, \Mm_2, \ldots, \Mm_k)$ be an RSM where the component
$\Mm_i$ is $(N_i, En_i, Ex_i, B_i, Y_i, A_i, \trans_i)$.
The semantics of $\Mm$ is the discrete  labelled transition system
$\sem{\Mm} = (S_\Mm, A_\Mm, \trans_\Mm)$ where:
\begin{itemize}
\item
  $S_\Mm \subseteq B^* {\times} Q$ is the set of states;
\item
  $A_\Mm = \cup_{i=1}^{k} A_i$ is the set of actions;
\item
  $\trans_\Mm : S_\Mm {\times} A_\Mm \to S_\Mm$ is the transition
  function such that for $s = (\sseq{\kappa}, q) \in S_\Mm$ and
  $a \in A_\Mm$, we have that $s' =  \trans_\Mm(s, a)$ if and only
  if one of the following holds:
  \begin{enumerate}
  \item
    the location $q$ is a call port, i.e. $q = (b, en) \in \call$,
    and $s' = (\sseq{\kappa, b}, en)$;
  \item
    the location $q$ is an exit node, i.e. $q = ex \in \Ex$
    and $s' = (\sseq{\kappa'}, (b, ex))$ where
    $(b, ex) \in \return(b)$ and $\kappa = (\kappa', b)$;
  \item
    the location $q$ is any other kind of location, and
    $s' = (\sseq{\kappa}, q')$ and $q' \in \trans(q, a)$.
  \end{enumerate}
\end{itemize}

Given $\Mm$ and a subset $Q' \subseteq Q$ of its nodes we define 
$\sem{Q'}_\Mm$ as  $\set{(\sseq{\kappa}, v') \::\: \kappa \in B^* \text{ and }
  v' \in Q'}$.
We define the terminal configurations $\term_\Mm$ as
the set $\set{(\sseq{\varepsilon}, ex) \::\: ex \in \Ex}$ with the empty context 
$\sseq{\varepsilon}$.
Given a recursive state machine  $\Mm$, an initial node $v$, and a set
of \emph{final locations} $F \subseteq Q$ the {\em reachability problem}
on $\Mm$ is defined as the reachability problem on the LTS $\sem{\Mm}$ with
the initial state $(\sseq{\varepsilon}, v)$ and final states $\sem{F}$.
We define \emph{termination problem} as the reachability of one of the exits with the empty context.
%The following is a well known result.
%\begin{theorem}[\cite{ABEGRY05}]
The reachability and the termination problem for recursive state machines can
be solved in polynomial time~\cite{ABEGRY05}. 
%\end{theorem}

%\paragraph{Games on Recursive State Machines.}
A partition $(Q_\Ach, Q_\Tor)$ of locations $Q$ of an RSM $\Mm$
(between \ach and \tort) gives rise to recursive game arena
$G =(\Mm, Q_\Ach, Q_\Tor)$.
Given an initial state, $v$, and a set of final states, $F$, the
reachability game on $\Mm$ is defined as the reachability game on the game
arena $(\sem{\Mm}, \sem{Q_\Ach}_\Mm, \sem{Q_\Tor}_\Mm)$ with the initial state
$(\sseq{\varepsilon}, v)$ and the set of final states $\sem{F}_\Mm$.
Also, the termination game $\Mm$ is defined as the reachability game on the
game arena $(\sem{\Mm}, \sem{Q_\Ach}_\Mm, \sem{Q_\Tor}_\Mm)$ with the initial
state $(\sseq{\varepsilon}, v)$ and the set of final states $\term_\Mm$.
It is a well known result (see, e.g.~\cite{Wal96},\cite{Ete04}) that
reachability games and termination games on RSMs are determined and decidable
(EXPTIME-complete).

\section{Recursive Hybrid Automata}
\label{sec:recursive-timed-automata}
Recursive hybrid automata (RHAs)  extend classical hybrid automata (HAs) with
recursion in a similar way RSMs extend LTSs.
We study a rather simpler subclass of HA known as singular hybrid automata where
all variables grow with constant-rates. 

\subsection{Syntax}
 Let $\Real$ be the set of real numbers.
 Let $\variables$ be a finite set of real-valued variables.
 A \emph{valuation} on $\variables$ is a function $\nu : \variables \to \Real$.
 We assume an arbitrary but fixed ordering on the variables and write $x_i$
 for the variable with order $i$. 
 This allows us to treat a valuation $\nu$ as a point $(\nu(x_1), \nu(x_2),
 \ldots, \nu(x_n)) \in \Real^{|\variables|}$. 
 Abusing notations slightly, we use a valuation on $\variables$ and a point in
 $\Real^{|\variables|}$ interchangeably. 
For a subset of variables $X \subseteq \variables$ and a valuation
$\nu' \in \variables$, we write $\nu[X{:=}\nu']$ for the valuation where
$\nu[X{:=}\nu'](x) = \nu'(x)$ if $x \in X$, and
$\nu[X{:=}\nu'](x) = \nu(x)$ otherwise.
The valuation $\zero \in \V$ is a special valuation such that
$\zero(x) = 0$ for all $x \in \variables$.

 We define a constraint over a set $\variables$ as a subset of $\Real^{|\variables|}$.
 We say that a constraint is \emph{rectangular} if it is defined as the conjunction
 of a finite set of constraints of the form 
 $x  \bowtie k,$  where $k \in \Int$, $x \in \variables$, and $\bowtie \in \{<,\leq, =,
 >, \geq\}$.    
 For a constraint $G$, we write $\sem{G}$ for the set of valuations in
 $\Real^{|\variables|}$ satisfying the constraint $G$.  
 We write $\top$ ( resp., $\bot$) for the special constraint that is true
 (resp., false) in all the valuations, i.e. $\sem{\top} = \Real^{|\variables|}$ 
(resp., $\sem{\bot} = \emptyset$). 
 We write $\rect(\variables)$ for the set of rectangular constraints over $\variables$ including
 $\top$ and $\bot$.  
\begin{definition}[Recursive Hybrid Automata]
  A \emph{recursive hybrid automaton} $\Hh = (\variables, (\Hh_1, \Hh_2, \ldots,
  \Hh_k))$ is a pair made of a set of variables $\variables$ and a collection of
  components $(\Hh_1, \Hh_2, \ldots, \Hh_k)$ where every component 
  $\Hh_i =  (N_i, \En_i, \Ex_i, B_i, Y_i, A_i, \trans_i, P_i, \inv_i,
  E_i, J_i, F_i)$ is such that:
  \begin{itemize}
  \item
    $N_i$ is a finite set of \emph{nodes} including a distinguished set $\En_i$ of
    \emph{entry nodes} and a set $\Ex_i$ of \emph{exit nodes} such that $\Ex_i$
    and $\En_i$ are disjoint sets; 
  \item
    $B_i$ is a finite set of \emph{boxes};
  \item
    $Y_i: B_i \to \set{1, 2, \ldots, k}$ is a mapping that assigns every box to
    a component. 
    (Call ports $\call(b)$ and return ports $\return(b)$ of a box $b \in B_i$,
    and call ports $\call_i$ and return ports $\return_i$ of a component $\Hh_i$
    are defined as before. 
    We set $Q_i = N_i \cup \call_i \cup \return_i$ and refer to this set as the
    set of locations of $\Hh_i$.) 
  \item
    $A_i$ is a finite set of \emph{actions}.
  \item
    $\trans_i : Q_i {{{\times}}} A_i \to Q_i$ is the transition function  with a
    condition that call ports and exit nodes do not have any outgoing transitions.
  \item
    $P_i: B_i \to 2^\variables$ is pass-by-value mapping that assigns every box
    the set of variables that are passed by value to the component
    mapped to the box; (The rest of the variables are assumed to be passed by
    reference.)
  \item
    $\inv_i : Q_i \to \rect(\variables)$ is the \emph{invariant condition};
  \item
    $E_i : Q_i {{\times}} A_i \to \rect(\variables)$ is the \emph{action enabledness function};  
  \item
    $J_i : A_i \to 2^{\variables}$ is the \emph{variable reset function}; and
  \item
    $F_i : Q_i \to \Nat^{|\variables|}$ is the \emph{flow function}
    characterizing the rate of each variable in each location.
  \end{itemize}
  We assume that the sets of boxes, nodes, locations, etc. are mutually
  disjoint across components and we write ($N, B, Y, Q, P, \trans$, etc.) to denote corresponding
  union over all components. 
%  If the flow function is $F_i : Q_i \to 1^{|\variables|}$ for all components $\Hh_i$ 
%  then the resulting RHA, has variables that are always growing with a rate 1. 
%  This RHA is referred to as RTA and corresponds to the RTA introduced in \cite{TW10}.
\end{definition}
%\subsection{Subclasses of RHA: Glitch-free, Hierarchical, Timed and Stopwatches}
% When we consider an RHA as an input of an algorithm, its size should be
% understood as the sum of the sizes of encodings of $Q$, $\variables$,
% $\inv$, $A$, $E$, and $\trans$. 
% Analogously as for RSMs, 
% we define special subclasses of RHAs: {\em 1-exit RHAs}, 
% for which each component is allowed to have just one exit, and {\em 1-box RHAs}, 
% that just consist of a single component with a single box inside of them.
We say that a recursive hybrid automaton is \emph{glitch-free} if for
every box either all variables are passed by value or none is passed by
value, i.e. for each $b \in B$ we have that either $P(b) = \variables$
or $P(b) = \emptyset$. 
%RHA without this restriction are called
% unrestricted.  
Any general recursive hybrid automaton with one variable is trivially
glitch-free.
We say that a RHA is \emph{hierarchical} if there exists an ordering over
components such that a component never invokes another component of higher
order or same order. 

We say that a variable $x \in \variables$ is a \emph{clock} (resp., a
stopwatch) if for every location $q \in Q$ we have that $F(q)(x) = 1$ (resp.,
$F(q)(x) \in \set{0, 1}$). 
A recursive timed automaton (RTA) is simply a recursive hybrid automata where all
variables $x \in \variables$ are clocks. 
Similarly, we define a recursive stopwatch automaton (RSA) as a recursive hybrid
automaton where all variables $x \in \variables$ are stopwatches. 
Since all of our results pertaining to recursive hybrid automata are shown in
the context of recursive stopwatch automata, we often confuse RHA with RSA. 

\subsection{Semantics}
A \emph{configuration} of an RHA $\Hh$ is a tuple $(\sseq{\kappa},
q, \val)$, where $\kappa \in (B {\times} \V)^*$ is 
sequence of pairs of boxes and variable valuations, $q \in Q$ is a
location and $\val \in \V$ is a variable valuation over $\variables$ such that
$\val \in \inv(q)$.
The sequence $\sseq{\kappa} \in (B {\times} \V)^*$ denotes the stack of
pending recursive calls and the valuation of all the variables at the
moment that call was made, and we refer to this sequence as the context of
the configuration.
Technically, it suffices to store the valuation of variables passed by
value, because other variables retain their value after returning from a
call to a box, but storing all of them simplifies the notation.
We denote the the empty context by $\sseq{\epsilon}$.
For any $t \in \Real$, we let $(\sseq{\kappa}, q, \val){+}t$
equal the configuration $(\sseq{\kappa}, q, \val{+} F(q) \cdot t)$.
Informally, the behaviour of an RHA is as follows.
In configuration $(\sseq{\kappa}, q, \val)$ time passes before an available
action is triggered, after which a discrete transition occurs.
Time passage is available only if the invariant condition $\inv(q)$
is satisfied while time elapses, and an action $a$ can be chosen after
time $t$ elapses only if it is enabled after time elapse, i.e.,\
if $\val{+} F(q) \cdot t \in E(q, a)$.
If the action $a$ is chosen then the successor state is
$(\sseq{\kappa}, q', \val')$ where $q' \in X(q,a)$ and
$\val' = (\val+t)[J(a) := \mathbf{0}]$.
Formally, the semantics of an RHA is given by an LTS which has both an
uncountably infinite number of states and transitions.

%%%%%%%%%%%%%%%%%%%%%%%%%%%%%%%%%%%%%%%%%%%%%%%%%%%%%%%%%%%%%%%%%%%%%%%%%
\begin{definition}[RHA semantics]
  Let $\Hh = (\variables, (\Hh_1, \Hh_2, \ldots, \Hh_k))$ be an RHA where
  each component is of the form
  $\Hh_i =  (N_i, \En_i, \Ex_i, B_i, Y_i, A_i, \trans_i, P_i, \inv_i,
  E_i, J_i, F_i)$.
  The semantics of $\Hh$ is a labelled transition system
  $\sem{\Hh} = (\nS, \nA, \ntrans)$ where:
  \begin{itemize}
  \item
    $\nS \subseteq (B{\times} \V)^* {\times} Q\, {\times}\, \V$, the set of
    states, is s.t. 
    $(\sseq{\kappa}, q, \val) {\in} \nS$ if $\val {\in} \inv(q)$.
  \item
    $\nA = \Rplus {{\times}} A$ is the set of \emph{timed actions}, where $\Rplus$ is the set of non-negative reals;
  \item
    $\ntrans : \nS {\times} \nA \to \nS$ is the transition function
    such that for $(\sseq{\kappa}, q, \val) \in \nS$ and
    $(t, a) \in \nA$, we have
    $(\sseq{\kappa'}, q', \val') = \ntrans((\sseq{\kappa}, q, \val),
    (t, a))$ if and only if the following condition holds:
    \begin{enumerate}
    \item
      if the location $q$ is a call port, i.e. $q = (b, en) \in
      \call$ then $t = 0$, the context
      $\sseq{\kappa'} = \sseq{\kappa, (b, \val)}$,
      $q' = en$, and $\val' = \val$.
    \item
      if the location $q$ is an exit node, i.e. $q = ex \in Ex$,
      $\sseq{\kappa} = \sseq{\kappa'', (b, \val'')}$,
      and let $(b, ex) \in \return(b)$, then
      $t = 0$; $\sseq{\kappa'} = \sseq{\kappa''}$;
      $q' {=} (b, ex)$; and $\val' {=} \val[P(b){:=}\val'']$.
    \item
      if location $q$ is any other kind of location, then  $\sseq{\kappa'} =
      \sseq{\kappa}$,  $q' \in \trans(q,a)$, and
 \begin{enumerate} 
      \item $\val {+} F(q)\cdot t' \in \inv(q)$ for all $t' \in [0, t]$;
      \item $\val {+} F(q)\cdot t \in E(q, a)$; 
      \item $\val' = (\val + F(q){\cdot} t)[J(a):= \zero]$.
      \end{enumerate}
    \end{enumerate}
  \end{itemize}
\end{definition}

\subsection{Reachability and Time-Bounded Reachability Game Problems}
For a subset $Q' \subseteq Q$ of states of RHA $\Hh$ we
define the set $\sem{Q'}_\Hh$ as the set $\set{(\sseq{\kappa}, q, \val) \in S_\Hh
  \::\: q \in Q'}$. 
We define the terminal configurations as
${\term_\Hh = \set{(\sseq{\varepsilon}, q, \val) \in S_\Hh \::\:   q \in \Ex}}$.
Given a recursive hybrid automaton $\Hh$, an initial node $q$ and
valuation $\val \in \V$, and a set of \emph{final locations} $F \subseteq
Q$, the {\em reachability problem} on $\Hh$ is to decide the existence of a run
in the LTS $\sem{\Hh}$ staring from the initial state $(\sseq{\varepsilon}, q,
\val)$ to some state in $\sem{F}_\Hh$.
As with RSMs, we also define \emph{termination problem} as reachability of one
of the exits with the empty context. 
Hence, given an RHA $\Hh$ and an initial node $q$ and a valuation $\val
\in \V$, the termination problem on $\Hh$ is to decide the existence of a run in the
LTS $\sem{\Hh}$ from initial state $(\sseq{\varepsilon}, q, \val)$ to a final
state in  $\term_\Hh$.

Given a run $r = \seq{s_0, (t_1, a_1), s_2, (t_2, a_2), \ldots, (s_n, t_n)}$ of an
RHA, its time duration $\Time(r)$ is defined as $\sum_{i=1}^{n}t_i$. 
Given a recursive hybrid automaton $\Hh$, an initial node $q$, a bound $T \in
\Nat$, and valuation $\val \in \V$, and a set of \emph{final locations} $F \subseteq
Q$, the {\em time-bounded reachability problem} on $\Hh$ is to decide the
existence of a run $r$ in the LTS $\sem{\Hh}$ staring from the initial state
$(\sseq{\varepsilon}, q, \val)$ to some state in $\sem{F}_\Hh$ such that
$\Time(r) \leq T$. 
Time-bounded termination problem is defined in an analogous manner. 

A partition $(Q_\Ach, Q_\Tor)$ of locations $Q$ of an RHA $\Hh$
gives rise to a recursive hybrid game arena 
$\Gamma = (\Hh, Q_\Ach, Q_\Tor)$.
Given an initial location $q$, a valuation $\nu \in V$ and a set of final
states $F$, the reachability game on $\Gamma$ is defined as the
reachability game on the game arena $(\sem{\Hh}, \sem{Q_\Ach}_\Tt,
\sem{Q_\Tor}_\Tt)$ with the initial state
$(\sseq{\varepsilon}, (q, \nu))$ and the set of final states $\sem{F}_\Tt$.
Also, termination game on $\Tt$ is defined as the reachability game on
the game arena $(\sem{\Tt}, \sem{Q_\Ach}_\Tt, \sem{Q_\Tor}_\Tt)$ with
the initial state $(\sseq{\varepsilon}, (q, \nu))$ and the set of final states
$\term_\Tt$.

We prove the following key theorem about reachability games on various
subclasses of recursive hybrid automata in Section \ref{sec:undec-rhg}.
\begin{theorem}
\label{th:main_game}
The reachability game problem is undecidable for: 
\begin{enumerate}
\item 
 Unrestricted RSA with 2 stopwatches,
\item 
 Glitchfree RSA with 3 stopwatches, 
\item 
 Unrestricted RTA with 3 clocks under bounded time, and 
\item 
 Glitchfree RSA with 4 stopwatches under bounded time.
\end{enumerate}
Moreover, all of these results hold even under hierarchical restriction. 
\end{theorem}

On a positive side, we observe that for glitch-free RSA with two stopwatches 
reachability games are decidable by exploiting the existence of finite bisimulation
for hybrid automata with 2 stopwatches. Details are given in Appendix~\ref{dec} 
\begin{theorem}
 \label{thm:dec_game}
 The reachability games are decidable for glitch-free RSA
 with atmost two stopwatches. 
\end{theorem}

We study the above mentioned problems when studied for a single player 
game. These problems have been detailed in Section \ref{sec:undec-rha}.

\begin{theorem}
\label{th:main}
The reachability problem is undecidable for 
\begin{enumerate}
\item Unrestricted RHA with 2 stopwatches,
\item Glitchfree RHA with 3 stopwatches, 
\item Unrestricted RTA with 5 clocks under bounded time, and 
\item Glitchfree RHA with 14 stopwatches under bounded time.
\end{enumerate}
Moreover, all of these results hold even under hierarchical restriction. 
\end{theorem}
On a positive side, we observe the following decidability results. 
\begin{theorem}
 \label{thm:dec}
 The reachability and the termination problems are decidable for 
 \begin{enumerate}
\item  Glitch-free RHA
 with atmost two stopwatches
\item Bounded context RHA under bounded time, where variables are always passed-by-reference.
\end{enumerate}
\end{theorem}
The result for Glitch-free RHA with two stopwatches follows from the decidability of two stopwatch hybrid automata.

\section{Undecidability Results with one player}
\label{sec:undec-rha}
 In this section, we provide a proof sketch of our undecidability results by
reducing the halting problem for two counter machines  to the reachability
problem in an RHA/RTA. 
A \emph{two-counter machine} $M$ is a tuple $(L, C)$ where ${L = \set{\ell_0,
    \ell_1, \ldots, \ell_n}}$ is the set of instructions including a
distinguished terminal instruction $\ell_n$ called HALT, and the set 
${C = \set{c_1, c_2}}$ of two \emph{counters}.  
The instructions $L$ are of the type:
\begin{enumerate}
\item (increment $c$) $\ell_i : c := c+1$;  goto  $\ell_k$,
\item (decrement $c$) $\ell_i : c := c-1$;  goto  $\ell_k$,
\item (zero-check $c$) $\ell_i$ : if $(c >0)$ then goto $\ell_k$
  else goto $\ell_m$,
\end{enumerate}
where $c \in C$, $\ell_i, \ell_k, \ell_m \in L$.
A configuration of a two-counter machine is a tuple $(l, c, d)$ where
$l \in L$ is an instruction, and $c, d \in \Nat$ is the value of counters $c_1$
and $c_2$, resp. 
A run of a two-counter machine is a (finite or infinite) sequence of
configurations $\seq{k_0, k_1, \ldots}$ where $k_0 = (\ell_0, 0, 0)$ and the
relation between subsequent configurations is governed by transitions between
respective instructions. 
The \emph{halting problem} for a two-counter machine asks whether 
its unique run ends at the terminal instruction $\ell_n$.
It is well known~(\cite{Min67}) that the halting problem for
two-counter machines is undecidable.

In order to prove four results of Theorem~\ref{th:main}, we construct a
recursive (timed/hybrid) automaton whose main components simulate various
instructions.  
In these constructions the reachability of the exit node of each component
corresponding to an instruction is linked to a faithful simulation of various
increment, decrement and zero check  instructions of the machine by choosing
appropriate delays to adjust the clocks/variables, to reflect changes  in
counter values. 
We specify a main component for each instruction of the two counter machine. 
The entry node and exit node of a main component corresponding to an 
instruction $\ell_i:  c := c +1$;  goto  $\ell_k$ are respectively $\ell_i$ and
$\ell_k$.   
Similarly, a main component corresponding to a zero check instruction 
$l_i$: if $(c >0)$ then goto $\ell_k$
else goto $\ell_m$, has a unique entry node  $\ell_i$, and two exit nodes
corresponding to $\ell_k$ and  $\ell_m$ respectively. 
We get the complete RHA for the two-counter machines when we connect these main
components in the same sequence as the corresponding machine. 
The halting  problem of the two counter machine now reduces to the reachability
(or termination) of an exit (HALT) node $\ell_n$ in some component.

For the correctness proofs, we represent runs in the RSA
using three different forms of transitions
  $s \stackrel[t]{g,J}{\longrightarrow} s'$, $s  \rightsquigarrow s'$ and $s \stackrel[M(V)]{*}{\longrightarrow} s'$ 
  defined in the following way:
  \begin{enumerate}
  \item The transitions of the form $s \stackrel[t]{g,J}{\longrightarrow} s'$, where $s=(\langle \kappa \rangle, n, \nu)$, 
  $s'=(\langle \kappa \rangle, n', \nu')$ are configurations of the RHA, $g$ is a constraint or guard 
  on variables that enables the transition,
   $J$ is a set of variables, and $t$ is a real number, holds if there is a transition in the RHA from vertex $n$ to $n'$ with guard $g$ and reset 
  set  $J$. Also, $\nu'=\nu+rt[J:=0]$, where $r$ is the rate vector 
  of state $s$.
  \item The transitions of the form $s \rightsquigarrow s'$ where $s=(\langle \kappa \rangle,n,\nu)$,
  $s'=(\langle \kappa' \rangle, n', \nu')$ correspond to the following cases:
  \begin{itemize}
  \item transitions from a call port to an entry node. That is, $n=(b,en)$ for some box
  $b \in B$ and $\kappa' =\langle \kappa, (b, \nu)\rangle$ and $n'=en \in \En$ while $\nu'=\nu$. 
  \item transitions from an exit node to a return port which restores values of the variables passed by value, that is, $\langle \kappa \rangle=\langle \kappa'',(b, \nu'')\rangle$, $n=ex \in \Ex$ and $n'=(b,ex) \in Ret(b)$ and $\kappa'=\kappa''$, while $\nu'=\nu[P(b):=\nu'']$.
  \end{itemize}
  \item The transitions of the form $s \stackrel[M(V)]{t}{\longrightarrow} s'$, called summary edges, where 
  $s=(\langle \kappa \rangle, n, \nu)$, $s'=(\langle \kappa \rangle, n', \nu')$ are such that $n=(b,en)$ and $n'=(b,ex)$ are call and return ports, respectively, of a box $b$ mapped
to $M$ which passes by value to $M$, the variables in $V$. $t$ is the time elapsed between the occurences of $(b,en)$ and $(b,ex)$. In other words, $t$ is the time elapsed in the component $M$.
    \end{enumerate}
  A configuration $(\langle \kappa \rangle, n, \nu)$ is also written as $(\langle \kappa \rangle, n, (\nu(x),\nu(y)))$.  

\subsection{Unrestricted RHA with 2 stopwatches}
\label{app:2sw}
For all the four  undecidability results, we construct a recursive automaton (timed/hybrid)
as per the case, whose main components are the modules for the instructions and 
the counters are encoded in the variables of the automaton.  
In these reductions,  the reachability of the exit node of each component
corresponding to an instruction is linked to a faithful simulation of various
increment, decrement and zero check  instructions of the machine by choosing
appropriate delays to adjust the clocks/variables, to reflect changes  in
counter values. 
We specify a main component for each type instruction of the two counter machine, for example $\Hh_{inc}$ for increment. 
The entry node and exit node of a main component $\Hh_{inc}$ corresponding to an 
instruction [$\ell_i:  c := c +1$;  goto  $\ell_k$] 
are respectively $\ell_i$ and $\ell_k$.  
 Similarly, a main component corresponding to a zero check instruction 
[$l_i$: if $(c >0)$ then goto $\ell_k$]
  else goto $\ell_m$, has a unique entry node  $\ell_i$, and two exit nodes corresponding to $\ell_k$ and 
  $\ell_m$ respectively. 
The various main components corresponding to the various instructions, when connected appropriately, gives the higher level component $\Hh_{M}$ and this completes the RHA $\Hh$. The entry node of $\Hh_{M}$ is the entry node of the main component for the first instruction of $M$ and the exit node is $Halt$. Suppose each main component for each type of instruction correctly simulates the instruction by accurately updating the counters encoded in the variables of $\Hh$. Then, the unique run in $M$ corresponds to an unique run in $\Hh_{M}$. The halting  problem of the two counter machine now boils down to the reachability of an exit node $Halt$ in $\Hh_{M}$.

\begin{lemma}
The reachability problem is undecidable for recursive hybrid automata 
  with at least two stopwatches.
\end{lemma}
\begin{proof}
We prove that the reachability and termination problems are undecidable for 2 stopwatch unrestricted RHA.  
In order to obtain the undecidability result, we 
use a reduction from the halting problem 
for two counter machines. 
Our reduction uses a RHA with stopwatches $x,y$.

We specify a main component for each instruction of the two counter machine. On entry into 
a main component for increment/decrement/zero check, 
we have $x=\frac{1}{2^c3^d}$, $y=0$ or $x=0, y=\frac{1}{2^c3^d}$, where $c,d$ are the current values of the counters. 
Given a two counter machine, we build a 2 stopwatch RHA 
whose building blocks are the main components for the instructions.   
The purpose of the components is to simulate faithfully the counter machine 
by choosing appropriate delays  to adjust the variables to reflect changes 
in counter values. 
On entering the entry node $en$ of a main component corresponding to an instruction $l_i$,
we have the configuration $(\langle \epsilon \rangle, en, (\frac{1}{2^c3^d},0))$ or 
$(\langle \epsilon \rangle, en, (0,\frac{1}{2^c3^d}))$
of the two stopwatch RHA.  

\begin{figure}[t]
  \centering
  {\small
\begin{tikzpicture}[node distance=4cm] 
  
  \draw(0, 0) rectangle (3.5,1.5);
  \draw (3.3, 1.7) node {$\db$};

  \node[loc](n1) at (0, 0.8) {$en$};  
  \node[rate] () [below of = n1,node distance = 5mm] {$\figrate{y}$};  
  
  \node[loc](x1) at (3.5, 0.8) {$ex$};
   
  \node[boxloc](b1) at (1.8, 0.8) {$\begin{array}{c} B_1{:}C_{\db}\end{array}$};
  \node () [below of = b1, node distance = 5mm] {$(x,y)$};
  \node[port](b1n1) [left of =b1,node distance = 7.2mm] {}; 
   \node[port](b1x1) [right of =b1,node distance = 7.2mm] {};
   
   \draw[trans] (n1)--(b1n1)  node [midway, above]{$y {>} 0$};
   \draw[trans] (b1x1)--(x1)  node [midway, above]{$\{x\}$};
   %%%%%%%%%%%%%%%%%%%%%%%%%%%%%%%%%%%%%%%%%%%%%%%%%%%%%%%%%%%%%%%%%%%%%%%%%%%%%%% 
  %%%%%%%%%%%%%%%%%%%%%%%%%%%%%% C_\db  %%%%%%%%%%%%%%%%%%%%%%%%%%%%%%%%%%%%%%%%%%%% 
   
  \draw(4.2, 0) rectangle (9,1.5);
  \draw (8.7, 1.7) node {$C_{\db}$};

  \node[loc](n2) at (4.2, 0.8) {$en_1$};  
  \node[loc](x2) at (9, 0.8) {$ex_1$};
  
  \node[boxloc](b2) at (5.5, 0.8) {$\begin{array}{c}B_2{:}M_1 \end{array}$};
  \node () [below of = b2, node distance = 5mm] {$(x)$};
  \node[port](b2n1) [left of =b2,node distance = 6mm] {}; 
   \node[port](b2x1) [right of =b2,node distance = 6mm] {};

  \node[boxloc](b3) at (7.2, 0.8) {$\begin{array}{c}B_3{:}M_1 \end{array}$};
  \node () [below of = b3, node distance = 5mm] {$(x)$};
  \node[port](b3n1) [left of =b3,node distance = 6mm] {}; 
   \node[port](b3x1) [right of =b3,node distance = 6mm] {};
  
  \draw[trans] (n2)--(b2n1) {};
  \draw[trans] (b2x1)--(b3n1) {};
  \draw[trans] (b3x1)--(x2)  node [midway, above]{$y {=} 2$};   
  %%%%%%%%%%%%%%%%%%%%%%%%%%%%%%%%%%%%%%%%%%%%%%%%%%%%%%%%%%%%%%%%%%%%%%%%%%%%%%% 
  %%%%%%%%%%%%%%%%%%%%%%%%%%%%%%%%%% M_1 %%%%%%%%%%%%%%%%%%%%%%%%%%%%%%%%%%%%%%%%%%% 

  \draw(9.7, 0) rectangle (11.3,1.5);
  \draw (10.7, 1.7) node {$M_1$};

  \node[loc](n3) at (9.7, 0.8) {$en_2$};  
  \node[loc](x3) at (11.3, 0.8) {$ex_2$};
 \node[rate] () [below of = n3,node distance = 5mm] {$\figrate{x,y}$};   `

  \draw[trans] (n3)--(x3) node [midway, above]{$x{=}1$};
   
  %%%%%%%%%%%%%%%%%%%%%%%%%%%%%%%%%%%%%%%%%%%%%%%%%%%%%%%%%%%%%%%%%%%%%%%%%%%%%%% 
  %%%%%%%%%%%%%%%%%%%%%%%%%%%%%%%%%%%  HF %%%%%%%%%%%%%%%%%%%%%%%%%%%%%%%%%%%%%%%%%%%% 

    \draw(0, -2) rectangle (3.5,-0.5);
  \draw (3.3, -0.3) node {$\hf$};

  \node[loc](n1) at (0, -1.2) {$en$};  
  \node[rate] () [below of = n1,node distance = 5mm] {$\figrate{y}$};  
  
  \node[loc](x1) at (3.5, -1.2) {$ex$};
   
  \node[boxloc](b1) at (1.8, -1.2) {$\begin{array}{c} B_4{:}C_{\hf}\end{array}$};
  \node () [below of = b1, node distance = 5mm] {$(x,y)$};
  \node[port](b1n1) [left of =b1,node distance = 7.2mm] {}; 
   \node[port](b1x1) [right of =b1,node distance = 7.2mm] {};
   
   \draw[trans] (n1)--(b1n1)  node [midway, above]{$y {>} 0$};
   \draw[trans] (b1x1)--(x1)  node [midway, above]{$\{x\}$};
   %%%%%%%%%%%%%%%%%%%%%%%%%%%%%%%%%%%%%%%%%%%%%%%%%%%%%%%%%%%%%%%%%%%%%%%%%%%%%%% 
  %%%%%%%%%%%%%%%%%%%%%%%%%%%%%% C_\hf  %%%%%%%%%%%%%%%%%%%%%%%%%%%%%%%%%%%%%%%%%%%% 
   
  \draw(4.2, -2) rectangle (9,-0.5);
  \draw (8.7, -0.3) node {$C_{\hf}$};

  \node[loc](n2) at (4.2, -1.2) {$en_3$};  
  \node[loc](x2) at (9, -1.2) {$ex_3$};
  
  \node[boxloc](b2) at (5.5, -1.2) {$\begin{array}{c}B_5{:}M_2 \end{array}$};
  \node () [below of = b2, node distance = 5mm] {$(y)$};
  \node[port](b2n1) [left of =b2,node distance = 6mm] {}; 
   \node[port](b2x1) [right of =b2,node distance = 6mm] {};

  \node[boxloc](b3) at (7.2, -1.2) {$\begin{array}{c}B_6{:}M_2 \end{array}$};
  \node () [below of = b3, node distance = 5mm] {$(y)$};
  \node[port](b3n1) [left of =b3,node distance = 6mm] {}; 
   \node[port](b3x1) [right of =b3,node distance = 6mm] {};
  
  \draw[trans] (n2)--(b2n1) {};
  \draw[trans] (b2x1)--(b3n1) {};
  \draw[trans] (b3x1)--(x2)  node [midway, above]{$x {=} 2$};   
  %%%%%%%%%%%%%%%%%%%%%%%%%%%%%%%%%%%%%%%%%%%%%%%%%%%%%%%%%%%%%%%%%%%%%%%%%%%%%%% 
  %%%%%%%%%%%%%%%%%%%%%%%%%%%%%%%%%% M_2 %%%%%%%%%%%%%%%%%%%%%%%%%%%%%%%%%%%%%%%%%%% 

  \draw(9.7, -2) rectangle (11.3,-0.5);
  \draw (10.7, -0.3) node {$M_2$};

  \node[loc](n3) at (9.7, -1.2) {$en_4$};  
  \node[loc](x3) at (11.3, -1.2) {$ex_4$};
 \node[rate] () [below of = n3,node distance = 5mm] {$\figrate{x,y}$};   `

  \draw[trans] (n3)--(x3) node [midway, above]{$x{=}1$};
   
  %%%%%%%%%%%%%%%%%%%%%%%%%%%%%%%%%%%%%%%%%%%%%%%%%%%%%%%%%%%%%%%%%%%%%%%%%%%%%%% 
  %%%%%%%%%%%%%%%%%%%%%%%%%%%%%%%%%%%%   zero check %%%%%%%%%%%%%%%%%%%%%%%%%%%%%%%%%%%%%%%%%%% 
  
  \draw(0, -4) rectangle (3.5,-2.5);
  \draw (3, -2.3) node {$zero~check$};

    \node[loc](n4) at (0, -3.3) {$en$}; 
  
    \node[loc](x4) at (3.5, -2.9) {$ex$};
    \node[loc](x41) at (3.5, -3.6) {$ex'$};
  
  \node[boxloc](b4) at (1.8, -3.3) {$\begin{array}{c} B_1{:}Po2  \end{array}$};
  \node[port](b4n1) [left of=b4,node distance = 7mm] {}; 
  \node[port](b4x1) at (2.5,-3.15) {}; %distance is add 0.7
  \node[port](b4x2) at (2.5,-3.45) {};
  \node () [below of = b4, node distance = 5mm] {$(x,y)$};

   \draw[trans] (n4)--(b4n1)  node [midway, above]{$y {=} 0$};
   \draw[trans] (b4x1)--(x4)  node [midway, above]{$y{=}0$};
  \draw[trans] (b4x2)--(x41)  node [midway, below]{$y{=}0$};

    \draw(4.2, -4.2) rectangle (11.3,-2.5);
  \draw (11, -2.3) node {$Po2$};

    \node[loc](n5) at (4.2, -3.3) {$en_5$}; 
  
    \node[loc](x5) at (11.3, -2.8) {$ex_5$};
    \node[loc](x51) at (11.3, -3.8) {$ex'_5$};
  
    \node[boxloc](b7) at (6, -3.3) {$\begin{array}{c} B_7{:}\db\end{array}$};
    \node[port](b7n1) [left of =b7,node distance = 6.5mm] {}; 
    \node[port](b7x1) [right of =b7,node distance = 6.5mm] {};

%     \tikzstyle{portN}=[draw,circle,fill=white, minimum size=.5em,inner sep=0em,linecolor=blue]
    \tikzstyle{boxlocN}=[draw, very thick, rectangle, minimum size=2em, inner sep=0.3em,color=blue,fill=white]

    \node[boxloc](b8) at (9, -3.3) {$\begin{array}{c} B_8{:}\db'\end{array}$};  
    \node[port](b8n1) [left of =b8,node distance = 6.6mm] {}; 
    \node[port](b8x1) [right of =b8,node distance = 6.6mm] {};
    
    \draw[trans] (n5) -- node[midway,above]{$y{=}0$}(b7n1);
    \draw[trans] (b7x1) --  node[near end,above]{$y{<}2$}(b8n1);
%     \draw[trans] (b7x1) -- node[above]{$y{=}2$}(7.5,-2.8) -- (x5);
%     \draw[trans] (b7x1) -- node[below]{$y{>}2$} (7.5,-3.8) --(x51);
    \draw[trans] (b7x1) -- node[above,rotate=30]{$y{=}2$}(7.5,-2.8) -- (x5);
    \draw[trans] (b7x1) -- node{$y{>}2$} (7.5,-3.8) --(x51);
    
    \tikzstyle{transN}=[-latex, rounded corners,color=blue]
    
     \draw[trans] (b8x1) -- node[rotate=25,above]{$x{=}2$}(x5);
     \draw[trans] (b8x1) -- node[rotate=-25,above]{$x{>}2$}(x51);
    \draw[trans] (b8x1) -- (x5);
%     \node () at (9.5,-3) {$x{=}2$};
%     \node () at (9.5,-3.5) {$x{>}2$};
    \draw[trans] (b8x1) -- (x51);
    \draw[trans] (b8x1) -- (10.3,-4) -- (7,-4)-- node[midway,above]{$x{<}2$}(5.2,-4) --(b7n1);
\end{tikzpicture}
}
  \caption{RHA : 2 stopwatch : Decrement $c$ is $\db$, increment $c$ is $\hf$ and zero check $d$}
\label{fig_undec_2swG}
\end{figure}
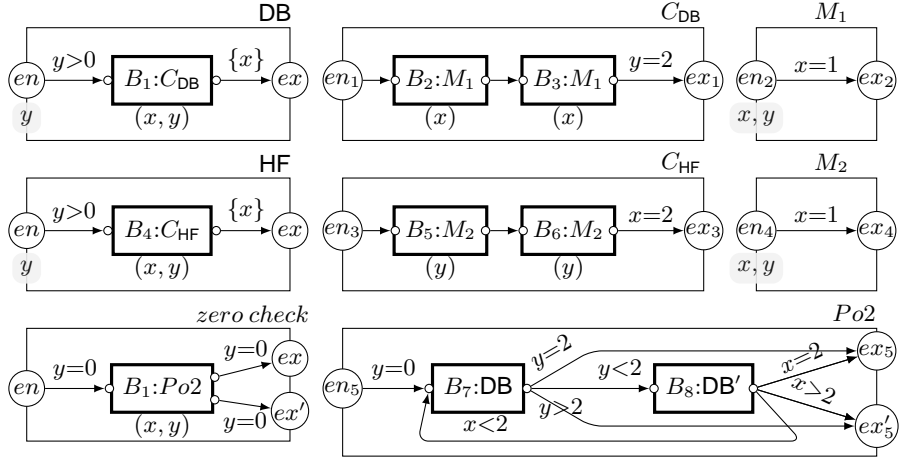

We shall now present the components for increment/decrement and zero check instructions.
In all the components, the ticking variables are written below respective locations in grey, while 
the variables passed by value are written below the boxes.

\noindent{\bf{Simulate decrement instruction}}: Lets consider the decrement instruction  
$\ell_i$: $c=c-1$; goto $\ell_k$. Figure \ref{fig_undec_2swG} gives the component $DB$ 
which decrements counter $c$, by doubling $\frac{1}{2^c3^d}$.  Assume that 
$x=\frac{1}{2^c3^d}$ and $y=0$ on entering $DB$. Lets denote by $x_{old}$ the value
$\frac{1}{2^c3^d}$.   
A non-deterministic amount of time $t$ is spent at the entry node $en$ of  
$DB$. This makes $x=x_{old}$ and $y=t$, at the call port 
of $B_1:C_{DB}$. 
Both $x,y$ are passed by value to $C_{DB}$. 

At the entry node of $C_{DB}$, the rates of $x,y$ are zero. At sometime, the call port of $B_2:M_1$ 
is reached with   $x=x_{old}$ and $y=t$. $M_1$ is called by passing $x$ by value. At the entry node of $M_1$, 
a time $1-x_{old}$ is spent, obtaining $x=1,y=t+1-x_{old}$. 
We return from the exit node of $M_1$ to the return port of $B_2:M_1$, with $x=x_{old}, y=1+t-x_{old}$.
The rates of $x,y$ are both zero here. After some time,  we are at the call port of $B_3:M_1$. 
Here again, $M_1$ is called by passing $x$ by value. Going through $M_1$ again gives us 
$x=1, y=2-2x_{old}+t$. At the return port of $B_3$, we thus have $x=x_{old}, y=2-2x_{old}+t$. Again since the rates of $x,y$ are both zero at the return port of $B_3$, to get to the exit node of $C_{DB}$, $y$ must be exactly equal to 2. That is, 
$2x_{old}=t$. In that case, when we get back to the return port of $B_1:C_{DB}$, we have 
$x=x_{old}, y=t$, with the guarantee that $t=2x_{old}$. The rates of $x,y$ are both zero here, so we 
 get to the exit node $ex$
of $DB$ resetting $x$. Thus, when we reach $ex$, we have $x=0, y=\frac{1}{2^{c-1}3^d}$.

\noindent{\bf{Simulate increment instruction}}:
The instruction [$\ell_i$: $c=c+1$; goto $\ell_k$] is handled by the component $HF$ in Figure \ref
{fig_undec_2swG}. 
The main component $HF$, when  entered with $x=\frac{1}{2^c3^d}, y=0$
will halve the value of $x$, and return $x=0, y=\frac{1}{2^{c+1}3^d}$. The working of the component 
$HF$ can be explained in a  similar way as that of $DB$.

\noindent{\bf{Zero check instruction}}: 
The component $zerocheck$ simulating [$\ell_i$ : if $(d>0)$ then goto $\ell_k$ else goto $l_m$] can be found in 
Figure \ref{fig_undec_2swG}. 
Assume we are at $en$ with $x=\frac{1}{2^c3^d}=x_{old}$ and $y=0$.
The rates of both $x,y$ are zero, so we reach the callport of $B_1{:}Po2$
with the same values of $x,y$. $Po2$ is called by passing both $x,y$ by value. 
No time is spent at the entry node $en_5$ of 
  $Po2$, so with $y=0$, we reach the call port of $B_7{:}DB$. 
  Recall that 
 $DB$ is the component that doubles the value of $x$ and stores it in $y$;
 $DB$ is called by passing both $x,y$ by reference; 
 when we return to the return port of $B_7$, we have $x=0, y=2x_{old}$.
 If $y=2$, then we go straightaway to the exit node $ex_5$ 
 of $Po2$. If $x <2$, then we goto the callport of $B_8{:}DB'$ from the return port of $B_7$.
 The component $DB'$ is similar to $DB$, with the roles of $x,y$ reversed as compared to $DB$, and 
 entry to $DB'$ happens with $x=0,y=2x_{old}$. At the exit node of $DB'$,  
we obtain $y=0,x=4x_{old}$. Now, if $x=2$, then we goto the 
exit node $ex_5$ 
 of $Po2$. If $x <2$, then we goto the callport of $B_7{:}DB$.
 In this way, we alternate between $DB,DB'$ until 
 we have multiplied $x_{old}$ by some number $k$ such that $k.x_{old}$ is exactly 2. 
 If we obtain $k.x_{old}=2$ at the return port of $B_7$, then we have $y=2$ and $x=0$, while 
 if we obtain  $k.x_{old}=2$ at the return port of $B_8$, then we have $x=2$ and $y=0$.
  If this happens, then $d=0$. If $d >0$, then we will never obtain 
  $k.x_{old}$ as 2. In this case, 
   we 
 go to the exit node $ex'_5$ when $x$ (or $y$) exceeds 2. If we reach the exit node $ex_5$ of $Po2$, 
 then we goto the exit node $ex$ of the zerocheck component, and 
 if we reach the exit node $ex'_5$ of $Po2$, then we goto the exit node $ex'$ of the zerocheck component.
 
  The following propositions show the correctness of the increment, decrement and zero check components.
%  Proposition \ref{2sw-db} proves the correctness of the component $DB$. 
  For the correctness proofs, we represent runs in the RHA using three different forms of transitions
  $s \stackrel[t]{g,J}{\longrightarrow} s'$, $s  \rightsquigarrow s'$ and $s \stackrel[M(V)]{*}{\longrightarrow} s'$ 
  defined in the following way:
  \begin{enumerate}
  \item The transitions of the form $s \stackrel[t]{g,J}{\longrightarrow} s'$, where $s=(\langle \kappa \rangle, n, \nu)$, 
  $s'=(\langle \kappa \rangle, n', \nu')$ are configurations of the RHA, $g$ is a constraint or guard 
  on variables that enables the transition,
   $J$ is a set of variables, and $t$ is a real number, holds if there is a transition in the RHA from vertex $n$ to $n'$ with guard $g$ and reset 
  set  $J$. Also, $\nu'=\nu+rt[J:=0]$, where $r$ is the rate vector 
  of state $s$.
  \item The transitions of the form $s \rightsquigarrow s'$ where $s=(\langle \kappa \rangle,n,\nu)$,
  $s'=(\langle \kappa' \rangle, n', \nu')$ correspond to the following cases:
  \begin{itemize}
  \item transitions from a call port to an entry node. That is, $n=(b,en)$ for some box
  $b \in B$ and $\kappa' =\langle \kappa, (b, \nu)\rangle$ and $n'=en \in \En$ while $\nu'=\nu$. 
  \item transitions from an exit node to a return port which restores values of the variables passed by value, that is, $\langle \kappa \rangle=\langle \kappa'',(b, \nu'')\rangle$, $n=ex \in \Ex$ and $n'=(b,ex) \in Ret(b)$ and $\kappa'=\kappa''$, while $\nu'=\nu[P(b):=\nu'']$.
  \end{itemize}
  \item The transitions of the form $s \stackrel[M(V)]{*}{\longrightarrow} s'$, called summary edges, where 
  $s=(\langle \kappa \rangle, n, \nu)$, $s'=(\langle \kappa' \rangle, n', \nu')$ are such that $n=(b,en)$ and $n'=(b,ex)$ are call and return ports, respectively, of a box $b$ mapped
to $M$ which passes by value to $M$, the variables in $V$.  
    \end{enumerate}
  A configuration $(\langle \kappa \rangle, n, \nu)$ is also written as $(\langle \kappa \rangle, n, (\nu(x),\nu(y)))$.  
  \begin{proposition}
  \label{2sw-db}
  For any context $\kappa$, any box $b \in B$, and $x \in [0,1]$, we have that 
   $(\langle \kappa \rangle, (b,en), (x,0)) \stackrel[DB]{*}{\longrightarrow} (\langle \kappa \rangle, (b,ex), (0,2x))$
  \end{proposition}
  \begin{proof}
  Component $DB$ uses components $C_{DB}$ and $M_1$. 
  The following is a unique run starting from  
  $(\langle \kappa \rangle, (b,en), (x,0))$ terminating in 
  $(\langle \kappa \rangle, (b,ex), (2x,0))$. 
  
  $(\langle \kappa \rangle, (b,en), (x_0,0)) \rightsquigarrow (\langle \kappa,b \rangle, en, (x_0,0))$\\\\
  $\stackrel[t]{y>0,\emptyset}{\longrightarrow}$
  $(\langle \kappa,b \rangle, (B_1,en_1), (x_0,t)) \rightsquigarrow (\langle \kappa,b,(B_1,(x_0,t)) \rangle, en_1, (x_0,t))$\\\\
  $\stackrel[any]{true,\emptyset}{\longrightarrow} (\langle \kappa,b,(B_1,(x_0,t)) \rangle, (B_2,en_2), (x_0,t))$
  $\rightsquigarrow (\langle \kappa,b,(B_1,(x_0,t)),(B_2,x_0) \rangle, en_2, (x_0,t))$\\\\
  $\stackrel[1-x_0]{x=1,\emptyset}{\longrightarrow} (\langle \kappa,b,(B_1,(x_0,t)),(B_2,x_0) \rangle, ex_2, (1,1-x_0+t))$\\\\ 
  $\rightsquigarrow (\langle \kappa,b,(B_1,(x_0,t))\rangle, (B_2,ex_2), (x_0,1-x_0+t))$\\\\
  $\stackrel[any]{true,\emptyset}{\longrightarrow}(\langle \kappa,b,(B_1,(x_0,t))\rangle, (B_3,en_2), (x_0,1-x_0+t))
  \\\\\rightsquigarrow 
  (\langle \kappa,b,(B_1,(x_0,t)),(B_3,x_0)\rangle, en_2, (x_0,1-x_0+t))$\\\\
  $\stackrel[1-x_0]{x=1,\emptyset}{\longrightarrow} (\langle \kappa,b,(B_1,(x_0,t)),(B_3,x_0) \rangle, ex_2, (1,2-2x_0+t)) 
  \\\\ \rightsquigarrow 
  (\langle \kappa,b,(B_1,(x_0,t))\rangle, (B_3,ex_2), (x_0,2-2x_0+t))$\\\\
  $\stackrel[any]{y=2,\emptyset}{\longrightarrow} 
  (\langle \kappa,b, (B_1,(x_0,t))\rangle, ex_1, (x_0,2))$ ($2-2x_0+t=2 \leftrightarrow t=2x_0$)\\\\
  $\rightsquigarrow 
  (\langle \kappa,b\rangle, (B_1,ex_1), (x_0,2x_0))$\\\\
 % \rightsquigarrow (\langle \kappa\rangle, (b,x_1), (x_0,t))$\\\\ 
  $\stackrel[any]{true,\{x\}}{\longrightarrow}
  (\langle \kappa,b\rangle,ex, (0,2x_0)) \rightsquigarrow (\langle \kappa\rangle,(b,ex), (0,2x_0)).$
 The transitions above easily follow from the descriptions 
 given in the decrement section. 
      \qed \end{proof}
  
    Proposition \ref{2sw-inc} proves the correctness of the component $HF$.
 \begin{proposition}
  \label{2sw-inc}
  For any context $\kappa$, any box $b \in B$, and $x \in [0,1]$, we have that 
   $(\langle \kappa \rangle, (b,en), (x,0)) \stackrel[HF]{*}{\longrightarrow} (\langle \kappa \rangle, (b,ex), (0,\frac{x}{2}))$
  \end{proposition}
\begin{proof}
Similar to Proposition \ref{2sw-db}.
\end{proof}
  Proposition \ref{2sw-zc} proves the correctness of the component
 $Po2$ that checks if $x$ is a power of 2:
 \begin{proposition}
  \label{2sw-zc}
  For any context $\kappa$, any box $b \in B$, and $x \in [0,1]$, we have that 
  starting from $(\langle \kappa \rangle, (b,en), (x,0))$,  $zero check$ terminates 
  at $(\langle \kappa \rangle, (b,ex), (x,0))$
iff $x=\frac{1}{2^i}$, $i \in \mathbb{N}$. Otherwise, 
it terminates in $(\langle \kappa \rangle, (b,ex'), (x,0))$.
 \end{proposition}
  \begin{proof}
  The proof of this follows from the correctness of the component $DB$ shown above.  
  Indeed, if $DB$ doubles the variable, clearly, $\frac{1}{2^c3^d}$ will become 2 
  eventually after $c+1$ invocations of $DB$ iff $d=0$. 
    \qed \end{proof}
    
   Note that the components for incrementing, decrementing and zero check for counter $d$ 
    can be obtained in a manner similar to $DB,HF$.  The only difference 
    is that we have to multiply and divide by 3; these gadgets can be obtained straightforwardly 
    by adapting $DB,HF$ appropriately.
    
      We now show that the two counter machine halts iff 
      a vertex $Halt$ corresponding to the halting instruction is reached 
      in the RHA. Clearly, all the main components discussed above ensure that all instructions are simulated correctly.   
  Assume the two counter machine halts. Then clearly, after going through all the main components corresponding to 
  relevant instructions, we reach the component that leads to the Halt vertex. 
  The $DB,HF,zerocheck$ subcomponents again ensure that simulation is done correctly 
  to reach the vertex Halt.    
   Conversely, assume that the two counter machine does not halt. Then 
   there are two possibilities: (1) the RHA proceeds component by component, forever, simulating all 
   instructions faithfully, or (2) the RHA is unable to take a transition, due to an error in the simulation of instructions. 
   In either case, the vertex Halt is never reached.  
     \end{proof}

\subsection{GlitchFree RHA with 3 stopwatches}
\begin{lemma}
\label{app:3sw}
  The reachability problem is undecidable for recursive hybrid automata with at least three stopwatches.
\end{lemma}
\begin{proof}
The proof of this Lemma is a straightforward adaptation of the techniques used in Lemma \ref{app:2sw}. 
The main difference here is that, at all times, we have to pass all variables either by value, or by reference.
This necessitates the need for an extra variable.
In particular, we always pass all variables only by value. Thus, our result 
holds for the case of ``pass  by value'' RHAs with 3 stopwatches.

We specify a main component for each instruction of the two counter machine. On entry into 
a main component for increment/decrement/zero check, 
we have $x=\frac{1}{2^c3^d}$, $y=z=0$
% or $x=0, y=\frac{1}{2^c3^d}$,
 where $c,d$ are the current values of the counters. 
Given a two counter machine, we build a 3 stopwatch RHA 
whose building blocks are the main components for the instructions.   
The purpose of the components is to simulate faithfully the counter machine 
by choosing appropriate delays  to adjust the variables to reflect changes 
in counter values. 
On entering the entry node $en$ of a main component corresponding to an instruction $\ell_i$,
we have the configuration $(\langle \epsilon \rangle, en, (\frac{1}{2^c3^d},0,0))$ 
%or $(\langle \epsilon \rangle, en, (0,\frac{1}{2^c3^d}))$
of the three stopwatch RHA.  
We shall now present the components for increment/decrement and zero check instructions.
In all the components, the ticking variables are written below respective locations in grey.

\noindent{\bf{Simulate decrement instruction}}: Lets consider the decrement instruction  
$\ell_i$: $c=c-1$; goto $\ell_k$. Figure \ref{fig_undec_3swGF_v2} gives the component $DB$ 
which decrements counter $c$, by doubling $\frac{1}{2^c3^d}$.  Assume that 
$x=\frac{1}{2^c3^d}$ and $y=z=0$ on entering $DB$. Lets denote by $x_{old}$ the value
$\frac{1}{2^c3^d}$.   
A time $1-x_{old}$ is spent at the entry node $en_1$ of $DB$, resulting 
in $x=0, y=1-x_{old},z=0$ at location $l$. 
A non-deterministic amount of time $t$ is spent at $l$. 
This makes $x=t$ and $y=1-x_{old},z=0$, at the call port 
of $A_1:C_{DB}$. 
All variables $x,y,z$ are passed by value to $C_{DB}$. 

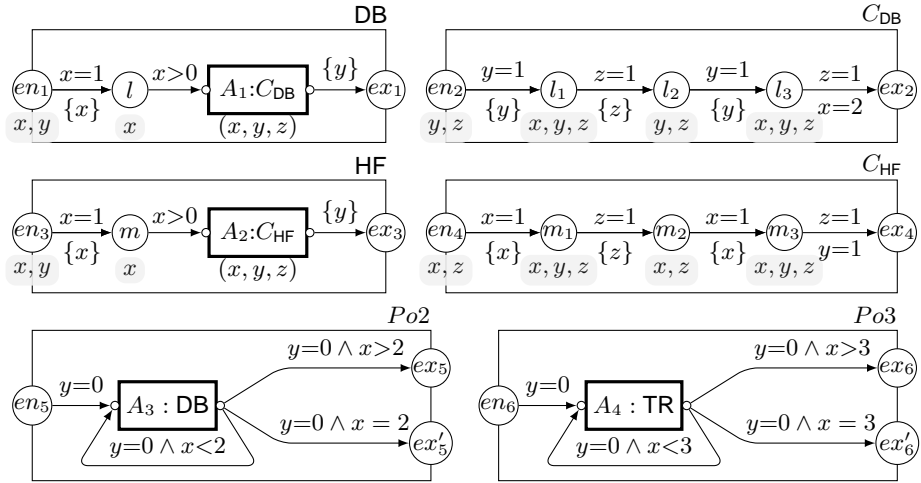
\begin{figure}[t]
  \centering
  {\small
\begin{tikzpicture}[node distance=4cm] 
  
    \draw(0, 0) rectangle (4.7,1.5);
  \draw (4.5, 1.7) node {$\db$};

  \node[loc](n1) at (0, 0.7) {$en_{1}$};  
  \node[rate]() [below of =n1,node distance=5mm] {$\figrate{x,y}$};
  
  \node[loc] (l) at (1.3,0.7) {$l$};
  \node[rate] () [below of = l,node distance = 5mm] {$\figrate{x}$};
  
  \node[boxloc](a1) at (3, 0.7) {$\begin{array}{c} A_1{:}C_{\db}\end{array}$};
  \node () [below of =a1, node distance=5mm]{$(x,y,z)$};
  \node[port](a1n1) [left of=a1,node distance = 6.9mm] {}; 
   \node[port](a1x1) [right of=a1,node distance = 6.9mm] {};
  
  \node[loc](x1) at (4.7, 0.7) {$ex_{1}$};
   
   \draw[trans] (n1)--(l)  node [midway, above]{$x{=}1$};
   \draw[trans] (n1)--(l)  node [midway, below]{$\set{x}$};
   \draw[trans] (l)--(a1n1)  node [midway, above]{$x {>} 0$};
   \draw[trans] (a1x1)--(x1)  node [midway, above]{$\set{y}$};
   %%%%%%%%%%%%%%%%%%%%%%%%%%%%%%%%%%%%%%%%%%%%%%%%%%%%%%%%%%%%%%%%%%%%%%%%%%%%%%% 
  %%%%%%%%%%%%%%%%%%%%%%%%%%%%%% C_\db  %%%%%%%%%%%%%%%%%%%%%%%%%%%%%%%%%%%%%%%%%%%% 
   
  \draw(5.5, 0) rectangle (11.5,1.5);
  \draw (11.3, 1.7) node {$C_{\db}$};

   \node[loc](n2) at (5.5,0.7) {$en_{2}$};     
   \node[rate] () [below of = n2,node distance = 5mm] {$\figrate{y,z}$};

   \node[loc](l1) at (7, 0.7) {$l_1$};
   \node[rate] () [below of = l1,node distance = 5mm] {$\figrate{x,y,z}$};
   
   \node[loc](l2) at (8.5, 0.7) {$l_2$};
   \node[rate] () [below of = l2,node distance = 5mm] {$\figrate{y,z}$};
   
   \node[loc](l3) at (10, 0.7) {$l_3$};
   \node[rate] () [below of = l3,node distance = 5mm] {$\figrate{x,y,z}$};
   
   \node[loc](x2) at (11.5, 0.7) {$ex_{2}$};
   
   \draw[trans] (n2)--(l1) node[midway,above] {$y{=}1$};
   \draw[trans] (n2)--(l1) node[midway,below] {$\set{y}$};

   \draw[trans] (l1)--(l2) node[midway,above] {$z{=}1$};
   \draw[trans] (l1)--(l2) node[midway,below] {$\set{z}$};
   
   \draw[trans] (l2)--(l3) node[midway,above] {$y{=}1$};
   \draw[trans] (l2)--(l3) node[midway,below] {$\set{y}$};
   
   \draw[trans] (l3)--(x2) node[midway,above] {$z{=}1$} node[midway,below] {$x{=}2$};
   
   %%%%%%%%%%%%%%%%%%%%%%%%%%%%%%%%%%%%%%%%%%%%%%%%%%%%%%%%%%%%%%%%%%%%%%%%%%%%%%% 
   %%%%%%%%%%%%%%%%%%%%%%%%%%%%%%%%% HF %%%%%%%%%%%%%%%%%%%%%%%%%%%%%%%%%%%%%%%%%% 
       \draw(0, -2) rectangle (4.7,-0.5);
  \draw (4.5, -0.3) node {$\hf$};

  \node[loc](n3) at (0, -1.2) {$en_{3}$};  
  \node[rate]() [below of =n3,node distance=5mm] {$\figrate{x,y}$};
  
  \node[loc] (m) at (1.3,-1.2) {$m$};
  \node[rate] () [below of = m,node distance = 5mm] {$\figrate{x}$};
  
  \node[boxloc](a2) at (3, -1.2) {$\begin{array}{c} A_2{:}C_{\hf}\end{array}$};
  \node () [below of =a2, node distance=5mm]{$(x,y,z)$};
  \node[port](a2n1) [left of=a2,node distance = 6.9mm] {}; 
   \node[port](a2x1) [right of=a2,node distance = 6.9mm] {};
  
  \node[loc](x3) at (4.7, -1.2) {$ex_{3}$};
   
   \draw[trans] (n3)--(m)  node [midway, above]{$x{=}1$};
   \draw[trans] (n3)--(m)  node [midway, below]{$\set{x}$};
   \draw[trans] (m)--(a2n1)  node [midway, above]{$x {>} 0$};
   \draw[trans] (a2x1)--(x3)  node [midway, above]{$\set{y}$};
   %%%%%%%%%%%%%%%%%%%%%%%%%%%%%%%%%%%%%%%%%%%%%%%%%%%%%%%%%%%%%%%%%%%%%%%%%%%%%%% 
  %%%%%%%%%%%%%%%%%%%%%%%%%%%%%% C_\db  %%%%%%%%%%%%%%%%%%%%%%%%%%%%%%%%%%%%%%%%%%%% 
   
  \draw(5.5, -2) rectangle (11.5,-0.5);
  \draw (11.3, -0.3) node {$C_{\hf}$};

   \node[loc](n4) at (5.5,-1.2) {$en_{4}$};     
   \node[rate] () [below of = n4,node distance = 5mm] {$\figrate{x,z}$};

   \node[loc](m1) at (7, -1.2) {$m_1$};
   \node[rate] () [below of = m1,node distance = 5mm] {$\figrate{x,y,z}$};
   
   \node[loc](m2) at (8.5, -1.2) {$m_2$};
   \node[rate] () [below of = m2,node distance = 5mm] {$\figrate{x,z}$};
   
   \node[loc](m3) at (10, -1.2) {$m_3$};
   \node[rate] () [below of = m3,node distance = 5mm] {$\figrate{x,y,z}$};
   
   \node[loc](x4) at (11.5, -1.2) {$ex_{4}$};
   
   \draw[trans] (n4)--(m1) node[midway,above] {$x{=}1$};
   \draw[trans] (n4)--(m1) node[midway,below] {$\set{x}$};

   \draw[trans] (m1)--(m2) node[midway,above] {$z{=}1$};
   \draw[trans] (m1)--(m2) node[midway,below] {$\set{z}$};
   
   \draw[trans] (m2)--(m3) node[midway,above] {$x{=}1$};
   \draw[trans] (m2)--(m3) node[midway,below] {$\set{x}$};
   
   \draw[trans] (m3)--(x4) node[midway,above] {$z{=}1$} node[midway,below] {$y{=}1$};
   
   %%%%%%%%%%%%%%%%%%%%%%%%%%%%%%%%%%%%%%%%%%%%%%%%%%%%%%%%%%%%%%%%%%%%%%%%%%%%%%% 
   %%%%%%%%%%%%%%%%%%%%%%%%%%%%%%%%%% p02 %%%%%%%%%%%%%%%%%%%%%%%%%%%%%%%%%%%%%%%%%%%%% 
     \draw(0, -4.5) rectangle (5.3, -2.5);
  \draw (5, -2.3) node {$Po2$};

  \node[loc](n25) at (0, -3.5) {$en_{5}$};
  \node[loc](x25yes) at (5.3, -3) {$ex_{5}$};
  \node[loc](x25no) at (5.3, -4) {$ex'_{5}$};

  \node[boxloc](a3) at (1.8, -3.5) {$A_{3}:\db$};
  \node[port](a3n1) at (1.1, -3.5) {};
  \node[port](a3x1) at (2.5, -3.5) {};

  \draw[trans] (n25)--(a3n1)  node [midway, above]{$y {=} 0$};
  \draw[trans] (a3x1) -- +(0.8, 0.5) -- (x25yes) node [midway, above] {$y {=} 0 \wedge x    {>} 2$};
  \draw[trans] (a3x1) -- +(0.8, -0.5) -- (x25no) node [midway, above] {$y {=} 0 \wedge x = 2$};

  \draw[trans] (a3x1)-- +(0.5, -0.8) -- +(-1.9, -0.8) node [midway,
  above]{$y {=} 0  \wedge x {<} 2$} -- (a3n1)  ;

  %%%%%%%%%%%%%%%%%%%%%%%%%%%%%%%%%%%%%%%%%%%%%%%%%%%%%%%%%%%%%%%%%%%%%%%%%%%%%%%
  %%%%%%%%%%%%%%%%%%%%%%%%%%%%%%%%%%%  p03%%%%%%%%%%%%%%%%%%%%%%%%%%%%%%%%%%%%%%%%%%%%
  \draw(6.2, -4.5) rectangle (11.5, -2.5);
  \draw (11.2, -2.3) node {$Po3$};

  \node[loc](n25) at (6.2, -3.5) {$en_{6}$};
  \node[loc](x25yes) at (11.5, -3) {$ex_{6}$};
  \node[loc](x25no) at (11.5, -4) {$ex'_{6}$};

  \node[boxloc](a3) at (8, -3.5) {$A_{4}:\tr$};
  \node[port](a3n1) at (7.3, -3.5) {};
  \node[port](a3x1) at (8.7, -3.5) {};
 
%  \node[loc](u5) at (10, 0.5) {$\ddot \frown$};
%  \draw[trans] (a3x1)-- (u5);

  \draw[trans] (n25)--(a3n1)  node [midway, above]{$y {=} 0$};
  \draw[trans] (a3x1) -- + (0.8, 0.5) -- (x25yes) node [midway, above] {$y {=} 0 \wedge x    {>} 3$};
  \draw[trans] (a3x1) -- + (0.8, -0.5) -- (x25no) node [midway, above] {$y {=} 0 \wedge x = 3$};

  \draw[trans] (a3x1)-- +(0.5, -0.8) -- +(-1.9, -0.8) node [midway,
  above]{$y {=} 0  \wedge x {<} 3$} -- (a3n1)  ;

\end{tikzpicture}
}
   \caption{Glitch-free RSA 3 stopwatch : Decrement $c$, Increment $c$ and $Po2$, $Po3$}
 \label{fig_undec_3swGF_v2}
\end{figure}

At the entry node $en_2$ of $C_{DB}$, the rates of $y,z$ are one.
A time $x_{old}$ is spent at $en_2$, 
 obtaining $y=0,x=t$ and $z=x_{old}$  
  at $l_1$. 
  At $l_1$, a time $1-x_{old}$ is spent, obtaining 
  $x=1+t-x_{old}, y=1-x_{old}$ and $z=0$ at $l_2$.
  A time $x_{old}$ is spent at $l_2$ obtaining $x=1+t-x_{old}$, $y=0,z=x_{old}$ at $l_3$.
  A time $1-x_{old}$ is spent at $l_3$, obtaining $z=1,x=2+t-2x_{old}$ and $y=1-x_{old}$.
  To move out of $l_3$, $x$ must be 2; that is  
  possible iff 
   $t=2x_{old}$.
  In this case, from the return port of $A_1:C_{DB}$ (rates are all 0 here), we 
  reach the exit node $ex_1$ of $DB$ resetting $y$, obtaining $x=t=2x_{old},y=z=0$, thereby successfully 
  decrementing $c$.

\noindent{\bf{Simulate increment instruction}}:
The instruction $\ell_i$: $c=c+1$; goto $\ell_k$ is handled by the component $HF$ in Figure \ref{fig_undec_3swGF_v2}. A time $1-x_{old}$ is spent at entry node $en_3$ 
of $HF$, reaching location   $m$ with $x=0,y=1-x_{old}$ and $z=0$. 
A non-deterministic time $t$ is spent in $m$,  reaching the entry node $en_4$ of $C_{HF}$ 
with $x=t$, $y=1-x_{old}$ and $z=0$. 
The exit node $ex_4$ of $C_{HF}$ can be reached iff $t=\frac{x_{old}}{2}$.
The working of these components are similar to $DB, C_{DB}$.

\noindent{\bf{Zero check instruction}}: 
The component $zerocheck$ simulating $\ell_i$ : if $(d>0)$ then goto $\ell_k$ 
is the same as the $zerocheck$ component in 
Figure \ref{fig_undec_2swG}, where the subcomponent $Po2$ is called, passing all variables by value.
The subcomponent $Po2$ called can be found in 
Figure \ref{fig_undec_3swGF_v2}.  
At the entry node of $Po2$, no time is spent, and we are at the call port of $DB$.
 
We have drawn DB here like a box to avoid clutter, but it is 
actually a  
 transition that goes from $en_5$ on $y=0$ to a location called $en_1$. 
Continue with the transitions drawn inside DB (treat them like normal transitions), 
and we have the sequence of transitions from $en_1$ to $ex_1$, where $C_{DB}$ is called in between. 
The edge $x<2 \wedge y=0$ is a transition from $ex_1$ to $en_1$.  In the figure, 
to avoid clutter, we have drawn it from the return port to the call port of DB.  
This loop from $ex_1$ to $en_1$ is invoked 
repeatedly, until we obtain $x$ exactly equal to 2. If this happens, then we know that $d=0$ 
in $\frac{1}{2^c3^d}=x_{old}$. If this does not happen, then at some point of time, we 
will obtain $x$ as more than 2. 
In this case, $d \neq 0$. 
In the former case, we go the exit node $ex'$ of $Po2$ from $ex_1$, and 
in the latter case, we go to the exit node $ex$ of $Po2$ from $ex_1$. 
Note that, whenever a box is called, we have always passed all the variables only by value.

The propositions proving correctness of the main and sub components is similar to Lemma \ref{app:2sw}. 
Also, it is clear that the node $Halt$ is reached iff the two counter machine halts.
 \qed \end{proof}
  
\subsection{GlitchFree RHA with 2 clocks and 1 stopwatch}
\begin{lemma}
\label{app:2c1s}
  The reachability problem is undecidable for recursive hybrid automata with at least two clocks and one stopwatch.
\end{lemma}
\begin{proof}
The proof of this Lemma is a straightforward adaptation of the techniques used in Lemma \ref{app:2sw}. 
The main difference here is that, at all times, we have to pass all variables either by value, or by reference.
This necessitates the need for an extra variable.
In particular, we pass all variables by reference in all except the zero check module. 
Note that all the calls with pass by reference can be removed by expanding 
%\todo{M:better word than expanding}
the sub-component (callee) in the main component (caller).
Thus, our result holds for the case of ``pass  by value'' RHAs with 2 clocks and 1 stopwatch.

We specify a main component for each instruction of the two counter machine. On entry into 
a main component for increment/decrement/zero check, 
we have two clocks $x=\frac{1}{2^c3^d}$, $y=0$ and one stopwatch $s=0$.
% or $x=0, y=\frac{1}{2^c3^d}$,
 where $c,d$ are the current values of the counters. 
Given a two counter machine, we build a 3 stopwatch RHA 
whose building blocks are the main components for the instructions.   
The purpose of the components is to simulate faithfully the counter machine 
by choosing appropriate delays  to adjust the variables to reflect changes 
in counter values. 
On entering the entry node $en$ of a main component corresponding to an instruction $\ell_i$,
we have the configuration $(\langle \epsilon \rangle, en, (\frac{1}{2^c3^d},0,0))$ 
%or $(\langle \epsilon \rangle, en, (0,\frac{1}{2^c3^d}))$
of the three stopwatch RHA.  
We shall now present the components for increment/decrement and zero check instructions.
In all the components, the ticking variables are written below respective locations in grey.

\noindent{\bf{Simulate decrement instruction}}: Lets consider the decrement instruction  
$\ell_i$: $c=c-1$; goto $\ell_k$. Figure \ref{fig_undec_3swGF_v2} gives the component $DB$ 
which decrements counter $c$, by doubling $\frac{1}{2^c3^d}$.  Assume that 
$x=\frac{1}{2^c3^d}$ and $y=z=0$ on entering $DB$. Lets denote by $x_{old}$ the value
$\frac{1}{2^c3^d}$.   
A time $1-x_{old}$ is spent at the entry node $en_1$ of $DB$, resulting 
in $x=0, y=1-x_{old},z=0$ at location $l$. 
A non-deterministic amount of time $t$ is spent at $l$. 
This makes $x=t$ and $y=1-x_{old},z=0$, at the call port 
of $A_1:C_{DB}$. 
All variables $x,y,z$ are passed by value to $C_{DB}$. 

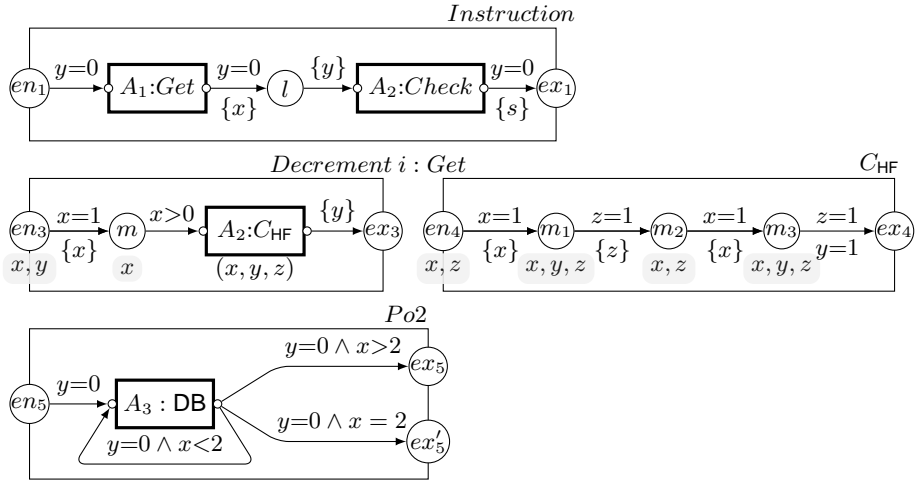
\begin{figure}[t]
  \centering
  {\small
\begin{tikzpicture}[node distance=4cm] 
  
    \draw(0, 0) rectangle (7,1.5);
  \draw (6.4, 1.7) node {$\instn$};

  \node[loc](n1) at (0, 0.7) {$en_{1}$};   
  
  \node[boxloc](a1) at (1.7, 0.7) {$\begin{array}{c} A_1{:}\get\end{array}$};  
  \node[port](a1n1) [left of=a1,node distance = 6.5mm] {}; 
   \node[port](a1x1) [right of=a1,node distance = 6.5mm] {};
  
  \node[loc] (l) at (3.4,0.7) {$l$};
  
  \node[boxloc](a2) at (5.2, 0.7) {$\begin{array}{c} A_2{:}\chk\end{array}$};  
  \node[port](a2n1) [left of=a2,node distance = 8.4mm] {}; 
   \node[port](a2x1) [right of=a2,node distance = 8.4mm] {};

  \node[loc](x1) at (7, 0.7) {$ex_{1}$};
   
   \draw[trans] (n1)--(a1n1)  node [midway, above]{$y{=}0$};
   \draw[trans] (a1x1)--(l) node [midway, above]{$y{=}0$} node [midway, below]{$\set{x}$};
   \draw[trans] (l)--(a2n1)  node [midway, above]{$\set{y}$};
   \draw[trans] (a2x1)--(x1)  node [midway, above]{$y{=}0$} node [midway, below]{$\set{s}$};
   
   %%%%%%%%%%%%%%%%%%%%%%%%%%%%%%%%%%%%%%%%%%%%%%%%%%%%%%%%%%%%%%%%%%%%%%%%%%%%%%% 
   %%%%%%%%%%%%%%%%%%%%%%%%%%%%%%%%% Dec_i:\get %%%%%%%%%%%%%%%%%%%%%%%%%%%%%%%%%%%%%%%%%% 
       \draw(0, -2) rectangle (4.7,-0.5);
  \draw (4.5, -0.3) node {$Decrement~i:\get$};

  \node[loc](n3) at (0, -1.2) {$en_{3}$};  
  \node[rate]() [below of =n3,node distance=5mm] {$\figrate{x,y}$};
  
  \node[loc] (m) at (1.3,-1.2) {$m$};
  \node[rate] () [below of = m,node distance = 5mm] {$\figrate{x}$};
  
  \node[boxloc](a2) at (3, -1.2) {$\begin{array}{c} A_2{:}C_{\hf}\end{array}$};
  \node () [below of =a2, node distance=5mm]{$(x,y,z)$};
  \node[port](a2n1) [left of=a2,node distance = 6.9mm] {}; 
   \node[port](a2x1) [right of=a2,node distance = 6.9mm] {};
  
  \node[loc](x3) at (4.7, -1.2) {$ex_{3}$};
   
   \draw[trans] (n3)--(m)  node [midway, above]{$x{=}1$};
   \draw[trans] (n3)--(m)  node [midway, below]{$\set{x}$};
   \draw[trans] (m)--(a2n1)  node [midway, above]{$x {>} 0$};
   \draw[trans] (a2x1)--(x3)  node [midway, above]{$\set{y}$};
   %%%%%%%%%%%%%%%%%%%%%%%%%%%%%%%%%%%%%%%%%%%%%%%%%%%%%%%%%%%%%%%%%%%%%%%%%%%%%%% 
  %%%%%%%%%%%%%%%%%%%%%%%%%%%%%% C_\db  %%%%%%%%%%%%%%%%%%%%%%%%%%%%%%%%%%%%%%%%%%%% 
   
  \draw(5.5, -2) rectangle (11.5,-0.5);
  \draw (11.3, -0.3) node {$C_{\hf}$};

   \node[loc](n4) at (5.5,-1.2) {$en_{4}$};     
   \node[rate] () [below of = n4,node distance = 5mm] {$\figrate{x,z}$};

   \node[loc](m1) at (7, -1.2) {$m_1$};
   \node[rate] () [below of = m1,node distance = 5mm] {$\figrate{x,y,z}$};
   
   \node[loc](m2) at (8.5, -1.2) {$m_2$};
   \node[rate] () [below of = m2,node distance = 5mm] {$\figrate{x,z}$};
   
   \node[loc](m3) at (10, -1.2) {$m_3$};
   \node[rate] () [below of = m3,node distance = 5mm] {$\figrate{x,y,z}$};
   
   \node[loc](x4) at (11.5, -1.2) {$ex_{4}$};
   
   \draw[trans] (n4)--(m1) node[midway,above] {$x{=}1$};
   \draw[trans] (n4)--(m1) node[midway,below] {$\set{x}$};

   \draw[trans] (m1)--(m2) node[midway,above] {$z{=}1$};
   \draw[trans] (m1)--(m2) node[midway,below] {$\set{z}$};
   
   \draw[trans] (m2)--(m3) node[midway,above] {$x{=}1$};
   \draw[trans] (m2)--(m3) node[midway,below] {$\set{x}$};
   
   \draw[trans] (m3)--(x4) node[midway,above] {$z{=}1$} node[midway,below] {$y{=}1$};
   
   %%%%%%%%%%%%%%%%%%%%%%%%%%%%%%%%%%%%%%%%%%%%%%%%%%%%%%%%%%%%%%%%%%%%%%%%%%%%%%% 
   %%%%%%%%%%%%%%%%%%%%%%%%%%%%%%%%%% p02 %%%%%%%%%%%%%%%%%%%%%%%%%%%%%%%%%%%%%%%%%%%%% 

      %%%%%%%%%%%%%%%%%%%%%%%%%%%%%%%%%%%%%%%%%%%%%%%%%%%%%%%%%%%%%%%%%%%%%%%%%%%%%%% 
   %%%%%%%%%%%%%%%%%%%%%%%%%%%%%%%%%% p02 %%%%%%%%%%%%%%%%%%%%%%%%%%%%%%%%%%%%%%%%%%%%% 
     \draw(0, -4.5) rectangle (5.3, -2.5);
  \draw (5, -2.3) node {$Po2$};

  \node[loc](n25) at (0, -3.5) {$en_{5}$};
  \node[loc](x25yes) at (5.3, -3) {$ex_{5}$};
  \node[loc](x25no) at (5.3, -4) {$ex'_{5}$};

  \node[boxloc](a3) at (1.8, -3.5) {$A_{3}:\db$};
  \node[port](a3n1) at (1.1, -3.5) {};
  \node[port](a3x1) at (2.5, -3.5) {};

  \draw[trans] (n25)--(a3n1)  node [midway, above]{$y {=} 0$};
  \draw[trans] (a3x1) -- +(0.8, 0.5) -- (x25yes) node [midway, above] {$y {=} 0 \wedge x    {>} 2$};
  \draw[trans] (a3x1) -- +(0.8, -0.5) -- (x25no) node [midway, above] {$y {=} 0 \wedge x = 2$};

  \draw[trans] (a3x1)-- +(0.5, -0.8) -- +(-1.9, -0.8) node [midway,
  above]{$y {=} 0  \wedge x {<} 2$} -- (a3n1)  ;

\end{tikzpicture}
}
   \caption{Glitch-free RSA 3 stopwatch : Decrement $c$, Increment $c$ and $Po2$, $Po3$}
 \label{fig_undec_3swGF_v2}
\end{figure}

At the entry node $en_2$ of $C_{DB}$, the rates of $y,z$ are one.
A time $x_{old}$ is spent at $en_2$, 
 obtaining $y=0,x=t$ and $z=x_{old}$  
  at $l_1$. 
  At $l_1$, a time $1-x_{old}$ is spent, obtaining 
  $x=1+t-x_{old}, y=1-x_{old}$ and $z=0$ at $l_2$.
  A time $x_{old}$ is spent at $l_2$ obtaining $x=1+t-x_{old}$, $y=0,z=x_{old}$ at $l_3$.
  A time $1-x_{old}$ is spent at $l_3$, obtaining $z=1,x=2+t-2x_{old}$ and $y=1-x_{old}$.
  To move out of $l_3$, $x$ must be 2; that is  
  possible iff 
   $t=2x_{old}$.
  In this case, from the return port of $A_1:C_{DB}$ (rates are all 0 here), we 
  reach the exit node $ex_1$ of $DB$ resetting $y$, obtaining $x=t=2x_{old},y=z=0$, thereby successfully 
  decrementing $c$.

\noindent{\bf{Simulate increment instruction}}:
The instruction $\ell_i$: $c=c+1$; goto $\ell_k$ is handled by the component $HF$ in Figure \ref{fig_undec_3swGF_v2}. A time $1-x_{old}$ is spent at entry node $en_3$ 
of $HF$, reaching location   $m$ with $x=0,y=1-x_{old}$ and $z=0$. 
A non-deterministic time $t$ is spent in $m$,  reaching the entry node $en_4$ of $C_{HF}$ 
with $x=t$, $y=1-x_{old}$ and $z=0$. 
The exit node $ex_4$ of $C_{HF}$ can be reached iff $t=\frac{x_{old}}{2}$.
The working of these components are similar to $DB, C_{DB}$.

\noindent{\bf{Zero check instruction}}: 
The component $zerocheck$ simulating $\ell_i$ : if $(d>0)$ then goto $\ell_k$ 
is the same as the $zerocheck$ component in 
Figure \ref{fig_undec_2swG}, where the subcomponent $Po2$ is called, passing all variables by value.
The subcomponent $Po2$ called can be found in 
Figure \ref{fig_undec_3swGF_v2}.  
At the entry node of $Po2$, no time is spent, and we are at the call port of $DB$.
 
We have drawn DB here like a box to avoid clutter, but it is 
actually a  
 transition that goes from $en_5$ on $y=0$ to a location called $en_1$. 
Continue with the transitions drawn inside DB (treat them like normal transitions), 
and we have the sequence of transitions from $en_1$ to $ex_1$, where $C_{DB}$ is called in between. 
The edge $x<2 \wedge y=0$ is a transition from $ex_1$ to $en_1$.  In the figure, 
to avoid clutter, we have drawn it from the return port to the call port of DB.  
This loop from $ex_1$ to $en_1$ is invoked 
repeatedly, until we obtain $x$ exactly equal to 2. If this happens, then we know that $d=0$ 
in $\frac{1}{2^c3^d}=x_{old}$. If this does not happen, then at some point of time, we 
will obtain $x$ as more than 2. 
In this case, $d \neq 0$. 
In the former case, we go the exit node $ex'$ of $Po2$ from $ex_1$, and 
in the latter case, we go to the exit node $ex$ of $Po2$ from $ex_1$. 
Note that, whenever a box is called, we have always passed all the variables only by value.

The propositions proving correctness of the main and sub components is similar to Lemma \ref{app:2sw}. 
Also, it is clear that the node $Halt$ is reached iff the two counter machine halts.
 \qed \end{proof}

\subsection{Unrestricted RTA over bounded time}
\label{app:rta}
\begin{lemma}
The time bounded reachability problem is undecidable for recursive timed automata 
  with at least 5 clocks.
\end{lemma}
\begin{proof}
We prove that  the problem of reaching a chosen vertex in an RTA  
within 18 units of total elapsed time is undecidable. 
In order to get the undecidability result, we use a reduction from the halting problem 
for two counter machines. Our reduction uses an RTA with atleast 5 clocks.

We specify a main component for each instruction of the two counter machine. 
We maintain 3 sets of clocks. The first set $X=\{x\}$ 
 encodes correctly the current value of counter $c$;
the second set $Y=\{y\}$  
encodes correctly the current value of counter $d$;
the third set $Z=\{z_1,z_2\}$ of 2 clocks 
helps in zero-check.  
%ensures that the entire simulation of the two counter machine happens in ``bounded'' time.
An extra clock $b$ is used to enforce urgency in some locations.
$b$ is zero at the entry nodes of all the main components.
Let $\mathcal{X}$ denote the set of all 5 clocks.  
The tuple of variables written below each box denotes 
variables passed by value.

To be precise, on entry into a main component 
simulating the $(k+1)$th instruction, 
we have the values of $z_1,z_2$ as $1-\frac{1}{2^k}$,
the value of $x$ as 
$1-\frac{1}{2^{c+k}}$, and the value of $y$ as
$1-\frac{1}{2^{d+k}}$, where $c,d$ are the current values of 
the counters after simulating the first $k$ instructions. 
We will denote this by saying that at the beginning of the $(k+1)$th instruction, 
we have $\nu(Z)=1-\frac{1}{2^k}$, $\nu(x)=1-\frac{1}{2^{c+k}}$ and 
$\nu(y)=1-\frac{1}{2^{d+k}}$. If the $(k+1)$th instruction $\ell_{k+1}$ is an 
increment counter $c$ instruction, then after the simulation 
of $\ell_{k+1}$, we need $\nu(Z)=1-\frac{1}{2^{k+1}}$,
$\nu(x)=1-\frac{1}{2^{c+k+2}}$ and 
$\nu(y)=1-\frac{1}{2^{d+k+1}}$. Similarly, if $\ell_{k+1}$ 
is a decrement $c$ instruction, then 
after the simulation 
of $\ell_{k+1}$, we need $\nu(Z)=1-\frac{1}{2^{k+1}}$,
$\nu(x)=1-\frac{1}{2^{c+k}}$ and 
$\nu(y)=1-\frac{1}{2^{d+k+1}}$.
Likewise, if $\ell_{k+1}$ 
is a zero check instruction, then 
after the simulation 
of $\ell_{k+1}$, we need $\nu(Z)=1-\frac{1}{2^{k+1}}$,
$\nu(x)=1-\frac{1}{2^{c+k+1}}$ and 
$\nu(y)=1-\frac{1}{2^{d+k+1}}$.

\noindent{\bf Simulate Increment Instruction}: Let us discuss the case of simulating an increment 
instruction for counter $c$. Assume that this is the $(k+1)$th instruction.  Figure \ref{fig_undec_5clk_Gl_TB_v3} gives the figure 
for incrementing counter $c$.  At the entry node $en_1$ of 
the component $Inc~c$, we have $\nu(x)=1-\frac{1}{2^{c+k}}$, $\nu(y)=1-\frac{1}{2^{d+k}}$ and $\nu(Z)=
1-\frac{1}{2^k}$, and $\nu(b)=0$.

The component $Inc~c$ has three subcomponents sequentially lined up one after the other: 
Let $\beta=\frac{1}{2^k}, \beta_c=\frac{1}{2^{c+k}}$, and $\beta_d=\frac{1}{2^{d+k}}$.
\begin{enumerate}
\item The first subcomponent 
is $Up_2^y$.  
If $Up_2^y$ is entered with $\nu(y)=1-\beta_d$, then on exit, 
we have $\nu(y)=1-\frac{\beta_d}{2}$. The values of $X,Z$ are unchanged. Also, the total time elapsed 
in $Up_2^y$ is $\leq \frac{5\beta}{2}$.
\item The next subcomponent is $Up_4^x$.  
If $Up_4^x$ is entered with $\nu(x)=1-\beta_c$, then on exit, 
we have $\nu(x)=1-\frac{\beta_c}{4}$. The values of $Z,Y$ are unchanged. Also, the total time elapsed 
in $Up_4^x$ is $\leq \frac{11\beta}{4}$.
\item The next subcomponent is $Up_2^Z$ updates the value of $Z$. 
If $Up_2^Z$ is entered with $\nu(Z)=1-\beta$, then on exit, 
we have $\nu(Z)=1-\frac{\beta}{2}$. The values of $X,Y$ are unchanged. Also, the total time elapsed 
in $Up_2^Z$ is $\leq \frac{5\beta}{2}$. 
\item Thus, at the end of the $Inc~c$, we obtain $\nu(Z)=1-\frac{1}{2^{k+1}}$, 
$\nu(x)= 1-\frac{1}{2^{c+k+2}}$, $\nu(y)=1-\frac{1}{2^{d+k+1}}$. Also, the total time elapsed 
in $Inc~c$ is $\leq [\frac{5}{2}+\frac{11}{4}+\frac{5}{2}]\beta < 8 \beta$.
\end{enumerate}
 On calling $Up_n^a$, for 
$a \in \{x,y\}$, the clock $a$  is passed by reference;
likewise, on calling $Up_n^Z$, clocks in $Z$ are passed by reference.  
Here, $n \in \{2,4\}$.
Next, we describe the structure of the components $Up_n^a$ for $a \in \{x,y\}$. At the entry node $en_2$ of 
 $Up_n^a$, we have the invariant $b=0$. Thus, no time is elapsed 
 in the entry node $en_1$ of $Inc~c$ also. $Up_n^a$ is made up of subcomponents 
 $D,C_{z_2}^{a=},D$ and $Chk_n^a$ lined sequentially.
 Let us discuss the details of $Up_2^y$, the others have similar functionality.
 \begin{enumerate}
 \item On entry into the first subcomponent $F_4{:}D$, we have $\nu(Z)=1-\beta$, $\nu(b)=0$,
 $\nu(x)=1-\beta_c$, $\nu(y)=1-\beta_d$. 
 $D$ is called, and clock $z_2$  is passed by reference and the rest by value.
  A non-deterministic amount of time $t_1$ 
 elapses at the entry node $en_3$ of $D$. Back at the return port of $F_4{:}D$, we 
 have clock $z_2$ added by $t_1$.

 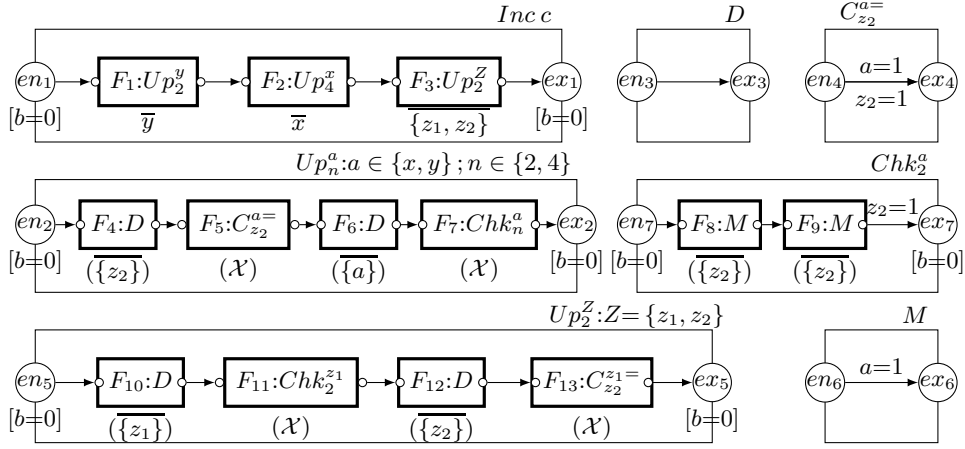
\begin{figure}[t]

{\small
\begin{tikzpicture}[node distance=4cm] 
  
  \draw(0, 0) rectangle (7,1.5);
  \draw (6.5, 1.7) node {$\incc$};

  \node[loc](n61) at (0, 0.8) {$en_{1}$};  
  \node[inv] () [below of = n61,node distance = 5mm] {$\figinv{b{=}0}$};    
  
  \node[boxloc](f1) at (1.5, 0.8) {$\begin{array}{c} F_1{:}\upAn{y}{2}\\ \end{array}$};
  \node () [below of = f1,node distance=5.4mm] {$\varbar{y}$};
  \node[port](f1n1) [left of=f1,node distance = 7mm] {}; 
  \node[port](f1x1) [right of=f1,node distance = 7mm] {};

  \node[boxloc](f2) at (3.5, 0.8) {$\begin{array}{c} F_2{:}\upAn{x}{4}\\ \end{array}$};
  \node () [below of = f2,node distance=5.4mm] {$\varbar{x}$};
  \node[port](f2n1) [left of=f2,node distance = 7mm] {}; 
  \node[port](f2x1) [right of=f2,node distance = 7mm] {};

  \node[boxloc](f3) at (5.5, 0.8) {$\begin{array}{c} F_3{:}\upAn{Z}{2}\\ \end{array}$};
  \node () [below of = f3,node distance=5.4mm] {$\varbar{\set{z_1,z_2}}$};
  \node[port](f3n1) [left of=f3,node distance = 7mm] {}; 
  \node[port](f3x1) [right of=f3,node distance = 7mm] {};

  \node[loc](x61) at (7, 0.8) {$ex_{1}$};
  \node[inv] () [below of = x61,node distance = 5mm] {$\figinv{b{=}0}$};  

   \draw[trans] (n61)-- (f1n1) {};
   \draw[trans] (f1x1)--(f2n1) {};
   \draw[trans] (f2x1)--(f3n1) {};
   \draw[trans] (f3x1)-- (x61)	;
   
   %%%%%%%%%%%%%%%%%%%%%%%%%%%%%%%%%%%%%%%%%%%%%%%%%%%%%%%%%%%%%%%%%%%%%%%%%%%%%%% 
   %%%%%%%%%%%%%%%%%%%%%%%%%%%%%%%%%% \Del %%%%%%%%%%%%%%%%%%%%%%%%%%%%%%%%%
 
    \draw(8, 0) rectangle (9.5,1.5);
   \draw (9.3, 1.7) node {$\Del$};

   \node[loc](n63) at (8, 0.8) {$\DEn$};  
   \node[loc](x63) at (9.5, 0.8) {$\DEx$};
     
   \draw[trans] (n63)--(x63);
    
   %%%%%%%%%%%%%%%%%%%%%%%%%%%%%%%%%%%%%%%%%%%%%%%%%%%%%%%%%%%%%%%%%%%%%%%%%%%%%%% 
   %%%%%%%%%%%%%%%%%%%%%%%%%%%%%%%%%%%%%%%%%%%%%%%%%%%%%%%%%%%%%%%%%%%%%%%%%%%%%%% 

   %%%%%%%%%%%%%%%%%%%%%%%%%%%%%%%%%%%%%%%%%%%%%%%%%%%%%%%%%%%%%%%%%%%%%%%%%%%%%%% 
   %%%%%%%%%%%%%%%%%%%%%%%%%%%%%%%%%% \chkEQ %%%%%%%%%%%%%%%%%%%%%%%%%%%%%%%%%
 
    \draw(10.5, 0) rectangle (12,1.5);
   \draw (11, 1.7) node {$\chkEQn{a}{z_2}$};
   %    \draw (11, 1.7) node {$\begin{array}{c}\chkEQn{a}{z_2};\\a\in\set{x,y,z_1}\end{array}$};

   \node[loc](n64) at (10.5, 0.8) {$\chkEQEn$};  
   \node[loc](x64) at (12, 0.8) {$\chkEQEx$};
     
   \draw[trans] (n64)--(x64) node [midway, above]{$a{=}1$};
   \draw[trans] (n64)--(x64) node [midway, below]{$z_2{=}1$};
    
   %%%%%%%%%%%%%%%%%%%%%%%%%%%%%%%%%%%%%%%%%%%%%%%%%%%%%%%%%%%%%%%%%%%%%%%%%%%%%%% 
   %%%%%%%%%%%%%%%%%%%%%%%%%%%%%%%%%%%%%%%%%%%%%%%%%%%%%%%%%%%%%%%%%%%%%%%%%%%%%%% 
   
  %%%%%%%%%%%%%%%%%%%%%%%%%%%%%%%%%%%%%%%%%%%%%%%%%%%%%%%%%%%%%%%%%%%%%%%%%%%%%%% 
  %%%%%%%%%%%%%%%%%%%%%%%%%%%%%%%%%% Up(A,n) %%%%%%%%%%%%%%%%%%%%%%%%%%%%%%%%%%%%

  \draw(0, -2) rectangle (7.2,-0.5);
  \draw (5.3, -0.3) node {$\upAn{a}{n} {:} a \in \set{x,y}; n \in \set{2,4}$};

  \node[loc](n62) at (0, -1.1) {$\upAnEn$}; 
  \node[inv] () [below of = n62,node distance = 5mm] {$\figinv{b{=}0}$};  
  
  \node[boxloc](f4) at (1.1, -1.1) {$\begin{array}{c} F_4{:}\Del\end{array}$};
  \node () [below of = f4, node distance = 6mm] {$(\varbar{\set{z_2}})$};
  \node[port](f4n1) [left of=f4,node distance = 5mm] {}; 
  \node[port](f4x1) [right of=f4,node distance = 5mm] {};  
  
  \node[boxloc](f5) at (2.7, -1.1) {$\begin{array}{c} F_5{:}\chkEQn{a}{z_2}\end{array}$};
  \node () [below of = f5, node distance = 6mm] {$(\variables)$};
  \node[port](f5n1) [left of=f5,node distance = 7.5mm] {}; 
  \node[port](f5x1) [right of=f5,node distance = 7.5mm] {};
  
  \node[boxloc](f6) at (4.3, -1.1) {$\begin{array}{c} F_6{:}\Del\end{array}$};
  \node () [below of = f6, node distance = 6mm] {$(\varbar{\set{a}})$};
  \node[port](f6n1) [left of=f6,node distance = 5mm] {}; 
  \node[port](f6x1) [right of=f6,node distance =5mm] {};

  \node[boxloc](f7) at (5.9, -1.1) {$\begin{array}{c} F_7{:}\chkAn{a}{n}\end{array}$}; 
  \node () [below of = f7, node distance = 6mm] {$(\variables)$};
  \node[port](f7n1) [left of=f7,node distance = 7.5mm] {}; 
  \node[port](f7x1) [right of=f7,node distance = 7.5mm] {};

  \node[loc](x62) at (7.2, -1.1) {$\upAnEx$};
  \node[inv] () [below of = x62,node distance = 5mm] {$\figinv{b{=}0}$};

  \draw[trans] (n62)--(f4n1);
  \draw[trans] (f4x1)--(f5n1);
  \draw[trans] (f5x1)--(f6n1);
  \draw[trans] (f6x1)--(f7n1);
  \draw[trans] (f7x1) -- (x62);
  
  %%%%%%%%%%%%%%%%%%%%%%%%%%%%%%%%%%%%%%%%%%%%%%%%%%%%%%%%%%%%%%%%%%%%%%%%%%%%%%% 
  %%%%%%%%%%%%%%%%%%%%%%%%%%%%%%%%%% \chkAn{a,2} %%%%%%%%%%%%%%%%%%%%%%%%%%%%%%%%%
 
   \draw(8, -2) rectangle (12,-0.5);
   \draw (11.5, -0.3) node {$\chkAn{a}{2}$};

   \node[loc](n67) at (8, -1.1) {$\chkAtwoEn$};  
   \node[inv] () [below of = n67,node distance = 5mm] {$\figinv{b{=}0}$};  
   
   \node[boxloc](f8) at (9.1, -1.1) {$\begin{array}{c} F_8{:}\DEq{a}\end{array}$};
   \node () [below of =f8,node distance=6mm]{$(\varbar{\set{z_2}})$};
   \node[port](f8n1) [left of=f8,node distance = 5mm] {}; 
   \node[port](f8x1) [right of=f8,node distance = 5mm] {};
   
   \node[boxloc](f9) at (10.5, -1.1) {$\begin{array}{c} F_9{:}\DEq{a} \end{array}$};
   \node () [below of =f9,node distance=6mm]{$(\varbar{\set{z_2}})$};
   \node[port](f9n1) [left of=f9,node distance = 5mm] {}; 
   \node[port](f9x1) [right of=f9,node distance = 5mm] {};
   
   \node[loc](x67) at (12, -1.1) {$\chkAtwoEx$};
   \node[inv] () [below of = x67,node distance = 5mm] {$\figinv{b{=}0}$};

   \draw[trans] (n67)--(f8n1);
   \draw[trans] (f8x1) -- (f9n1);  
   \draw[trans] (f9x1) -- (x67) node[midway,above]{$ z_2{=}1$};
    
   %%%%%%%%%%%%%%%%%%%%%%%%%%%%%%%%%%%%%%%%%%%%%%%%%%%%%%%%%%%%%%%%%%%%%%%%%%%%%%% 
   %%%%%%%%%%%%%%%%%%%%%%%%%%%%%%%%%%%%%%%%%%%%%%%%%%%%%%%%%%%%%%%%%%%%%%%%%%%%%%% 
%   
% 
     %%%%%%%%%%%%%%%%%%%%%%%%%%%%%%%%%%%%%%%%%%%%%%%%%%%%%%%%%%%%%%%%%%%%%%%%%%%%%%% 
   %%%%%%%%%%%%%%%%%%%%%%%%%%%%%%%%%% \DEq{a} %%%%%%%%%%%%%%%%%%%%%%%%%%%%%%%%%
 
   \draw(10.5, -4) rectangle (12,-2.5);
   \draw (11.7, -2.3) node {$\DEq{a}$};

   \node[loc](n63) at (10.5, -3.2) {$en_6$};  
   \node[loc](x63) at (12, -3.2) {$ex_6$};
     
   \draw[trans] (n63)--node[midway,above]{$a{=}1$}(x63);
    
%   %%%%%%%%%%%%%%%%%%%%%%%%%%%%%%%%%%%%%%%%%%%%%%%%%%%%%%%%%%%%%%%%%%%%%%%%%%%%%%% 
%   %%%%%%%%%%%%%%%%%%%%%%%%%%%%%%%%%%%%%%%%%%%%%%%%%%%%%%%%%%%%%%%%%%%%%%%%%%%%%%% 
% 
% %   
%   %%%%%%%%%%%%%%%%%%%%%%%%%%%%%%%%%%%%%%%%%%%%%%%%%%%%%%%%%%%%%%%%%%%%%%%%%%%%%%% 
%   %%%%%%%%%%%%%%%%%%%%%%%%%%%%%%%%%% zero check %%%%%%%%%%%%%%%%%%%%%%%%%%%%%%%%%
% 
%   \draw(10, -2) rectangle (12,-0.5);
%   \draw (11, -0.3) node {$Zero Chk {:} c=0?$};
%   
%   
%   \node[loc](n65) at (10, -1.25) {$en_{5}$};  
%   \node[loc](x65yes) at (12, -0.9) {$ex_{5}$};
%   \node[loc](x65no) at (12, -1.6) {$ex_{5}'$};
%  
%   
% %   \draw[trans] (n65)--(x65yes) node [midway, above]{$z_1=1 \wedge x_1=1$};
% %   \draw[trans] (n65)--(x65no) node [midway, below]{$z_1=1 \wedge x_1 \not =1$};
% 
% %     \draw[trans] (n65)--(x65yes) node [rotate=12,midway]{$\begin{array}{l}z_1{=}1 \wedge\\ x_1{=}1\end{array}$};
% %     \draw[trans] (n65)--(x65no) node [rotate=-12,midway]{$\begin{array}{l}z_1{=}1 \wedge\\ x_1{\not=}1\end{array}$};
% 
%   \draw[trans] (n65)--(10.5,-0.9) -- (x65yes) node [midway]{$\begin{array}{c}z_1{=}1\\ \wedge x_1{=}1\end{array}$};
%   \draw[trans] (n65)--(10.5,-1.6) -- (x65no) node [midway]{$\begin{array}{c}z_1{=}1 \\\wedge x_1 {\not =}1\end{array}$};
%   %%%%%%%%%%%%%%%%%%%%%%%%%%%%%%%%%%%%%%%%%%%%%%%%%%%%%%%%%%%%%%%%%%%%%%%%%%%%%%% 
%   %%%%%%%%%%%%%%%%%%%%%%%%%%%%%%%%%%%%%%%%%%%%%%%%%%%%%%%%%%%%%%%%%%%%%%%%%%%%%%% 
%   
  
   %%%%%%%%%%%%%%%%%%%%%%%%%%%%%%%%%%%%%%%%%%%%%%%%%%%%%%%%%%%%%%%%%%%%%%%%%%%%%%% 
  %%%%%%%%%%%%%%%%%%%%%%%%%%%%%%%%%% Up(Z,n) %%%%%%%%%%%%%%%%%%%%%%%%%%%%%%%%%%%%

  \draw(0, -4) rectangle (9,-2.5);
  \draw (8, -2.3) node {$\upAn{Z}{2} {:} Z {=} \set{z_1,z_2}$};

  \node[loc](n65) at (0, -3.2) {$en_5$}; 
  \node[inv] () [below of = n65,node distance = 5mm] {$\figinv{b{=}0}$};  
  
  \node[boxloc](f10) at (1.4, -3.2) {$\begin{array}{c} F_{10}{:}\Del\end{array}$};
  \node () [below of = f10, node distance = 6mm] {$(\varbar{\set{z_1}})$};
  \node[port](f10n1) [left of=f10,node distance = 5.5mm] {}; 
  \node[port](f10x1) [right of=f10,node distance = 5.5mm] {};  
  
  \node[boxloc](f11) at (3.4, -3.2) {$\begin{array}{c} F_{11}{:}\chkAn{z_1}{2}\end{array}$};
  \node () [below of = f11, node distance = 6mm] {$(\variables)$};
  \node[port](f11n1) [left of=f11,node distance = 9.5mm] {}; 
  \node[port](f11x1) [right of=f11,node distance = 9.5mm] {};
  
  \node[boxloc](f12) at (5.4, -3.2) {$\begin{array}{c} F_{12}{:}\Del\end{array}$};
  \node () [below of = f12, node distance = 6mm] {$(\varbar{\set{z_2}})$};
  \node[port](f12n1) [left of=f12,node distance = 5.5mm] {}; 
  \node[port](f12x1) [right of=f12,node distance =5.5mm] {};

  \node[boxloc](f13) at (7.4, -3.2) {$\begin{array}{c} F_{13}{:}\chkEQn{z_1}{z_2}\end{array}$}; 
  \node () [below of = f13, node distance = 6mm] {$(\variables)$};
  \node[port](f13n1) [left of=f13,node distance = 7.5mm] {}; 
  \node[port](f13x1) [right of=f13,node distance = 7.5mm] {};

  \node[loc](x65) at (9, -3.2) {$ex_5$};
  \node[inv] () [below of = x65,node distance = 5mm] {$\figinv{b{=}0}$};

  \draw[trans] (n65)--(f10n1);
  \draw[trans] (f10x1)--(f11n1);
  \draw[trans] (f11x1)--(f12n1);
  \draw[trans] (f12x1)--(f13n1);
  \draw[trans] (f13x1) -- (x65);

\end{tikzpicture}
}
 \caption{$TB~Term$ in RTA: Increment $c$. Note that $\varbar{S} = \variables-S$ and $a \in \set{x,y}$. 
 $C_{z_2}^{z_1=}$ is obtained by instantiating $a=z_1$ in  $C_{z_2}^{a=}$. 
  The component $\chkAn{a}{4}$ is similar to $\chkAn{a}{2}$. It has 4 calls to $\DEq{a}$ inside it each time passing only $z_2$ by reference. Zero Check component follows the same pattern as $\incc$ calling $\upAn{a}{2}$, for all $a \in \set{x,y}$, followed by $Up_2^Z$, and then calls $ZC$ passing all variables by value. $ZC$ checks if $z_1=x$ (with guard $z_1{=}1\wedge x=1$) to check if counter $c$ is 0 and $z_1=y$ to check if $d$ is 0.}
 \label{fig_undec_5clk_Gl_TB_v3}
\end{figure}

 \item We are then at the entry node of the subcomponent $F_5{:}C_{z_2}^{y=}$ with values
 $\nu(z_2)=1-\beta+t_1$,  and $\nu(z_1)=1-\beta$,  
 $\nu(x)=1-\beta_c$, $\nu(y)=1-\beta_d$ and $\nu(b)=0$.
 $C_{z_2}^{y=}$ is called by passing all clocks by value. 
 The subcomponent $C_{z_2}^{y=}$ ensures that $t_1=\beta-\beta_d$.
 \item To ensure $t_1=\beta-\beta_d$, at the entry node $en_4$ of $C_{z_2}^{y=}$, 
 a time $\beta_d$ elapses. This makes $y=1$. If $z_2$ must be 1, then 
 we need $1-\beta+t_1+\beta_d=1$, or the time $t_1$ elapsed is 
 $\beta-\beta_d$. That is, 
 $C_{z_2}^{y=}$ ensures that $z_2$ has grown 
 to be equal to $y$ by calling $F_4{:}D$.
  Back at the return port of $F_5{:}C_{z_2}^{y=}$, we next 
  enter the call port  of $F_6{:}D$ with $\nu(z_2)=\nu(y)=1-\beta_d$ and 
  $\nu(z_1)=1-\beta$. $D$ is called by passing 
  $y$ by reference, and all others by value. 
  A non-deterministic amount of time $t_2$ is elapsed in $D$.  
  At the return port of $F_6{:}D$, we get $\nu(z_1)=1-\beta$, 
  $\nu(z_2)=1-\beta_d$, and $\nu(y)=1-\beta_d+t_2$. 
  
  \item At the call port of $F_7{:}Chk_2^y$, we have the same values, since 
  $b=0$ has to be satisfied at the exit node $ex_7$ of $Chk_2^y$. 
  That is, at the call port of $F_7{:}Chk_2^y$, we have 
  $\nu(z_1)=1-\beta$, 
  $\nu(z_2)=1-\beta_d$, and $\nu(y)=1-\beta_d+t_2$.
  $F_7$ calls $Chk_2^y$, and 
  passes all clocks by value. 
  $Chk_2^y$ checks that $t_2=\frac{\beta_d}{2}$. 
  
  \item At the entry port $en_7$ of $Chk_2^y$, no time elapses.
  $Chk_2^y$ sequentially calls $M$ twice, each time passing $z_2$ by reference, and all others by value.
 In  the first invocation of $M$, we want $y$ to reach 1; thus a time $\beta_d-t_2$  
  is spent at $en_6$. This makes $z_2=1-\beta_d+\beta_d-t_2=1-t_2$.
  After the second invocation, we obtain $z_2=1+\beta_d-2t_2$ 
  at the return port of $F_9{:}M$. No time can elapse at the return port 
  of $F_9{:}M$; 
  for $z_2$ to be 1, we need $t_2=\frac{\beta_d}{2}$. 
\item No time elapses in the return port of $F_7{:}Chk_2^y$, and we are at the exit node $ex_2$ 
of $Up_2^y$.  
\item Thus, at the exit node of $Up_2^y$, we have 
$\nu(z_1)=1-\beta$, $\nu(z_2)=1-\beta_d$ and $\nu(y)=1-\beta_d+t_2=1-\frac{\beta_d}{2}$.
\item The time elapsed in $Up_2^y$ is 
the sum of $t_1, t_2$ and the times elapsed in $C_{z_2}^{y=}$ and 
$Chk_2^y$. That is, $(\beta - \beta_d)+\frac{\beta_d}{2}+\beta_d+2(\beta_d-t_2)$ 
    =$\beta+\frac{3\beta_d}{2} \leq \frac{5\beta}{2}$ since $\beta_d \leq \beta$.
 \end{enumerate}
At the return port of $F_1:Up_2^y$, we thus have 
$\nu(Z)=1-\beta$ ($\nu(z)$ restored to $1-\beta$ as it was passed by value to $Up^y_2$), $\nu(x)=1-\beta_c$, and $\nu(y)=1-\frac{\beta_d}{2}$. No time elapses 
here, and we are at the call port of $F_2:Up_4^x$. 
The component $\chkAn{x}{4}$ is similar to $\chkAn{y}{2}$. It has 4 calls to $M$ inside it, 
 each passing respectively, $z_2$ by reference to $M$ and $x$ by value. 
 %Note that we never invoke $\chkAn{A}{4}$ 
% with $Z$ as $A$ and hence always have $a_4$, $a_5$ in this component. 
An analysis similar to the above gives that 
the total time elapsed in $Up_4^x$ is $\leq \frac{11 \beta}{4}$, and at the return port 
of $F_2:Up_4^x$, we get $\nu(x)=1-\frac{\beta_c}{4}$, $\nu(y)=1-\frac{\beta_d}{2}$ and 
$\nu(Z)=1-\beta$. This is followed by entering $F_3:Up_2^Z$, 
with these values. 
At the return port of  $F_3:Up_2^Z$, we obtain 
$\nu(x)=1-\frac{\beta_c}{4}$, $\nu(y)=1-\frac{\beta_d}{2}$ and 
$\nu(Z)=1-\frac{\beta}{2}$, with the total time elapsed 
in $Up_2^Z$ being $\leq \frac{5\beta}{2}$. 

From the explanations above, the following propositions 
can be proved. The same arguments given above will apply to prove this.

\begin{proposition}
 For any box $B$ and context $\langle \kappa \rangle$, and $\nu(Z)=1-\beta$, 
 we have that $(\langle \kappa \rangle, (B,en), (\nu(x), \nu(y),1-\beta, \nu(b))) 
 \stackrel [\upAn{Z}{2}]{*}{\longrightarrow} (\langle \kappa \rangle, (B,ex), (\nu(x), \nu(y),1-\frac{\beta}{2}, \nu(b)))$.
\end{proposition}

\begin{proposition}
 For any box $B$ and context $\langle \kappa \rangle$, and $\nu(x)=1-\beta_c$, 
 we have that $(\langle \kappa \rangle, (B,en), (1-\beta_c, \nu(y), \nu(Z), \nu(b))) 
 \stackrel [\upAn{x}{4}]{*}{\longrightarrow} (\langle \kappa \rangle, (B,ex), 1-\frac{\beta_c}{4}, \nu(y), \nu(Z), \nu(b)))$.
\end{proposition}

\begin{proposition}
 For any box $B$ and context $\langle \kappa \rangle$, and $\nu(y)=1-\beta_d$, 
 we have that $(\langle \kappa \rangle, (B,en), (\nu(x),1-\beta_d, \nu(Z),\nu(b))) 
 \stackrel [\upAn{y}{2}]{*}{\longrightarrow} (\langle \kappa \rangle, (B,ex), (\nu(x),1-\frac{\beta_d}{2}, \nu(Z), \nu(b)))$.
\end{proposition}

\noindent{\bf Simulate Decrement Instruction}: 
Assume that the $(k+1)$st instruction is decrementing counter $c$. Then 
we construct the main component $Dec~c$
 similar to the component $Inc~c$ above. The main change is the following:
 \begin{itemize}
 \item The main component $Dec~c$ will have the subcomponents 
   $Up_2^y$ and  $Up_2^Z$ lined up sequentially. 
   There is no need for any $Up_n^x$ subcomponent here, since 
   the value of $x$ stays unchanged on decrementing $c$. 
   Also, the subcomponents $Up_2^y$ and  $Up_2^Z$ do not alter 
   the value of $x$ : 
   the functionality 
   $Up_2^Z$ and  $Up_2^y$ are the same as the one in $Inc~c$.
  The total time spent in $Dec~c$ is also, less than $8 \beta$.
  \end{itemize}

\noindent{\bf Zero Check Instruction}:
The main component for Zero Check  follows the same pattern as $\incc$.
The main change is the following:
 \begin{itemize}
 \item The main component $ZeroCheck$ will have the subcomponents 
   $Up_2^y$ and $Up_2^x$ and $Up_2^Z$ lined up sequentially. The functionality 
   $Up_2^Z,Up_2^x$ and  $Up_2^y$ are the same as the one in $Inc~c$. 
   After these three, we invoke a subcomponent $ZC$.
   $ZC$ is called by passing all clocks by value. 
         At the entry node $en$ of 
   $ZC$, we have two transitions, one on $z_1=1 \wedge x=1$ leading to an exit node $ex$, and another 
   one on $z_1= 1\wedge x \neq 1$ leading to $ex'$.
   Recall that $z_1=\frac{1}{2^k}$. Thus, for $z_1$  to reach 1, 
     a time elapse $\frac{\beta}{2}=\frac{1}{2^{k+1}}$ is needed. If this also makes makes $x=1$, then 
   we know that $x$ on entry was $1-\frac{1}{2^{c+k+1}}= 1-\frac{1}{2^{k+1}}$ 
   impying that $c=0$. Likewise, if $z_1$ attains 1, but $x$ does not, then $c \neq 0$.  
 Since all clocks are passed by value, at the return port of $ZC$ within the main component
 $ZeroCheck$, we regain back the clock values obtained 
 after going through $Up_2^Z,Up_2^x$ and  $Up_2^y$: that is, 
 $\nu(Z)=1-\frac{\beta}{2}$, $\nu(x)=1-\frac{\beta_c}{2}$ and $\nu(y)=1-\frac{\beta_d}{2}$. 
  The time elapsed in $ZC$ is $\frac{\beta}{2}$. The times elapsed in   
     $Up_2^Z$ and $Up_2^x$ and $Up_2^y$ are same as calculated 
     in the case of Increment $c$. Thus, the total time elapsed here is $< 8\beta+\beta=9\beta$.
           \end{itemize}

 We conclude by calculating the total time elapsed during the entire simulation. 
 We have established so far that for the $(k+1)$th instruction, the time elapsed is no more than $9 \beta$, 
 for $\beta=\frac{1}{2^k}$.  For the first instruction, the time elapsed is at most $9 $, for the second instruction it is 
 $\frac{9}{2}$, for the third it is $\frac{9}{2^2}$ and so on. 
\begin{equation}
  \text{Total~time~duration} = \left \{
 \begin{array}{l}
 9 (1 + \frac{1}{2} +\frac{1}{4} + \frac{1}{8} + \frac{1}{16} + \cdots)\\
  =9 ( 1 + (<1))\\
  < 18
 \end{array}
 \right.
\end{equation}

Note that the components for incrementing, decrementing and zero check of counter $d$ 
can be obtained in a manner similar to the above.

The proof that we reach the vertex $Halt$ of the RTA iff the two counter machine halts 
follows: Clearly, the exit node of each main component is reached iff the corresponding instruction is simulated correctly. Thus, if the counter machine halts, we will indeed reach the exit node of the main component corresponding 
to the last instruction. However, if the machine does not halt, then 
we keep going between the various main components simulating each instruction, and never reach $Halt$. 
 \qed
\end{proof}

\subsection{Glitchfree RHA with 14 stopwatches}
\begin{lemma}
The time bounded reachability problem is undecidable for recursive hybrid automata 
  with at least 14 stopwatches.
\end{lemma}
\begin{proof}
We prove that  the problem of reaching a chosen vertex in an RHA  
within 18 units of total elapsed time is undecidable. 
In order to get the undecidability result, we use a reduction from the halting problem 
for two counter machines. Our reduction uses an RHA with atleast 14 stopwatches.

We specify a main component for each instruction of the two counter machine. 
We maintain 3 sets of stopwatches. The first set $X=\set{x_1,\cdots x_5}$ encodes correctly the current value of counter $c$;
the second set $Y=\set{y_1,\cdots y_5}$ encodes correctly the current value of counter $d$;
and third set $Z=\set{z_1,z_2,z_3}$ encodes the end of $(k)$th instruction.

An extra stopwatch $b$ is used to enforce urgency in some locations.
$b$ is zero at the entry nodes of all the main components.
Let $\variables$ denote the set of all 14 stopwatches.  

To be precise, on entry into a main component 
simulating the $(k+1)$th instruction, 
we have the values of $z_1,z_2,z_3$ as $1-\frac{1}{2^k}$,
the values of $x_1, \dots, x_5$ as 
$1-\frac{1}{2^{c+k}}$, and the values of $y_1, \dots, y_5$ as
$1-\frac{1}{2^{d+k}}$, where $c,d$ are the current values of 
the counters after simulating the first $k$ instructions. 
We will denote this by saying that at the beginning of the $(k+1)$th instruction, 
we have $\nu(Z)=1-\frac{1}{2^k}$, $\nu(X)=1-\frac{1}{2^{c+k}}$ and 
$\nu(Y)=1-\frac{1}{2^{d+k}}$. If the $(k+1)$th instruction $\ell_{k+1}$ is an 
increment counter $c$ instruction, then after the simulation 
of $\ell_{k+1}$, we need $\nu(Z)=1-\frac{1}{2^{k+1}}$,
$\nu(X)=1-\frac{1}{2^{c+k+2}}$ and 
$\nu(Y)=1-\frac{1}{2^{d+k+1}}$. Similarly, if $l_{k+1}$ 
is a decrement $c$ instruction, then 
after the simulation 
of $\ell_{k+1}$, we need $\nu(Z)=1-\frac{1}{2^{k+1}}$,
$\nu(X)=1-\frac{1}{2^{c+k}}$ and 
$\nu(Y)=1-\frac{1}{2^{d+k+1}}$.
Likewise, if $\ell_{k+1}$ 
is a zero check instruction, then 
after the simulation 
of $\ell_{k+1}$, we need $\nu(Z)=1-\frac{1}{2^{k+1}}$,
$\nu(X)=1-\frac{1}{2^{c+k+1}}$ and 
$\nu(Y)=1-\frac{1}{2^{d+k+1}}$.

\noindent{\bf Simulate Increment Instruction}: Let us discuss the case of simulating an increment 
instruction for counter $c$. Assume that this is the $(k+1)$th instruction.  Figure \ref{fig_undec_14sw_GF_TB_v3} gives the figure 
for incrementing couner $c$.  At the entry node $en_1$ of 
the component $Inc~c$, we have $\nu(X)=1-\frac{1}{2^{c+k}}$, $\nu(Y)=1-\frac{1}{2^{d+k}}$ and $\nu(Z)=
1-\frac{1}{2^k}$, and $\nu(b)=0$.

The component $Inc~c$ has three subcomponents sequentially lined up one after the other: 
Let $\beta=\frac{1}{2^k}, \beta_c=\frac{1}{2^{c+k}}$, and $\beta_d=\frac{1}{2^{d+k}}$.
\begin{enumerate}
\item The first subcomponent $Up_2^Z$ (component $UP_n^A$ in the figure \ref{fig_undec_14sw_GF_TB_v3} with $A$ as $Z$ and $n=2$)  updates the value of $Z$. 
If $Up_2^Z$ is entered with $\nu(Z)=1-\beta$, then on exit, 
we have $\nu(Z)=1-\frac{\beta}{2}$. The values of $X,Y$ are unchanged as their rate of growth is 0 throughout the component $UP_2^Z$ and they are always passed by value to the subcomponents. 
Also, the total time elapsed in $Up_2^Z$ is $\leq \frac{5\beta}{2}$. 
\item The next subcomponent is $Up_4^X$.  
If $Up_4^X$ is entered with $\nu(X)=1-\beta_c$, then on exit, 
we have $\nu(X)=1-\frac{\beta_c}{4}$. The values of $Z,Y$ are unchanged. Also, the total time elapsed 
in $Up_4^X$ is $\leq \frac{11\beta}{4}$.
\item The next subcomponent is $Up_2^Y$.  
If $Up_2^Y$ is entered with $\nu(Y)=1-\beta_d$, then on exit, 
we have $\nu(Y)=1-\frac{\beta_d}{2}$. The values of $X,Z$ are unchanged. Also, the total time elapsed 
in $Up_2^Y$ is $\leq \frac{5\beta}{2}$.
\item Thus, at the end of the $Inc~c$, we obtain $\nu(Z)=1-\frac{1}{2^{k+1}}$, 
$\nu(X)= 1-\frac{1}{2^{c+k+2}}$, $\nu(Y)=1-\frac{1}{2^{d+k+1}}$. Also, the total time elapsed 
in $Inc~c$ is $\leq [\frac{5}{2}+\frac{11}{4}+\frac{5}{2}]\beta < 8 \beta$.
\end{enumerate}
%Note that all stopwatches are passed by reference on calling $Up_n^A$, for 
%$A \in \{X,Y,Z\}$ and 
%$n \in \{2,4\}$. Infact, 
To avoid clutter, We have drawn $Up_2^Z$ like a box inside $Inc~c$; actually, think 
of it as the sequence of transitions 
from $en_2$ to $ex_2$, with 2 boxes called in between. The same holds 
for ``boxes'' $Up_4^X$ and $Up_2^Y$.

Next, we describe the structure of the components $Up_n^A$. At the entry node $en_2$ of 
 $Up_n^A$, we have the invariant $b=0$. Thus, no time is elapsed 
 in the entry node $en_1$ of $Inc~c$ also. $Up_n^A$ is made up of subcomponents 
 $Chk_n^A$ and $Chk^{=}$.
 Let us discuss the details of $Up_2^Z$, the others have similar functionality.
 \begin{enumerate}
 \item On entry into the location $m_1$, we have $\nu(Z)=1-\beta$, $\nu(b)=0$,
 $\nu(X)=1-\beta_c$, $\nu(Y)=1-\beta_d$. 
 In $m_1$ only stopwatches $z_1,z_3$  are grow. A non-deterministic amount of time $t_1$ 
 elapses here. Thus when leaving $m_1$, we 
 have stopwatches $z_1,z_3$ added by $t_1$. 
 
 \begin{figure}[t]

{\small
\begin{tikzpicture}[node distance=4cm] 
  
  \draw(0, 0) rectangle (7,1.5);
  \draw (6.5, 1.7) node {$\incc$};

  \node[loc](n61) at (0, 0.8) {$en_{1}$};  
  \node[inv] () [below of = n61,node distance = 5mm] {$\figinv{b{=}0}$};    
  
  \node[boxloc](f1) at (1.5, 0.8) {$\begin{array}{c} F_1{:}\upAn{Z}{2}\\ \end{array}$};
  \node[port](f1n1) [left of=f1,node distance = 7mm] {}; 
  \node[port](f1x1) [right of=f1,node distance = 7mm] {};

  \node[boxloc](f2) at (3.5, 0.8) {$\begin{array}{c} F_2{:}\upAn{X}{4}\\ \end{array}$};
  \node[port](f2n1) [left of=f2,node distance = 7mm] {}; 
  \node[port](f2x1) [right of=f2,node distance = 7mm] {};

  \node[boxloc](f3) at (5.5, 0.8) {$\begin{array}{c} F_3{:}\upAn{Y}{2}\\ \end{array}$};
  \node[port](f3n1) [left of=f3,node distance = 7mm] {}; 
  \node[port](f3x1) [right of=f3,node distance = 7mm] {};

  \node[loc](x61) at (7, 0.8) {$ex_{1}$};
  \node[inv] () [below of = x61,node distance = 5mm] {$\figinv{b{=}0}$};  

   \draw[trans] (n61)-- (f1n1) {};
   \draw[trans] (f1x1)--(f2n1) {};
   \draw[trans] (f2x1)--(f3n1) {};
   \draw[trans] (f3x1)-- (x61)	;
   
%   %%%%%%%%%%%%%%%%%%%%%%%%%%%%%%%%%%%%%%%%%%%%%%%%%%%%%%%%%%%%%%%%%%%%%%%%%%%%%%% 
%   %%%%%%%%%%%%%%%%%%%%%%%%%%%%%%%%%% \Del %%%%%%%%%%%%%%%%%%%%%%%%%%%%%%%%%
% 
%    \draw(8, 0) rectangle (9.5,1.5);
%   \draw (9.3, 1.7) node {$\Del$};
%   
%   
%   \node[loc](n63) at (8, 0.8) {$\DEn$};  
%   \node[loc](x63) at (9.5, 0.8) {$\DEx$};
%     
%   \draw[trans] (n63)--(x63);
%    
%   %%%%%%%%%%%%%%%%%%%%%%%%%%%%%%%%%%%%%%%%%%%%%%%%%%%%%%%%%%%%%%%%%%%%%%%%%%%%%%% 
%   %%%%%%%%%%%%%%%%%%%%%%%%%%%%%%%%%%%%%%%%%%%%%%%%%%%%%%%%%%%%%%%%%%%%%%%%%%%%%%% 

   %%%%%%%%%%%%%%%%%%%%%%%%%%%%%%%%%%%%%%%%%%%%%%%%%%%%%%%%%%%%%%%%%%%%%%%%%%%%%%% 
   %%%%%%%%%%%%%%%%%%%%%%%%%%%%%%%%%% \chkEQ %%%%%%%%%%%%%%%%%%%%%%%%%%%%%%%%%
 
    \draw(8, 0) rectangle (10.5,1.5);
   \draw (10, 1.7) node {$\chkEQ{a_1}{a_2}$};

   \node[loc](n64) at (8, 0.8) {$\chkEQEn$};  
   \node[rate]() [below of =n64,node distance=5mm] {$\figrate{a_1,a_2}$};
   \node[rate]() [above of =n64,node distance=5mm] {$\figrate{\neg b}$};
   \node[loc](x64) at (10.5, 0.8) {$\chkEQEx$};
     
   \draw[trans] (n64)--(x64) node [midway, above]{$a_1{=}1$};
   \draw[trans] (n64)--(x64) node [midway, below]{$a_2{=}1$};
    
   %%%%%%%%%%%%%%%%%%%%%%%%%%%%%%%%%%%%%%%%%%%%%%%%%%%%%%%%%%%%%%%%%%%%%%%%%%%%%%% 
   %%%%%%%%%%%%%%%%%%%%%%%%%%%%%%%%%%%%%%%%%%%%%%%%%%%%%%%%%%%%%%%%%%%%%%%%%%%%%%% 
   
  %%%%%%%%%%%%%%%%%%%%%%%%%%%%%%%%%%%%%%%%%%%%%%%%%%%%%%%%%%%%%%%%%%%%%%%%%%%%%%% 
  %%%%%%%%%%%%%%%%%%%%%%%%%%%%%%%%%% Up(A,n) %%%%%%%%%%%%%%%%%%%%%%%%%%%%%%%%%%%%

  \draw(0, -2) rectangle (8.1,-0.5);
  \draw (6.1, -0.3) node {$\upAn{A}{n} {:} A \in \set{X,Y,Z}; n \in \set{2,4}$};

  \node[loc](n62) at (0, -1.1) {$\upAnEn$}; 
  \node[inv] () [below of = n62,node distance = 5mm] {$\figinv{b{=}0}$};

  \node[loc](m1) at (1.3, -1.1) {$m_1$};
  \node[rate] () [below of = m1, node distance = 5mm] {$\figrate{\varbar{A{-}\set{a_2}}}$};
  \node[rate] () [above of = m1, node distance = 5mm] {$\neg b$};
  
  \node[boxloc](f5) at (3.1, -1.1) {$\begin{array}{c} F_5{:}\chkAn{A}{n}\end{array}$};
  \node () [below of = f5, node distance = 6mm] {$(\variables)$};
  \node[port](f5n1) [left of=f5,node distance = 7.5mm] {}; 
  \node[port](f5x1) [right of=f5,node distance = 7.5mm] {};

  \node[loc](m2) at (4.9, -1.1) {$m_2$};
  \node[rate] () [below of = m2, node distance = 5mm] {$\figrate{a_2}$};
  \node[rate] () [above of = m2, node distance = 5mm] {$\neg b$};

  \node[boxloc](f7) at (6.6, -1.1) {$\begin{array}{c} F_7{:}\chkEQ{a_1}{a_2}\end{array}$};
  \node () [below of = f7, node distance = 6mm] {$(\variables)$};
  \node[port](f7n1) [left of=f7,node distance = 7.5mm] {}; 
  \node[port](f7x1) [right of=f7,node distance = 7.5mm] {};

  \node[loc](x62) at (8.1, -1.1) {$\upAnEx$};
  \node[inv] () [below of = x62,node distance = 5mm] {$\figinv{b{=}0}$};

  \draw[trans] (n62)--(m1);
  \draw[trans] (m1)--(f5n1);
  \draw[trans] (f5x1)--(m2);
  \draw[trans] (m2)--(f7n1);
  \draw[trans] (f7x1) -- (x62);

   \draw(0, -4) rectangle (6,-2.5);
   \draw (5.7, -2.3) node {$\chkAn{A}{2}$};

   \node[loc](n67) at (0, -3.2) {$\chkAtwoEn$};  
   \node[inv] () [below of = n67,node distance = 5mm] {$\figinv{b{=}0}$};

  \node[loc](f8) at (1.5, -3.2) {$n_1$};
   \node[rate] () [below of =f8,node distance=5mm]{$\figrate{a_1,a_2}$};
   \node[rate] () [above of =f8,node distance=5mm]{$\figrate{\neg b}$};

  \node[loc](f9) at (3.5, -3.2) {$n_2$};
   \node[rate] () [below of =f9,node distance=5mm]{$\figrate{a_3,a_2}$};
   \node[rate] () [above of =f9,node distance=5mm]{$\figrate{\neg b}$};

   \node[loc](x67) at (6, -3.2) {$\chkAtwoEx$};
   \node[inv] () [below of = x67,node distance = 5mm] {$\figinv{b{=}0}$};

   \draw[trans] (n67)--(f8);
   \draw[trans] (f8) -- (f9);  
   \draw[trans] (f9) -- (x67) node[midway,below]{$ a_3{=}1$} 
    node [midway, above]{$\begin{array}{c} a_1 {=} a_2 {=}1\end{array}$};
    
   %%%%%%%%%%%%%%%%%%%%%%%%%%%%%%%%%%%%%%%%%%%%%%%%%%%%%%%%%%%%%%%%%%%%%%%%%%%%%%% 
   %%%%%%%%%%%%%%%%%%%%%%%%%%%%%%%%%%%%%%%%%%%%%%%%%%%%%%%%%%%%%%%%%%%%%%%%%%%%%%% 
   
  %%%%%%%%%%%%%%%%%%%%%%%%%%%%%%%%%%%%%%%%%%%%%%%%%%%%%%%%%%%%%%%%%%%%%%%%%%%%%%% 
  %%%%%%%%%%%%%%%%%%%%%%%%%%%%%%%%%% zero check %%%%%%%%%%%%%%%%%%%%%%%%%%%%%%%%%

  \draw(7.5, -4) rectangle (10.5,-2.5);
  \draw (9.5, -2.3) node {$ZC {:} c=0?$};

  \node[loc](n65) at (7.5, -3.2) {$en_{5}$};  
   \node[rate] () [below of =n65,node distance=5mm]{$\figrate{z_1,x_1}$};
   \node[rate] () [above of =n65,node distance=5mm]{$\figrate{\neg b}$};
  
  \node[loc](x65yes) at (10.5, -2.8) {$ex_{5}$};
  \node[loc](x65no) at (10.5, -3.6) {$ex_{5}'$};

   \draw[trans] (n65)--(x65yes) node [rotate=11,midway, above]{$z_1=1 \wedge x_1=1$};
   \draw[trans] (n65)--(x65no) node [rotate=-11,midway, below]{$z_1=1 \wedge x_1 \not =1$};

\end{tikzpicture}
}
 \caption{Time bounded reachability in 14 stopwatch RSA: Increment $c$. Note that the variables which tick in a location are indicated below it. $b$ ticks everywhere except in locations where it is specified as $\neg b$. Also $\varbar{S}$ denoted stopwatches $\variables-S$. 
 }
 \label{fig_undec_14sw_GF_TB_v3}
\end{figure}
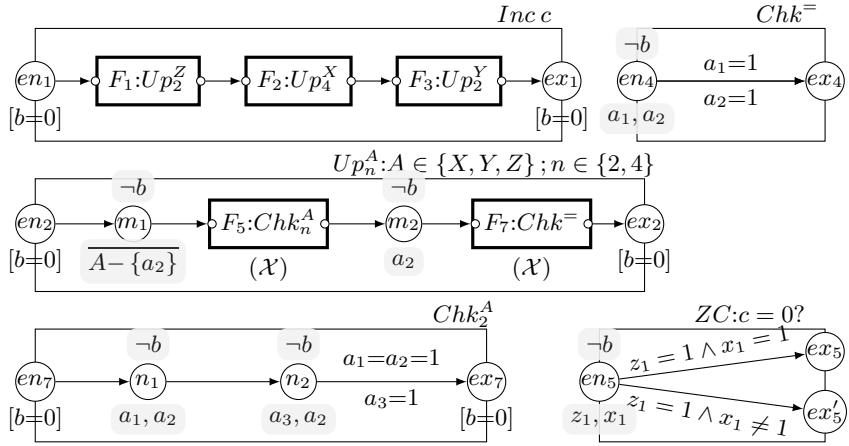

 \item We are then at the entry node of the subcomponent $F_5{:}Chk_2^Z$ with values
 $\nu(z_2)=1-\beta$,  and $\nu(z_i)=1-\beta+t_1$ for $i=1,3$, 
 $\nu(X)=1-\beta_c$, $\nu(Y)=1-\beta_d$ and $\nu(b)=0$.
 $Chk_2^Z$ is called by passing all stopwatches by value. 
 The subcomponent $Chk_2^Z$ ensures that $t_1=\frac{\beta}{2}$.
 
 \item To ensure $t_1=\frac{\beta}{2}$, at the entry node $en_7$ of $Chk_2^Z$, no time can elapse.
 If $t_2$ and $t_3$ are times elapsed in $n_1$ and $n_2$ then, upon reaching exit node $en_7$, we have $z_1=1-\beta+t_1+t_2$, $z_2=1-\beta+t_2+t_3$ and $z_3=1-\beta+t_1+t_3$. Additionally, $z_1=z_2=z_3=1$ implies $t_1+t_2=\beta=t_2+t_3=t_1+t_3$. 
  Thus, we get $t_1=t_2=t_3=\frac{\beta}{2}$. 
  Thus, the total time spent in $Chk_2^Z$ is $t_2+t_3=\beta$.
 
 \item At the return port of $F_5{:}Chk_2^Z$, we 
    restore all values to what they were, at 
    the call port of $F_5{:}Chk_2^Z$. That is, $\nu(z_2)=1-\beta$,  and $\nu(z_i)=1-\beta+t_1$ for $i=1,3$, 
    with the guarantee that $t_1=\frac{\beta}{2}$. 
    
    A time $t_4$ is elapsed in location $m_2$ affecting only $z_2$ to become $z_2=1-\beta+t_4$.

\item Finally, we call the subcomponent $Chk^{=}$ with $z_1=1-\frac{\beta}{2}$ and $z_2=1-\beta+t_4$.  All stopwatches are passed by value.
$Chk^{=}$ ensures that $z_1=z_2$; that is, $t_4=\frac{\beta}{2}$. 
A time $t_5=\frac{\beta}{2}$ is spent at the entry node $en_4$ of $Chk^{=}$ to ensure this. 
Thus, 
at the return port of $F_7{:}Chk^{=}$, we have $z_1=z_2=z_3=1-\frac{\beta}{2}$, and the rest of the stopwatches unchanged.
No time can elapse at the exit node $ex_2$ of $Up_2^Z$. Thus, at the return port of $F_1:Up_2^Z$, we get $\nu(Z)=1-\frac{1}{2^{k+1}}$.   
\item The total time elapsed in $F_1:Up_2^Z$ is $t_1+t_2+t_3+t_4+t_5=\frac{\beta}{2}+\beta+\frac{\beta}{2}+\frac{\beta}{2}=\frac{5\beta}{2}$.
 \end{enumerate}
At the return port of $F_2:Up_2^Z$, we thus have 
$\nu(Z)=1-\frac{\beta}{2}$, $\nu(X)=1-\beta_c$, and $\nu(Y)=1-\beta_d$. No time elapses 
here, and we are at the call port of $F_2:Up_4^X$. 
The component $\chkAn{X}{4}$ is similar to $\chkAn{Z}{2}$. It has 4 locations $h_1,h_3,h_4,h_5$ inside it, each $h_i$ has stopwatches $x_2$ and $x_i$ ticking. An analysis similar to the above gives that 
the total time elapsed in $Up_4^X$ is $\frac{11 \beta}{4}$, and at the return port 
of $F_2:Up_4^X$, we get $\nu(X)=1-\frac{\beta_c}{4}$, $\nu(Y)=1-\beta_d$ and 
$\nu(Z)=1-\frac{\beta}{2}$. This is followed by entering $F_3:Up_2^Y$, 
with these values. 
At the return port of  $F_3:Up_2^Y$, we obtain 
$\nu(X)=1-\frac{\beta_c}{4}$, $\nu(Y)=1-\frac{\beta_d}{2}$ and 
$\nu(Z)=1-\frac{\beta}{2}$, with the total time elapsed 
in $Up_2^Y$ being $\frac{5\beta_d}{4}$. 

From the explanations above, the following propositions 
can be proved. The same arguments given above will apply to prove this.

\begin{proposition}
 For any box $B$ and context $\langle \kappa \rangle$, and $\nu(Z)=1-\beta$, 
 we have that $(\langle \kappa \rangle, (B,en), (\nu(X), \nu(Y),1-\beta, \nu(b))) 
 \stackrel [\upAn{Z}{2}]{*}{\longrightarrow} (\langle \kappa \rangle, (B,ex), (\nu(X), \nu(Y),1-\frac{\beta}{2}, \nu(b)))$.
\end{proposition}

\begin{proposition}
 For any box $B$ and context $\langle \kappa \rangle$, and $\nu(X)=1-\beta_c$, 
 we have that $(\langle \kappa \rangle, (B,en), (1-\beta_c, \nu(Y), \nu(Z), \nu(b))) 
 \stackrel [\upAn{X}{4}]{*}{\longrightarrow} (\langle \kappa \rangle, (B,ex), (1-\frac{\beta_c}{4}, \nu(Y), \nu(Z), \nu(b)))$.
\end{proposition}

\begin{proposition}
 For any box $B$ and context $\langle \kappa \rangle$, and $\nu(Y)=1-\beta_d$, 
 we have that $(\langle \kappa \rangle, (B,en), (\nu(X),1-\beta_d, \nu(Z),\nu(b))) 
 \stackrel [\upAn{Y}{2}]{*}{\longrightarrow} (\langle \kappa \rangle, (B,ex), (\nu(X),1-\frac{\beta_d}{2}, \nu(Y), \nu(b)))$.
\end{proposition}

\noindent{\bf Simulate Decrement Instruction}: 
Assume that the $(k+1)$st instruction is decrementing counter $c$. Then 
we construct the main component $Dec~c$
 similar to the component $Inc~c$ above. The main change is the following:
 \begin{itemize}
 \item The main component $Dec~c$ will have the subcomponents 
   $Up_2^Z$ and  $Up_2^Y$ lined up sequentially. The functionality 
   $Up_2^Z$ and  $Up_2^Y$ are the same as the one in $Inc~c$.
  The total time spent in $Dec~c$ is also, less than $8 \beta$.
  \end{itemize}

\noindent{\bf Zero Check Instruction}:
The main component for Zero Check  follows the same pattern as $\incc$.
The main change is the following:
 \begin{itemize}
 \item The main component $ZeroCheck$ will have the subcomponents 
   $Up_2^Z$ and $Up_2^X$ and $Up_2^Y$ lined up sequentially. The functionality 
   $Up_2^Z,Up_2^X$ and  $Up_2^Y$ are the same as the one in $Inc~c$. 
   After these three, we invoke the subcomponent $ZC:c=0?$ shown in Figure \ref{fig_undec_14sw_GF_TB_v3}. 
   $ZC:c=0?$ is called by passing all stopwatches by value. At the entry node $en_5$ of 
   $ZC$, $z_1$ and $x_1$ are ticking. At $en_5$ if a time elapse $\frac{\beta}{2}=\frac{1}{2^{k+1}}$ makes $z_1=x_1=1$, then 
   we know that $x_1$ on entry was $1-\frac{1}{2^{c+k+1}}= 1-\frac{1}{2^{k+1}}$ 
   impying that $c=0$. Likewise, if $z_1$ attains 1, but $x_1$ does not, then $c \neq 0$.  
 Since all stopwatches are passed by value, at the return port of $ZC$ within the main component
 $ZeroCheck$, we regain back the stopwatch values obtained 
 after going through $Up_2^Z,Up_2^X$ and  $Up_2^Y$: that is, 
 $\nu(Z)=1-\frac{\beta}{2}$, $\nu(X)=1-\frac{\beta_c}{2}$ and $\nu(Y)=1-\frac{\beta_d}{2}$. 
  The time elapsed in $ZC$ is $\frac{\beta}{2}$. The times elapsed in   
     $Up_2^Z$ and $Up_2^X$ and $Up_2^Y$ are same as calculated 
     in the case of Increment $c$. Thus, the total time elapsed here is $< 8\beta+\beta=9\beta$.
           \end{itemize}
  
 We conclude by calculating the total time elapsed during the entire simulation. 
 We have established so far that for the $(k+1)$th instruction, the time elapsed is no more than $9 \beta$, 
 for $\beta=\frac{1}{2^k}$.  For the first instruction, the time elapsed is at most $9 $, for the second instruction it is 
 $\frac{9}{2}$, for the third it is $\frac{9}{2^2}$ and so on. 
\begin{equation}
  \text{Total~time~duration} = \left \{
 \begin{array}{l}
 9 (1 + \frac{1}{2} +\frac{1}{4} + \frac{1}{8} + \frac{1}{16} + \cdots)\\
  9 ( 1 + (<1))\\
  < 18~t.u
 \end{array}
 \right.
\end{equation}

The proof that we reach the vertex $Halt$ of the RHA iff the two counter machine halts 
follows: Clearly, the exit node of each main component iff the corresponding instruction is simulated correctly. Thus, if the counter machine halts, we will indeed reach the exit node of the main component corresponding 
to the last instruction. However, if the machine does not halt, then 
we keep going between the various main components simulating each instruction, and never reach $Halt$.

\qed
\end{proof}

%  \end{document}

% SHORT VERSIONS
% \input{rha-2-short.tex}
% \input{rha-3-short.tex}  
% \input{tb-rta-5-short.tex}
% \input{tb-rha-14-short.tex}

\section{Undecidability Resuls with two players}
\label{sec:undec-rhg}
 For the undecidability results for reachability games, we construct a recursive automaton (timed/hybrid)
as per the case, whose main components are the modules for the instructions and 
the counters are encoded in the variables of the automaton.  
In these reductions,  the reachability of the exit node of each component
corresponding to an instruction is linked to a faithful simulation of various
increment, decrement and zero check  instructions of the machine by choosing
appropriate delays to adjust the clocks/variables, to reflect changes  in
counter values. 
We specify a main component for each type instruction of the two counter machine, for example $\Hh_{inc}$ for increment. 
The entry node and exit node of a main component $\Hh_{inc}$ corresponding to an 
instruction [$\ell_i:  c := c +1$;  goto  $\ell_k$] 
are respectively $\ell_i$ and $\ell_k$.  
 Similarly, a main component corresponding to a zero check instruction 
[$l_i$: if $(c >0)$ then goto $\ell_k$]
  else goto $\ell_m$, has a unique entry node  $\ell_i$, and two exit nodes corresponding to $\ell_k$ and 
  $\ell_m$ respectively. 
The various main components corresponding to the various instructions, when connected appropriately, gives the higher level component $\Hh_{M}$ and this completes the RHA $\Hh$. 
The entry node of $\Hh_{M}$ is the entry node of the main component for the
first instruction of $M$ and the exit node is $Halt$. 
\ach simulates the machine while \tort verifies the simulation. Suppose in each main component for each type of instruction correctly \ach simulates the instruction by accurately updating the counters encoded in the variables of $\Hh$. Then, the unique run in $M$ corresponds to an unique run in $\Hh_{M}$. The halting  problem of the two counter machine now boils down to existence of a \ach strategy to ensure the reachability of an exit node $Halt$ (and $\ddot \smile$) in $\Hh_{M}$.

\subsection{Time Bounded Reachability games in Unrestricted RTA with 3 clocks}
\label{undec:rtg_3clk}

\begin{lemma}
\label{3clk}
The time bounded reachability game problem is undecidable for recursive timed automata 
  with at least 3 clocks.
\end{lemma}
\begin{proof}
We prove that the reachability problem is undecidable for  unrestricted RTA with 3 clocks.  
In order to obtain the undecidability result, we 
use a reduction from the halting problem 
for two counter machines. 
Our reduction uses a RTA with three clocks $x,y,z$.

We specify a main component for each instruction of the two counter machine. On entry into 
a main component for increment/decrement/zero check, 
we have $x=\frac{1}{2^{k+c}3^{k+d}}$, $y=\frac{1}{2^k}$ and $z=0$,
where $c,d$ are the current values of the counters and $k$ is the current instruction. Note that $z$ is used only to enforce urgency in several vertices. 
Given a two counter machine, we build a 3 clock RTA 
whose building blocks are the main components for the instructions.   
The purpose of the components is to simulate faithfully the counter machine 
by choosing appropriate delays  to adjust the variables to reflect changes 
in counter values. 
On entering the entry node $en$ of a main component corresponding to an instruction $l_i$,
we have the configuration $(\langle \epsilon \rangle, en, (\frac{1}{2^{k+c}3^{k+d}},\frac{1}{2^{k}},0))$ 
of the three clock RTA.  

We shall now present the components for increment/decrement and zero check instructions.
In all the components, the variables passed by value are written below the boxes and the invaraints of the 
locations are indicated below them.

\noindent{\bf{Simulate increment instruction}}: Lets consider the increment instruction  
$\ell_i$: $c=c+1$; goto $\ell_k$. The component for this instruction is component $Inc~c$ given in Figure 
\ref{fig_3clk_rtg_inc}. Assume that $x=\frac{1}{2^{k+c}3^{k+d}}$, $y=\frac{1}{2^k}$ and $z=0$ at 
the entry node $en_1$ of the component $Inc~c$. To correctly simulate the increment of counter $c$, the clock values at the exit node $ex_1$ should be $x=\frac{1}{2^{k+c+2}3^{k+d+1}}$, $y=\frac{1}{2^{k+1}}$ and $z=0$. 

Let $\alpha=\frac{1}{2^{k+c}3^{k+d}}$ and $\beta=\frac{1}{2^k}$.  We want  $x = \frac{\alpha}{12}$ and $y = \frac{\beta}{2}$ at $ex_1$. We utilise the component $Div\{a,n\}$ (with $a=x,n=12$ and $a=y,n=2$) to perform these divisions. 
Lets walk through the working of the component $Inc~c$. As seen above, at the entry node $en_1$, we have 
$x=\frac{1}{2^{k+c}3^{k+d}}$, $y=\frac{1}{2^k}$ and $z=0$. 

\begin{enumerate}
\item No time is spent at $en_1$ due to the invariant $z=0$. 
$Div\{y,2\}$ is called, passing $x,z$ by value.  
At the call port  of $A_1:Div\{y,2\}$, we have the same values of $x,y,z$. Let us examine the component $Div\{y,2\}$. 
We instantiate $Div\{a,n\}$ with $a=y,n=2$. Thus, the clock referred to as $b$ 
in $Div\{a,n\}$ is $x$ after the instantiation. 
At the entry node $en_2$ of $Div\{y,2\}$, no time is spent due to the invariant $z=0$; we have 
$a=y=\beta,b=x=\alpha,z=0$.  
Resetting $b=x$, we are at the call port of $A_3:D$. 
$A_3$ is called, passing $a,z$ by value. 
A nondeterministic time $t$ is spent at the entry node $en_3$ of $D$. 
 Thus, at the return port of $A_3$, we have $a=y=\beta,b=x=t,z=0$. 
The return port of $A_3$ is a node belongining to \tort; for \ach to 
reach $\ddot \smile$, $t$ must be $\frac{\beta}{2}$. \tort has 
two choices to make at the return port of $A_3$: he can continue the simulation, 
by resetting $a=y$ and going to the call port of $A_5:D$, 
or he can verify if $t$ is indeed $\frac{\beta}{2}$, by going to the call port of $A_4$. 

\begin{itemize}
\item Assume \tort goes to the call port of $A_4: C^{x=}_{y/2}$ (recall, that by the instantiation, $b=x$, $a=y$ and 
$n=2$). $z$ is passed by value. At the entry node $en_5$ of  
$C^{x=}_{y/2}$, no time elapses due to the invariant $z=0$. 
Thus, we have $x=b=t$, $a=y=\beta,z=0$ at $en_5$. 
The component $A_7:M_x$ is 
invoked, passing $x,z$ by value. At the entry node $en_6$ of $M_b$, a time $1-t$ 
is spent, giving $a=y=\beta+1-t$, $b=x=t$ and $z=0$ at the return port of $A_7$. 
Since $n=2$, one more invocation of $A_7: M_x$ is made, obtaining 
$a=y=\beta+2(1-t)$, $b=x=t$ and $z=0$ at the return port 
of $A_7$ after the second invocation. To reach the 
exit node $ex_5$
of $C^{x=}_{y/2}$, $a$ must be exactly 2, since 
no time can be spent at the return port of $A_7$; this is so since the invariant $z=0$ at the exit node $ex_5$ of 
$C^{x=}_{y/2}$ is satisfied only when no time 
is spent at the return port of $A_7$.  
If $a$ is exactly 2, we have 
 $\beta=2t$. In this case, 
from the return port of $A_4$,  
$\ddot \smile$ can be reached.  
 
 \item Now consider the case that \tort moves ahead from 
   the return port of $A_3$, resetting $a=y$ to the call port of $A_5:D$.
 The values are $a=y=0,b=x=t=\frac{\beta}{2}$ and $z=0$. 
   $A_5:D$ is invoked passing $b=x$ and $z$ by value. 
   A non-deterministic amount of time $t'$ is spent at the entry node $en_3$ of 
   $D$, giving $a=y=t'$, $b=x=\frac{\beta}{2}$ and $z=0$ 
   at the return port of $A_5$.  Again, the return port of $A_5$ is a node 
   belonging to \tort. Here \tort, thus has two choices: he can continue 
   with the simulation going to $ex_2$, or can 
   verify that $t'=\frac{\beta}{2}$ by going to the call port of $A_6:C^{y=}_x$.
   $C^{y=}_x$ is a component that checks if $y$ has ``caught up'' with $x$; that is, 
   whether $t'=t=\frac{\beta}{2}$. At the entry node $en_4$ of $C^{y=}_x$,
   $a$ and $b$ can simultaneeously reach 1 iff $t=t'$; that is, $t'=\frac{\beta}{2}$.
   Then, from the return port of $A_6$, we can reach $\ddot \smile$.   
   
  \item Thus, we reach $ex_2$ with $x=y=\frac{\beta}{2},z=0$. At the return port of 
  $A_1:Div\{y,2\}$, we thus have $x=\alpha, y=\frac{\beta}{2},z=0$.       
\end{itemize}
\item From the return port of $A_1:Div\{y,2\}$, we reach the call port of $A_2:Div\{x,12\}$.
$y,z$ are passed by value. The functioning of $A_2$ is similar to that of $A_1$:
at the return port of $A_1$, we obtain $x=\frac{\alpha}{12}$, $y=\frac{\beta}{2}$ and $z=0$. 
\end{enumerate}

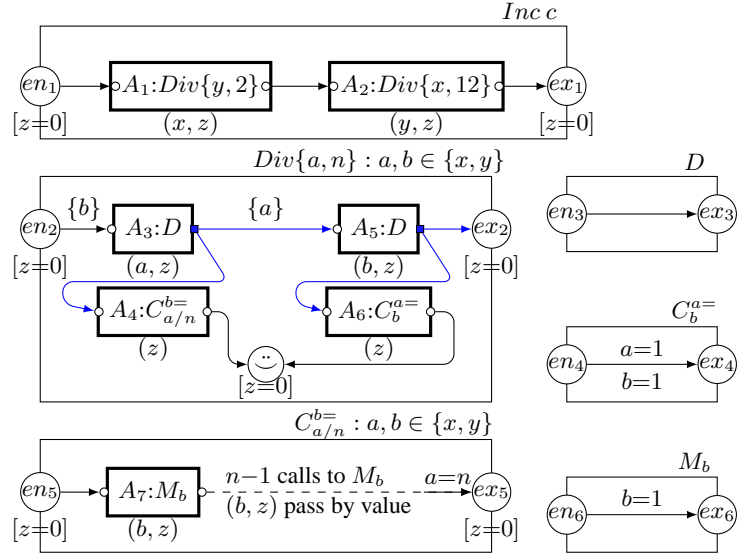
\begin{figure}
\begin{center}
\begin{tikzpicture}[node distance=4cm] 
\draw(0, 0) rectangle (7,1.5);
\draw (6.5, 1.7) node {$\incc$};

  \node[loc](n1) at (0, 0.7) {$en_{1}$};  
  \node[inv]() [below of =n1, node distance=5mm]{$\figinv{z{=}0}$};
  
  \node[boxloc](a1) at (2, 0.7) {$\begin{array}{c} A_1{:}\Div{y}{2}\end{array}$};
  \node[varpass] () [below of =a1, node distance=5mm]{$(x,z)$};
  \node[port](a1n1) [left of=a1,node distance = 10mm] {}; 
  \node[port](a1x1) [right of=a1,node distance = 10mm] {};

  \node[boxloc](a2) at (5, 0.7) {$\begin{array}{c} A_2{:}\Div{x}{12}\end{array}$};
  \node[varpass] () [below of =a2, node distance=5mm]{$(y,z)$};
  \node[port](a2n1) [left of=a2,node distance = 11mm] {}; 
  \node[port](a2x1) [right of=a2,node distance = 11mm] {};
  
  \node[loc](x1) at (7, 0.7) {$ex_{1}$};
  \node[inv]() [below of =x1, node distance=5mm]{$\figinv{z{=}0}$};
  
  \draw[trans] (n1)--(a1n1);
  \draw[trans] (a1x1) -- (a2n1);
  \draw[trans] (a2x1) -- (x1);
%%%%%%%%%%%%%%%%%%%%%%%%%%%%%%%%%%%%%%%%%%%%%%%%%%%%%%%%%%%%%%%%%%%%%%%%%%%%%%%%%%%%
%%%%%%%%%%%%%%%%%%%%%    div(a,n)  %%%%%%%%%%%%%%%%%%%%%%%%%%%%%%%%%%%%%%%%%%%%%%%%%

  \draw(0, -3.5) rectangle (6,-0.5);
  \draw (4.5, -0.3) node {$\Div{a}{n}:a,b\in\set{x,y}$};

  \node[loc](n2) at (0, -1.2) {$en_{2}$};  
  \node[inv]() [below of =n2,node distance=5mm] {$\figinv{z{=}0}$};  
  
  \node[boxloc](a3) at (1.5, -1.2) {$\begin{array}{c} A_3{:}\Del\end{array}$};
  \node[varpass] () [below of =a3, node distance=5mm]{$(a,z)$};
  \node[port](a3n1) [left of=a3,node distance = 5.5mm] {}; 
   \node[oport](a3x1) [right of=a3,node distance =5.5mm] {};  
  
  \node[boxloc](a4) at (1.5, -2.3) {$\begin{array}{c} A_4{:}\chkDiv{b}{a}{n}\end{array}$};
  \node[varpass] () [below of =a4, node distance=5mm]{$(z)$};
  \node[port](a4n1) [left of=a4,node distance = 7.2mm] {}; 
  \node[port](a4x1) [right of=a4,node distance = 7.2mm] {}; 
  
  \node[boxloc](a5) at (4.5, -1.2) {$\begin{array}{c} A_5{:}\Del\end{array}$};
  \node[varpass] () [below of =a5, node distance=5mm]{$(b,z)$};
  \node[port](a5n1) [left of=a5,node distance = 5.5mm] {}; 
    \node[oport](a5x1) [right of=a5,node distance =5.5mm] {};  
  
  \node[boxloc](a6) at (4.5, -2.3) {$\begin{array}{c} A_6{:}\chkEQn{a}{b}\end{array}$};
  \node[varpass] () [below of =a6, node distance=5mm]{$(z)$};
  \node[port](a6n1) [left of=a6,node distance = 6.9mm] {}; 
   \node[port](a6x1) [right of=a6,node distance = 6.9mm] {};
   
   \node[loc] (x2) at (6,-1.2) {$ex_2$};
     \node[inv]() [below of =x2,node distance=5mm] {$\figinv{z{=}0}$};  
  
   \node[loc] (x2s) at (3,-3) {$\ddot \smile$};
     \node[inv]() [below of =x2s,node distance=3mm] {$\figinv{z{=}0}$};  
  
   \draw[trans] (n2)--(a3n1)  node [midway, above]{$\set{b}$};
   \draw[otrans] (a3x1)--(2.5,-1.9)--(2,-1.9)--(0.3,-1.9)--(0.3,-2.2)--(a4n1);
   \draw[otrans] (a3x1)--node [midway, above]{\textcolor{black}{$\set{a}$}}(a5n1);
   \draw[otrans] (a5x1)--(5.5,-1.9) -- (3.4,-1.9)--(3.4,-2.2)--(a6n1);
   \draw[otrans] (a5x1)--(x2);
%    \draw[trans] (a6x1)--(x2s);
%    \draw[trans] (a4x1)-- +(0.5,0)-- +(0.5,-0.8) -- +(2.5,-0.8)-- (x2s);
   \draw[trans] (a4x1) -- (2.5,-2.3)-- (2.5,-3) -- (x2s);
   \draw[trans] (a6x1) -- (5.5,-2.3) -- (5.5,-3) -- (x2s);
   
%%%%%%%%%%%%%%%%%%%%%%%%%%%%%%%%%%%%%%%%%%%%%%%%%%%%%%%%%%%%%%%%%%%%%%%%%%%%%
%%%%%%%%%%%%%%%%%%%%%%%%%%%% \del %%%%%%%%%%%%%%%%%%%%%%%%%%%%%%%%%%%%%%%%%%%

  \draw(7, -1.5) rectangle (9,-0.5);
  \draw (8.7, -0.3) node {$\Del$};
  
  \node[loc] (n3) at (7,-1) {$en_3$};
  \node[loc] (x3) at (9,-1) {$ex_3$};
  
  \draw[trans] (n3)--(x3);

%%%%%%%%%%%%%%%%%%%%%%%%%%%%%%%%%%%%%%%%%%%%%%%%%%%%%%%%%%%%%%%%%%%%%%%%%%%%%
%%%%%%%%%%%%%%%%%%%%%%%%%%%% \chkEQn{a}{b} %%%%%%%%%%%%%%%%%%%%%%%%%%%%%%%%%%

  \draw(7, -3.5) rectangle (9,-2.5);
  \draw (8.7, -2.3) node {$\chkEQn{a}{b}$};
  
  \node[loc] (n4) at (7,-3) {$en_4$};
  \node[loc] (x4) at (9,-3) {$ex_4$};
  
  \draw[trans] (n4)--(x4) node[midway,above]{$a{=}1$} node[midway,below]{$b{=}1$};

%%%%%%%%%%%%%%%%%%%%%%%%%%%%%%%%%%%%%%%%%%%%%%%%%%%%%%%%%%%%%%%%%%%%%%%%%%%%%
%%%%%%%%%%%%%%%%%%%%%%%%%%%% \chkDiv{b}{a}{n} %%%%%%%%%%%%%%%%%%%%%%%%%%%%%%%%

  \draw(0, -5.5) rectangle (6,-4);
  \draw (4.7, -3.8) node {$\chkDiv{b}{a}{n}:a,b\in\set{x,y}$};

  \node[loc](n5) at (0, -4.7) {$en_{5}$};  
  \node[inv]() [below of =n5,node distance=5mm] {$\figinv{z{=}0}$};  
  
  \node[boxloc](a7) at (1.5, -4.7) {$\begin{array}{c} A_7{:}M_b\end{array}$};
  \node[varpass] () [below of =a7, node distance=5mm]{$(b,z)$};
  \node[port](a7n1) [left of=a7,node distance = 6.5mm] {}; 
   \node[port](a7x1) [right of=a7,node distance =6.5mm] {};  
  
  \node (a) at (3.7,-4.7) {$\begin{array}{l} \mbox{$n{-}1$ calls to $M_b$} \\ \mbox{$(b,z)$ pass by value} \end{array}$};
  
   \node[loc] (x5) at (6,-4.7) {$ex_5$};
   \node[inv]() [below of =x5,node distance=5mm] {$\figinv{z{=}0}$};  
  
   \draw[trans] (n5)--(a7n1);
   \draw[trans,dashed] (a7x1) -- (x5);
   \draw[trans] (a) -- (x5) node[midway,above] {$a{=}n$};
%    \draw[trans] (a7x1)-- (a);   
%    \draw[trans] (a)-- (x5); 

%%%%%%%%%%%%%%%%%%%%%%%%%%%%%%%%%%%%%%%%%%%%%%%%%%%%%%%%%%%%%%%%%%%%%%%%%%%%%
%%%%%%%%%%%%%%%%%%%%%%%%%%%% M_b %%%%%%%%%%%%%%%%%%%%%%%%%%%%%%%%%%

  \draw(7, -5.5) rectangle (9,-4.5);
  \draw (8.7, -4.3) node {$M_b$};
  
  \node[loc] (n6) at (7,-5) {$en_6$};
  \node[loc] (x6) at (9,-5) {$ex_6$};
  
  \draw[trans] (n6)--(x6) node[midway,above]{$b{=}1$};

\end{tikzpicture}

 \caption{Games on RTA with 3 clocks : Increment $c$.}
 \label{fig_3clk_rtg_inc}
 
\end{center}
\end{figure}

\noindent{\it Time taken}: 
Now we discuss the total time to reach a $\ddot \smile$ node   
or the exit node $ex_1$ of the component $Inc~c$ while simulating the increment instruction. 
 At the entry node $en_1$,
 clock values are $x=\frac{1}{2^{k+c}3^{k+d}}$, $y=\frac{1}{2^k}$ and $z=0$. 
 Let $\alpha = \frac{1}{2^{k+c}3^{k+d}}$ and $\beta = \frac{1}{2^{k}}$. The invaraint $z=0$ at the entry and the exit nodes $en_1$ and $ex_1$ ensures that no time elapses in these nodes and also in the return ports of $A_1$ and $A_2$. 
From the analysis above, it follows that at the return port of $A_1:\Div{y}{2}$,  $x=\alpha$, $y=\frac{\beta}{2}$ and $z=0$. Similarly at the return port of $A_2:Div\{x,12\}$, the clock values are $x=\frac{\alpha}{12}$, $y=\frac{\beta}{2}$ and $z=0$. Thus, counter $c$ has been incremented and the end of instruction $k$ has been recorded in $x$ and $y$. 
The time spent along the path from $en_1$ to $ex_1$ is the sum of times spent in $A_1:Div\{y,2\}$ and 
$A_2:Div\{x,12\}$. 
\begin{itemize}
\item Time spent in $A_1: Div\{y,2\}$. The time spent in $A_3:D$, as 
well as $A_5:D$ is both 
  $\frac{\beta}{2}$. Recall that \tort can verify that the times  $t,t'$ 
  spent in $A_3,A_5$ are both $\frac{\beta}{2}$.
If \tort enters $A_4$ to verify $t=\frac{\beta}{2}$, then the time taken is 
$2(1-t)$. In this case, the time taken to reach $\ddot \smile$  from the return port of $A_4$ is 
$t+ 2(1-t)=2-\frac{\beta}{2}$. Likewise, if \tort continued from $A_3$ to $A_5$, and goes on to verify that 
the time $t'$ spent in $A_5$ is also $\frac{\beta}{2}$, 
then the total time spent before reaching the $\ddot \smile$ from 
the return port of $A_6$ is 
  $t+t'+(1-t')=1+t=1+\frac{\beta}{2}$. Thus, 
if we are back at the return port of $A_1$, the time spent in $A_1$ is $t+t'=\beta$. 

\item Time spent in $A_2: Div\{x,12\}$. Here, the time spent in $A_3:D$ as well as $A_5:D$ is 
$\frac{\alpha}{12}$. In case \tort verifies that the time $t$ spent in $A_3:D$ is indeed $\frac{\alpha}{12}$,
then he invokes $A_4$. The time elapsed in $C^{y=}_{x/12}$ is $12(1-t)=12(1-\frac{\alpha}{12}) < 12$. 
Likewise, if \tort continued from $A_3$ to $A_5$, and goes on to verify that 
the time $t'$ spent in $A_5$ is also $\frac{\alpha}{12}$, 
then the total time spent before reaching the $\ddot \smile$ from 
the return port of $A_6$ is 
  $t+t'+(1-t')=1+t=1+\frac{\alpha}{12}$. Thus, 
if we are back at the return port of $A_2$, the time spent in $A_2$ is $t+t'=\frac{2\alpha}{12}$. 

\item In general, the component $Div\{a,n\}$ divides the value in clock $a$ by $n$. 
 If $a=\zeta$ on entering $\Div{a}{n}$, then upon exit, its value is $a=\frac{\zeta}{n}$.  The time taken to reach the exit $ex_2$ is $2 * (\frac{\zeta}{n})$. The time taken to reach the node $\ddot \smile$ in $\Div{a}{n}$ is $<n$ (due to $n$ calls to $M_b$ component).

\item Total time spent in $Inc~c$. Thus, if we come back to the return port of $A_2$, the total time spent is 
$\beta+\frac{2\alpha}{12} < 2\beta$, on entering with $y=\beta$. Recall that $x=\alpha=\frac{1}{2^{k+c}3^{k+d}}$ and $y=\beta=\frac{1}{2^k}$ and thus $\alpha \leq \beta$ always.
 \end{itemize}

From the analysis above, the following propositions are easy to see. 
%\begin{proposition}
 % \label{prop:3clk_rtg_div}
 % For any context $\kappa$, any box $b \in B$, and $x \in [0,1]$, we have that 
  % $(\langle \kappa \rangle, (b,en_2), (x,y,0)) \stackrel[\Div{y}{2}]{*}{\longrightarrow} (\langle \kappa \rangle, (b,ex_2), (x,\frac{y}{2},0))$.
  %\end{proposition}
 \begin{proposition}
  \label{prop:3clk_rtg_div}
  For any context $\kappa \in (B \times V)^*$, any box $b \in B$, and
  $x,y \in [0, 1]$, there exists a unique strategy of \ach such that
 % \begin{equation}
 %   \begin{array}{ll}
 %    (\langle \kappa \rangle, (b,en_2), (x,y=\beta,0)) \stackrel[\Div{y}{2}]{\beta}{\longrightarrow} (\langle \kappa \rangle, (b,ex_2), (x,\frac{\beta}{2},0)),\\
    
 %    \textcolor{red}{\text{ or } (\langle \kappa \rangle, (b,en_2), (x,y,0)) \stackrel[\Div{y}{2}]{2-\frac{\beta}{2}}{\longrightarrow} (\langle \kappa \rangle, (b,\ddot \smile), (\frac{\beta}{2},y,0)),}\\
    
 %    \textcolor{red}{\text{ or } (\langle \kappa \rangle, (b,en_2), (x,y,0)) \stackrel[\Div{y}{2}]{1+\frac{\beta}{2}}{\longrightarrow} (\langle \kappa \rangle, (b,\ddot \smile), (\frac{\beta}{2},\frac{\beta}{2},0))}.\\
 %    \textcolor{red}{\text{ DOUBT: summary-edge syntax thinks that $(b,en_2))$ is entry port of box $\Div{a}{n}$}}\\
 %    \textcolor{red}{\text{called in $Inc_c$ and $(b,\ddot \smile)$ is a return port of this box}}.
 %   \end{array}
 %  \end{equation}
 %  RE-WRITING
  \begin{equation}
   \begin{array}{ll}
    (\langle \kappa \rangle, (b,en_2), (x,y=\beta,0)) \stackrel[\Div{y}{2}]{\beta}{\longrightarrow} (\langle \kappa \rangle, (b,ex_2), (x,\frac{\beta}{2},0)),\\
    
    \text{ or } (\langle \kappa \rangle, (b,en_2), (x,y,0)) \stackrel[\Div{y}{2}]{2-\frac{\beta}{2}}{\longrightarrow} (\langle \kappa,(b,(x,y,0)) \rangle, \ddot \smile, (\frac{\beta}{2},y,0)),\\
    
    \text{ or } (\langle \kappa \rangle, (b,en_2), (x,y,0)) \stackrel[\Div{y}{2}]{1+\frac{\beta}{2}}{\longrightarrow} (\langle \kappa,(b,(x,y,0)) \rangle, \ddot \smile, (\frac{\beta}{2},\frac{\beta}{2},0)).\\
%    \text{because the $\ddot \smile$ is not a return port.}\\
%    \text{it is a node inside the component $\Div{a}{n}$ and hence when it is reached the contex has call to this box.}
   \end{array}
  \end{equation}
\end{proposition}
%\begin{proof}
 
%\qed \end{proof}

\begin{proposition}
  \label{prop:3clk_rtg_incc}
  For any context $\kappa \in (B \times V)^*$, any box $b \in B$, and $x,y \in [0,1]$, we have that \\
   $(\langle \kappa \rangle, (b,en_1), (x,y,0)) \stackrel[\incc]{<2\beta}{\longrightarrow} (\langle \kappa \rangle, (b,ex_1), (\frac{x}{12},\frac{y}{2},0))$. 
  \end{proposition}
  
  The proof essentially relies on the argument given above. Using summary edges, we can easily 
obtain the result.

\noindent{\bf{Simulate Zero check instruction:}} Let us now simulate the instruction $l_i$: if $(d >0)$ then goto $\ell_k$, else goto $\ell_j$.
Figure \ref{fig_3clk_rtg_zc} describes this. The component for this instruction is component $Zero~Check: d=0?$. 
Starting with $x=\frac{1}{2^{k+c}3^{k+d}}=\alpha, y=\frac{1}{2^k}=\beta$ and $z=0$ at the entry node $en_1$, 
we want to reach the node corresponding to $\ell_k$ if $d>0$, with 
$x=\frac{1}{2^{k+c+1}3^{k+d+1}}=\frac{\alpha}{6}, y=\frac{1}{2^{k+1}}=\frac{\beta}{2}$ and $z=0$, and 
to the node corresponding to $\ell_j$ if $d=0$, with the same clock values.
The nodes $d=0$ and $d>0$ are respectively the exit nodes of the component $Zero~Check: d=0?$. 
In the following, we analyse the zero check component in detail.
\begin{enumerate}
\item The first component invoked on entry is $A:Div\{y,2\}$ that records the $k+1$th instruction 
by dividing $y$ by 2. Clocks $x,z$ are passed by value. This component is the same as seen in the $Inc~c$ component. 
As seen there, at the return port of $A:Div\{y,2\}$, we obtain $x=\alpha,y=\frac{\beta}{2},z=0$.
A time of $\beta$ is spent in the process. Similarly, 
at the return port of $B:Div\{x,6\}$, we obtain $x=\frac{\alpha}{6}, y=\frac{\beta}{2}$ and $z=0$.
A total time of $\frac{2\alpha}{6}$ is elapsed 
in $B:Div\{x,6\}$. Thus, the total time spent on coming to the return port of $B:Div\{x,6\}$ is 
$\beta+\frac{2\alpha}{6} < 2 \beta$. 
\item At the return port of $B$, we goto the node $m$, elapsing 
no time. This is needed since the exit nodes of the zero check component have the invariant $z=0$.
At $m$, \ach guesses whether $d=0$ or not, and goes to one of $m_1,m_2$. 
Both these nodes belong to \tort. At both $m_1,m_2$, \tort has two choices:
he can go to an exit node of the zero check component, or choose to verify the correctness 
of the guess of \ach. The $\ddot \smile$ node is reachable from the upper component $B_1:ZC_{=0}^d$
 if $d=0$, while $\ddot \smile$ node is reachable from the lower component $B_1:ZC_{=0}^d$
  if $d>0$. Lets now look at the component $ZC_{=0}^d$.
  \item At the entry node $en_2$, we have $x=\frac{\alpha}{6}, y=\frac{\beta}{2}$ and $z=0$.
  To check if $d=0$, we first eliminate the $k$ from $x,y$, obtaining 
  $x=6^{k+1}.\frac{\alpha}{6}=\frac{1}{2^c3^d}$ and $y=2^{k+1}.\frac{1}{2^{k+1}}=1$. 
  The component $B_3$ multiplies $y$ by 2 once, and invokes $B_4$, which multiplies $x$ by 6;
  this is repeated until $y$ becomes 1. $B_3$ is invoked passing $x,z$ by value, while 
  $B_4$ is invoked passing $y,z$ by value. Lets examine the functioning of $Mul\{y,2\}$, 
  the functioning of $Mul\{x,6\}$ is similar. 
  \item  At the entry node $en_3$ of $Mul\{a,n\}$, 
  with $a=y,n=2$ and $b=x$, 
  we have $z=0, x=\frac{\alpha}{6}, y=\frac{\beta}{2}$.
  Resetting $b=x$, we goto the call port of $B_6:D$. 
  $D$ is called passing $y,z$ by value.  
  A non-deterministic time $t$ is spent at the entry node $en_5$ of $D$; 
  thus, at the return port of $B_6$, we have $b=x=t,a=y=\frac{\beta}{2},z=0$.  
  The time $t$ must be $\beta$; \tort can verify this
  by invoking $B_7:C_{2*y}^{x=}$. 
  $C_{2*y}^{x=}$  invokes $M_y$ two times, passing $y,z$ by value. 
  Each time,  in $M_y$, a time of $1-\frac{\beta}{2}$ is spent.
  After the two invoctions, we obtain $b=x=t+2(1-\frac{\beta}{2})$, 
  $a=y=\frac{\beta}{2}$ and $z=0$. This $b$ must be exactly 2 to reach the exit node $ex_4$ 
  of $C_{2*y}^{x=}$; this is possible iff $t=\beta$. In this case, 
  \tort will allow \ach to goto the $\ddot \smile$ node from the return port of $B_6:D$. 
  If \tort skips the verification and goes directly to $B_8:D$, 
  then at the call port of $B_8:D$, we have $b=x=\beta$, $a=y=0$ and $z=0$. 
  $D$ is called by passing $b=x,z$ by value. 
  A time $t'$ is elapsed in $D$, obtaining $b=x=\beta$, $a=y=t'$ and $z=0$. 
  This $t'$ must be exactly $\beta$; \tort can verify this by 
  invoking $B_9:C_x^{y=}$, at the return port of $B_8$. 
   $C_x^{y=}$ checks if $y$ has ``caught up'' with $x$; that is, if $y$ is also $\beta$.
   Clearly, the exit node $ex_6$ is reached iff $a=b$; that is, $t'=\beta$. 
  At the return port of $B_8$, we thus have $a=b=\beta$, $z=0$. 
  Back at the return port of $B_3$, we thus obtain $x=\frac{\alpha}{6}, y=\beta,z=0$.
    \item In a similar way, $Mul\{x,6\}$ multiplies $x$ by 6. This, at the return port of $B_4$,
  $y$ is mutiplied by 2 and $x$ by 6, once. The process repeats until 
  we obtain $y=1$ at the return port of $B_4$. At this time, 
  we know that the loop has happened $k+1$ times, that is, 
  $y=1$ and $x=\frac{1}{2^c3^d}$. 
   \item Now we can check if $d$ is zero or not, by multiplying $x$ by 2
   $c+1$ times. If $x$ becomes exactly 2 at sometime, then clearly $d$ is zero; 
   otherwise, $x$ will never become exactly 2. Then 
   the only option is to goto the exit node $d>0$ of $ZC_{=0}^d$.
    If \ach had guessed corerctly that $d=0$ and gone to node $m_1$, 
    in the zero check component, then $ZC^d_{=0}$ will reach the upper exit node $d=0$;
    From this return port of $B_1$, $\ddot \smile$ is reachable. Similar is the case when \ach guesses correctly 
    that $d>0$ at $m$.
 \end{enumerate}

 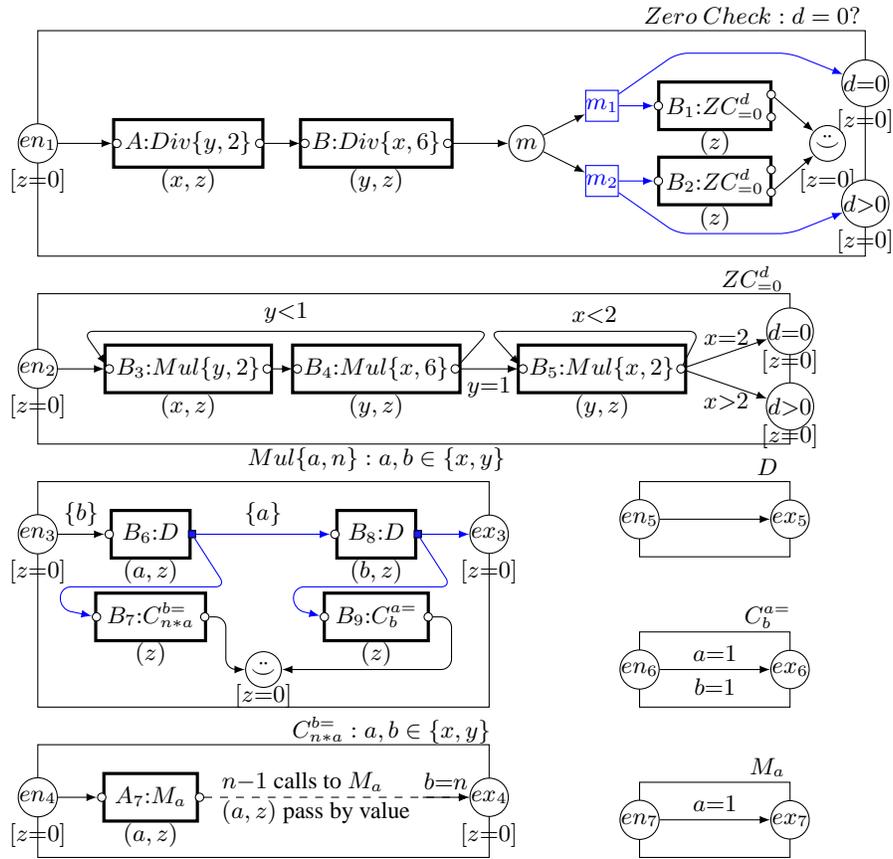
\begin{figure}
\begin{center}
\begin{tikzpicture}[node distance=4cm] 

\draw(0, 0) rectangle (11,3);
\draw (9.5, 3.2) node {$Zero~Check:d=0?$};

  \node[loc](n1) at (0, 1.5) {$en_{1}$};  
  \node[inv]() [below of =n1, node distance=5mm]{$\figinv{z{=}0}$};

  \node[boxloc](a1) at (2, 1.5) {$\begin{array}{c} A{:}\Div{y}{2}\end{array}$};
  \node[varpass] () [below of =a1, node distance=5mm]{$(x,z)$};
  \node[port](a1n1) [left of=a1,node distance = 9.5mm] {}; 
  \node[port](a1x1) [right of=a1,node distance = 9.5mm] {};

  \node[boxloc](a2) at (4.5, 1.5) {$\begin{array}{c} B{:}\Div{x}{6}\end{array}$};
  \node[varpass] () [below of =a2, node distance=5mm]{$(y,z)$};
  \node[port](a2n1) [left of=a2,node distance = 9.5mm] {}; 
  \node[port](a2x1) [right of=a2,node distance = 9.5mm] {};

  \node[loc](m) at (6.5, 1.5) {$m$};    
  
  \node[oloc] (m1) at (7.5,2){$m_1$};
  
  \node[boxloc](b1) at (9,2) {$B_1{:}\zchk{d}$};
  \node[varpass] () [below of =b1, node distance=5mm]{$(z)$};
   \node[port](b1n1) [left of=b1,node distance = 7.5mm] {}; 
   \node[port](b1x1s) at (9.75,2.15) {};
    \node[port](b1x1n) at (9.75,1.85) {};

  \node[oloc] (m2) at (7.5,1){$m_2$};
  
  \node[boxloc](b2) at (9,1) {$B_2{:}\zchk{d}$};
  \node[varpass] () [below of =b2, node distance=5mm]{$(z)$};
   \node[port](b2n1) [left of=b2,node distance = 7.5mm] {}; 
    \node[port](b2x1n) at (9.75,0.85) {};  
    \node[port](b2x1s) at (9.75,1.15) {};

  \node[loc](x1yes) at (11, 2.3) {$d{=}0$};
  \node[inv]() [below of =x1yes, node distance=5mm]{$\figinv{z{=}0}$};
  
   \node[loc](x1sm) at (10.5, 1.5) {$\ddot \smile$};
%   \node[loc](x1sm) at (10.5, 1.5) {$\Large{\smiley{}}$};
  \node[inv]() [below of =x1sm, node distance=5mm]{$\figinv{z{=}0}$};  
  
  \node[loc](x1no) at (11, 0.7) {${d{>}0}$};
  \node[inv]() [below of =x1no, node distance=5mm]{$\figinv{z{=}0}$};
%   \node() [right of =x1no, node distance=5mm]{$d{>}0$};
  
  \draw[trans] (n1)--(a1n1);
  \draw[trans] (a1x1) -- (a2n1);
  \draw[trans] (a2x1) -- (m);
  \draw[trans] (m) -- (m1);
  \draw[trans] (m) -- (m2);
  \draw[otrans] (m1) -- (8.5,2.7) -- (10,2.7)-- (x1yes);
  \draw[otrans] (m1) -- (b1n1);
  \draw[otrans] (m2) -- (8.5,0.3) -- (10,0.3)--(x1no);
  \draw[otrans] (m2) -- (b2n1);
  \draw[trans] (b1x1s) -- (x1sm);
  \draw[trans] (b2x1n) -- (x1sm);

% %%%%%%%%%%%%%%%%%%%%%%%%%%%%%%%%%%%%%%%%%%%%%%%%%%%%%%%%%%%%%%%%%%%%%%%%%%%%%%%%%%%%
% %%%%%%%%%%%%%%%%%%%%%    \zchk{d}  %%%%%%%%%%%%%%%%%%%%%%%%%%%%%%%%%%%%%%%%%%%%%%%%%
   \draw(0, -2.5) rectangle (10,-0.5);
   \draw (9.5, -0.3) node {$\zchk{d}$};
   
  \node[loc](n2) at (0, -1.5) {$en_{2}$};  
  \node[inv]() [below of =n2,node distance=5mm] {$\figinv{z{=}0}$};  
   
   \node[boxloc](a3) at (2, -1.5) {$\begin{array}{c} B_3{:}\Mul{y}{2}\end{array}$};
   \node[varpass] () [below of =a3, node distance=5mm]{$(x,z)$};
   \node[port](a3n1) [left of=a3,node distance = 10.5mm] {}; 
   \node[port](a3x1) [right of=a3,node distance =10.5mm] {};  
  
  \node[boxloc](a4) at (4.5, -1.5) {$\begin{array}{c} B_4{:}\Mul{x}{6}\end{array}$};
   \node[varpass] () [below of =a4, node distance=5mm]{$(y,z)$};
   \node[port](a4n1) [left of=a4,node distance = 10.5mm] {}; 
   \node[port](a4x1) [right of=a4,node distance =10.5mm] {};  
   
  \node[boxloc](a5) at (7.5, -1.5) {$\begin{array}{c} B_5{:}\Mul{x}{2}\end{array}$};
   \node[varpass] () [below of =a5, node distance=5mm]{$(y,z)$};
   \node[port](a5n1) [left of=a5,node distance = 10.5mm] {}; 
   \node[port](a5x1) [right of=a5,node distance =10.5mm] {};  
   
  \node[loc](x2e) at (10, -1) {$d{=}0$};  
  \node[inv]() [below of =x2e,node distance=4mm] {$\figinv{z{=}0}$};  

  \node[loc](x2g) at (10, -2) {$d{>}0$};  
  \node[inv]() [below of =x2g,node distance=4mm] {$\figinv{z{=}0}$};  
  
  \draw[trans] (n2) -- (a3n1);
  \draw[trans] (a3x1) -- (a4n1);
   \draw[trans] (a4x1) -- (6,-1) -- (0.6,-1) node [midway,above]{$y{<}1$} -- (a3n1);
  \draw[trans] (a4x1) -- (a5n1) node [midway,below]{$y{=}1$};
  
  \draw[trans] (a5x1) -- (8.8,-1) -- (6,-1) node[midway,above]{$x{<}2$} --(a5n1);
  \draw[trans] (a5x1) -- (x2e)node[midway,above]{$x{=}2$};
  \draw[trans] (a5x1) -- (x2g)node[midway,below]{$x{>}2$};

 %%%%%%%%%%%%%%%%%%%%%%%%%%%%%%%%%%%%%%%%%%%%%%%%%%%%%%%%%%%%%%%%%%%%%%%%%%%%%%%%%%%%
%%%%%%%%%%%%%%%%%%%%%    mul(a,n)  %%%%%%%%%%%%%%%%%%%%%%%%%%%%%%%%%%%%%%%%%%%%%%%%%

  \draw(0, -6) rectangle (6,-3);
  \draw (4.5, -2.7) node {$\Mul{a}{n}:a,b\in\set{x,y}$};

  \node[loc](n3) at (0, -3.7) {$en_{3}$};  
  \node[inv]() [below of =n3,node distance=5mm] {$\figinv{z{=}0}$};  
  
  \node[boxloc](b6) at (1.5, -3.7) {$\begin{array}{c} B_6{:}\Del\end{array}$};
  \node[varpass] () [below of =b6, node distance=5mm]{$(a,z)$};
  \node[port](b6n1) [left of=b6,node distance = 5.5mm] {}; 
   \node[oport](b6x1) [right of=b6,node distance =5.5mm] {};  
  
  \node[boxloc](b7) at (1.5, -4.8) {$\begin{array}{c} B_7{:}\chkMul{b}{n}{a}\end{array}$};
  \node[varpass] () [below of =b7, node distance=5mm]{$(z)$};
  \node[port](b7n1) [left of=b7,node distance = 7.2mm] {}; 
  \node[port](b7x1) [right of=b7,node distance = 7.2mm] {}; 
  
  \node[boxloc](b8) at (4.5, -3.7) {$\begin{array}{c} B_8{:}\Del\end{array}$};
  \node[varpass] () [below of =b8, node distance=5mm]{$(b,z)$};
  \node[port](b8n1) [left of=b8,node distance = 5.5mm] {}; 
    \node[oport](b8x1) [right of=b8,node distance =5.5mm] {};  
  
  \node[boxloc](b9) at (4.5, -4.8) {$\begin{array}{c} B_9{:}\chkEQn{a}{b}\end{array}$};
  \node[varpass] () [below of =b9, node distance=5mm]{$(z)$};
  \node[port](b9n1) [left of=b9,node distance = 6.9mm] {}; 
   \node[port](b9x1) [right of=b9,node distance = 6.9mm] {};
   
   \node[loc] (x3) at (6,-3.7) {$ex_3$};
     \node[inv]() [below of =x3,node distance=5mm] {$\figinv{z{=}0}$};  
  
   \node[loc] (x3s) at (3,-5.5) {$\ddot \smile$};
     \node[inv]() [below of =x3s,node distance=3.5mm] {$\figinv{z{=}0}$};  
  
   \draw[trans] (n3)--(b6n1)  node [midway, above]{$\set{b}$};
   \draw[otrans] (b6x1)--(2.5,-4.4)--(2,-4.4)--(0.3,-4.4)--(0.3,-4.7)--(b7n1);
   \draw[otrans] (b6x1)--node [midway, above]{\textcolor{black}{$\set{a}$}}(b8n1);
   \draw[otrans] (b8x1)--(5.5,-4.4) -- (3.4,-4.4)--(3.4,-4.7)--(b9n1);
   \draw[otrans] (b8x1)--(x3);
   \draw[trans] (b9x1)--(5.5,-4.8) -- (5.5,-5.5) --(x3s);
%    \draw[trans] (b7x1)-- +(0.5,0)-- +(0.5,-0.8) -- +(2.5,-0.8)-- (x3s);
   \draw[trans] (b7x1) -- (2.5,-4.8) -- (2.5,-5.5)-- (x3s);
   
%%%%%%%%%%%%%%%%%%%%%%%%%%%%%%%%%%%%%%%%%%%%%%%%%%%%%%%%%%%%%%%%%%%%%%%%%%%%%
%%%%%%%%%%%%%%%%%%%%%%%%%%%% \chkMul{b}{n}{a} %%%%%%%%%%%%%%%%%%%%%%%%%%%%%%%%

  \draw(0, -8) rectangle (6,-6.5);
  \draw (4.7, -6.3) node {$\chkMul{b}{n}{a}:a,b\in\set{x,y}$};

  \node[loc](n4) at (0, -7.2) {$en_{4}$};  
  \node[inv]() [below of =n4,node distance=5mm] {$\figinv{z{=}0}$};  
  
  \node[boxloc](a7) at (1.5, -7.2) {$\begin{array}{c} A_7{:}M_a\end{array}$};
  \node[varpass] () [below of =a7, node distance=5mm]{$(a,z)$};
  \node[port](a7n1) [left of=a7,node distance = 6.5mm] {}; 
   \node[port](a7x1) [right of=a7,node distance =6.5mm] {};  
  
  \node (a) at (3.7,-7.2) {$\begin{array}{l} \mbox{$n{-}1$ calls to $M_a$} \\ \mbox{$(a,z)$ pass by value} \end{array}$};
  
   \node[loc] (x4) at (6,-7.2) {$ex_4$};
   \node[inv]() [below of =x4,node distance=5mm] {$\figinv{z{=}0}$};  
  
   \draw[trans] (n4)--(a7n1);
   \draw[trans,dashed] (a7x1) -- (x4);
   \draw[trans] (a) -- (x4) node[midway,above] {$b{=}n$};

  %%%%%%%%%%%%%%%%%%%%%%%%%%%%%%%%%%%%%%%%%%%%%%%%%%%%%%%%%%%%%%%%%%%%%%%%%%%%%
 %%%%%%%%%%%%%%%%%%%%%%%%%%%% \del %%%%%%%%%%%%%%%%%%%%%%%%%%%%%%%%%%%%%%%%%%%
 
   \draw(8, -4) rectangle (10,-3);
   \draw (9.7, -2.8) node {$\Del$};
   
   \node[loc] (n5) at (8,-3.5) {$en_5$};
   \node[loc] (x5) at (10,-3.5) {$ex_5$};
   
   \draw[trans] (n5)--(x5);

   %%%%%%%%%%%%%%%%%%%%%%%%%%%%%%%%%%%%%%%%%%%%%%%%%%%%%%%%%%%%%%%%%%%%%%%%%%%%%
%%%%%%%%%%%%%%%%%%%%%%%%%%%% \chkEQn{a}{b} %%%%%%%%%%%%%%%%%%%%%%%%%%%%%%%%%%

  \draw(8, -6) rectangle (10,-5);
  \draw (9.7, -4.8) node {$\chkEQn{a}{b}$};
  
  \node[loc] (n6) at (8,-5.5) {$en_6$};
  \node[loc] (x6) at (10,-5.5) {$ex_6$};
  
  \draw[trans] (n6)--(x6) node[midway,above]{$a{=}1$} node[midway,below]{$b{=}1$};

   %%%%%%%%%%%%%%%%%%%%%%%%%%%%%%%%%%%%%%%%%%%%%%%%%%%%%%%%%%%%%%%%%%%%%%%%%%%%%
 %%%%%%%%%%%%%%%%%%%%%%%%%%%% M_a %%%%%%%%%%%%%%%%%%%%%%%%%%%%%%%%%%%%%%%%%%%
 
   \draw(8, -8) rectangle (10,-7);
   \draw (9.7, -6.8) node {$M_a$};
   
   \node[loc] (n7) at (8,-7.5) {$en_7$};
   \node[loc] (x7) at (10,-7.5) {$ex_7$};
   
   \draw[trans] (n7)--(x7)node[midway,above]{$a{=}1$};

\end{tikzpicture}

 \caption{Games on RTA with 3 clocks : Zero check $c=0?$.}
 \label{fig_3clk_rtg_zc}
 \end{center}
 
\end{figure}

\noindent{\it Time taken}: 
\begin{itemize}
\item The total time taken to reach the return port of $B:Div\{x,6\}$ is 
$\beta+\frac{\alpha}{3} < 2\beta$,  on entering $en_1$ with $y=\beta,x=\alpha,z=0$.
\item The total time taken to reach the return port of $B_3$, having entered 
the call port of $B_3$ 
with 
$y=\frac{\beta}{2}$ is $ 2\beta=4\frac{\beta}{2}$. Likewise, the total time taken to 
reach the return port of $B_4$, having entered 
 the call port of $B_4$ with $x=\frac{\alpha}{6}$ is 
 $2 \alpha = 12\frac{\alpha}{6}$. Thus, the total time to reach 
 the return port of $B_4$ after one round of multiplication of $x,y$ is $4 \frac{\beta}{2}+ 12 \frac{\alpha}{6}< 4 \beta+12 \beta=16 \beta$. 
 The second time $B_3,B_4$ loop is invoked is with $y=\beta$, $x=\alpha$,
 the times taken respectively will be $4.\beta$ and $12.\alpha$, and so on. 
 Thus, the total time taken until $y$ becomes 1 is 
 $<16(\beta+2\beta+2^2\beta+ \dots +1)<16$. Recall that $x=\alpha=\frac{1}{2^{k+c}3^{k+d}}$ and $y=\beta=\frac{1}{2^k}$ and thus $\alpha \leq \beta$ always.
 \item Once $y$ becomes 1, the $B_5:Mul\{x,2\}$ loop is taken until $x$ reaches 2 or beyond. 
 $B_5$ is entered with $x=\frac{1}{2^c3^d}=\gamma$, and $B_5$ is invoked $c+1$ times. 
 The first time $Mul\{x,2\}$ is invoked with $x=\gamma$, the time elapsed is $2\gamma$; 
 the next time $Mul\{x,2\}$ is invoked with $x=2\gamma$, the time elapsed is $4\gamma$ and so on. Thus, the total time 
 elapsed in $B_5$ loop is $2\gamma+2^2 \gamma+ \dots + 2^{c+1} \gamma <2$, where $2^{c+1}\gamma=\frac{1}{3^d}$. If $d=0$, 
 then after $c+1$ steps, the exit node $d=0$ of $ZC_{=0}^d$ is reached; if $d>0$, then 
 the loop is taken $d+2$ more times; in this case also, the total time elapsed 
 to reach the exit node $d>0$ is $<2$. 
 Thus the total time taken in $ZC_{=0}^d$ component is $<16$ from the $B_3-B_4$ loop and $<2$ from $B_5$ loop. Thus time to reach either exit of this component is $<18$.
 \item In general, the component $\Mul{a}{n}$ multiplies the value in clock $a$ by $n$. 
 If $a=\zeta$ on entering $\Mul{a}{n}$, then upon exit, its value is $a=n * \zeta$. The functioning of this component is very similar to that of $\Div{a}{n}$ described earlier.
 The time taken to reach the exit $ex_3$ is $2 * (n * \zeta)$. The time taken to reach the node $\ddot \smile$ in $\Mul{a}{n}$ is $<n$ (due to $n$ calls to $M_a$ component in $ \chkMul{b}{n}{a}$). 
 \item Total time taken in $Zero~Check:~d=0?$. Time taken to come to the return port 
 of $B$ is $<2 \beta$. No time is spent at the return port of $B$, at node $m,m_1,m_2$. No time is thus spent 
 on reaching the exit nodes $d=0$ or $d>0$ from the return port of $B$. Thus, total time taken to reach 
 an exit node of  $Zero~Check:~d=0?$ is $< 2 \beta$, on entering with $y=\beta$. 
 The time taken to reach $\ddot \smile$ node in this component is $<18+2\beta$ where $<18$ t.u is the time elapsed in component $ZC_{=0}^d$.
   \end{itemize}

\noindent{\bf{Other instructions}}: The main component to simulate other instructions are as follows.
\begin{itemize}
 \item Decrement $c$ : In main component $\incc$ of Figure \ref{prop:3clk_rtg_incc}, the second call $\Div{x}{12}$ is replaced by $\Div{x}{3}$ thus updating $x$ from $\frac{1}{2^{k+c}3^{k+d}}$ to $\frac{1}{2^{k+c}3^{k+d+1}}$ to record end of $k$ instruction. 
 \item Increment $d$ : $\Div{x}{12}$ is replaced by $\Div{x}{18}$ to update $x$ to $\frac{1}{2^{k+c+1}3^{k+d+2}}$. 
 \item Decrement $d$ : $\Div{x}{2}$ is used to update $x$ to $\frac{1}{2^{k+c+1}3^{k+d}}$ recording end of $k$ instruction.
 \item Zero check $c=0?$ : Call $B_5:\Mul{x}{2}$ is replaced by $B_5:\Mul{x}{3}$ and the time taken to reach the exits remains the same.
\end{itemize}
In all these cases, the time taken to reach $\ddot \smile$ would be $< 18$ time units(in $\Div{x}{18}$).
Also, on entering any of the main components with $y=\beta$, an exit node is reached in $<2 \beta$ units of time.

\noindent{\bf{Complete RTA}} : We obtain the full RTA simulating the two counter machine 
by connecting the entry and exit of main components of instructions according to the machine's sequence of instructions. 
If the machine halts, then the RTA has an exit node corresponding to $HALT$. 
Anytime \tort embarks on a check, a $\ddot \smile$ is reachable if \ach has simulated the instruction 
correctly. As observed above, on entering any component corresponding to an instruction 
with $y=\beta$, the exit node of that component can be reached in time $< 2 \beta$, and 
a $\ddot \smile$ node can be reached in time $<18$.   
   Now the time to reach the exit node $HALT$ is the time taken for the entire simulation of the machine. As \tort can enter any of the check components, \ach is bound to choose the correct delays to update the counters accurately.
We conclude by calculating the total time elapsed during the entire simulation. 
 We have established so far that for the $(k)$th instruction, the time elapsed is no more than $2 \beta$, 
 for $\beta=\frac{1}{2^k}$.  For the first instruction, the time elapsed is at most $2$, for the second instruction it is 
 $\frac{2}{2}$, for the third it is $\frac{2}{2^2}$ and so on. 
\begin{equation}
  \text{Total~time~duration} = \left \{
 \begin{array}{l}
 2 (1 + \frac{1}{2} +\frac{1}{4} + \frac{1}{8} + \frac{1}{16} + \cdots)\\
  =2 ( 1 + (<1))\\
  < 4
 \end{array}
 \right.
\end{equation}

We now show that the two counter machine halts iff \ach has a strategy to reach $HALT$ or $\ddot \smile$. 
Suppose the machine halts. Then the strategy for \ach is to choose the appropriate delays to update the counters in each main component. Now if \tort does not verify (by entering check components) in any of the main components, then the exit $ex_1$ of the main component is reached. If \tort decides to verify then the node $\ddot \smile$ (follows from Proposition \ref{prop:3clk_rtg_div} and \ref{prop:3clk_rtg_incc}) is reached. Thus, if \ach simulates the machine correctly then either the $HALT$ exit or $\ddot \smile$  is reached  if the machine halts.    
  
Conversely, assume that the two counter machine does not halt. Then we show that \ach has no strategy to reach either $HALT$ or $\ddot \smile$. Consider a strategy of \ach which correctly simulates all the instructions. Then $\ddot \smile$ is reached only if \tort chooses to verify. But if \tort does not choose to verify then $\ddot \smile$ can not be reached. The simulation continues and as the machine does not halt, the exit node $HALT$ is never reached. Now, consider any other strategy of \ach which does an error in simulation (in a hope to reach $HALT$). \tort could verify this, and in this case, the node  $\ddot \smile$ will not be reached as the delays are incorrect. Thus \ach can not ensure reaching $HALT$ or $\ddot \smile$  with a simulation error. 
\end{proof}

\subsection{Time bounded reachability games in RSA}
\label{undec:3ws-GF}

\begin{lemma}
\label{4sw-GF-TB}
The time bounded reachability game problem is undecidable for glitchfree recursive stopwatch automata 
  with at least 4 stopwatches.
\end{lemma}
\begin{proof}
We outline quickly the changes as compared to Lemma \ref{3clk}. The proof proceeds by the simulation 
of a two counter machine. Figure \ref{fig_3sw_rsg_inc} gives the component for incrementing 
counter $c$. There are 4 stopwatches $x,y,z,u$. The encoding 
of the counters in the variables is similar to Lemma \ref{3clk}: at the entry 
node of each main component simulating the $k$th instruction, we have 
$x=\frac{1}{2^{c+k}3^{d+k}}=\alpha$, $y=\frac{1}{2^k}=\beta$ and $z=0$, where $c,d$ are the current values of the counters. We use the extra stopwatch $u$ for rough work and hence we do not ensure that $u=0$ when a component is entered.

\noindent{\bf{Simulate increment instruction}}:
As was the case in Lemma \ref{3clk}, simulation of the $(k+1)$th instruction, incrementing $c$ amounts to 
dividing $y$ by 2 and $x$ by 12. 
 In Lemma \ref{3clk}, it was possible to pass some clocks by value, and some by reference, but 
 here, all variables must be either passed by value or by reference. 
 
 The $Div\{a,n\}$ module here is similar to that in Lemma \ref{3clk}: 
  the box $A_3:D$ in Figure \ref{fig_3clk_rtg_inc} is replaced by the node $l_1$, where only $u$ ticks and accumulates a time $t$. (Recall that $u$ is the stopwatch used for rough work and has no bearing on the encoding.) In node $l_2$, 
  only $z$ ticks.
  $l_2$ is a node belonging to \tort. 
  The time $t$ spent at $l_1$ must be exactly $t=\frac{\beta}{2}$, where $\beta=\frac{1}{2^k}$ 
  is the value of $a=y$ on entering $Div\{y,2\}$. In this case, \tort, 
  even if he enters the check module $C^{u=}_{y/2}$, 
    will reach $\ddot \smile$. 
    
    Again, note that 
  the module $C^{u=}_{y/2}$ is similar to the one in Figure  \ref{fig_3clk_rtg_inc}. We use $b=x$ for rough work in this component. Due to this the earlier value of $x$ is lost. However, this does not affect the machine simulation as only $\ddot \smile$ of $Div\{y,2\}$ is reached and not the exit node and the simulation does not continue.
   At $en_5$ (of $M_u$), we have $b=x=0$, $u=t$ and $a=y=\beta$. $a,b,u$ tick at $en_5$. 
   A time $1-t$ is spent at $en_5$, obtaining $b=1-t,u=0,a=\beta+(1-t)$ at $l$. 
   At $l$, only $b,u$ tick obtaining $b=0,u=t,a=\beta+(1-t)$ at $ex_5$. 
   A second invocation of 
    $C^{u=}_{y/2}$ gives $b=0,u=t,a=\beta+2(1-t)$.
    For $a=2$, to reach $ex_3$, we thus need $t=\frac{\beta}{2}$. The time elapsed 
    in one invocation of $M_u$ is 
1 time unit; thus a total of 2+t time units is elapsed  before reaching $\ddot \smile$ (via module $C^{u=}_{y/2}$in $\Div{a}{n}$).

   If \tort skips the check at 
    $l_2$ and proceeds to $l_3$ resetting $a=y$, we
      have at $l_3$, $z=0,u=t=\frac{\beta}{2}$ and $a=0$.
      Only $a$ ticks at $l_3$, $a$ is supposed to ``catch up'' with $u$ at $l_3$, by elapsing $t=\frac{\beta}{2}$ 
      in $l_3$. Again, at $l_4$, only $z$ ticks. \tort can verify whethere $a=u$ 
      by going to $C_{u}^{a=}$. The component 
$C_{u}^{a=}$ is exactly same as that in Figure 
\ref{fig_3clk_rtg_inc}. 
A time of $1-t$ is elapsed in 
$C_{u}^{a=}$. Thus, the time taken to reach $\ddot \smile$ 
from $C_{u}^{a=}$ is $t+t+1-t=1+t$. Thus, 
$ex_2$ is reached in time $2t=2\frac{\beta}{2}=\beta$. 
As was the case in Lemma \ref{3clk}, the time taken to reach the exit node 
of $Inc~c$, starting with $y=\beta,x=\alpha,z=0$ is 
$\beta+ 2 \frac{\alpha}{12} < 2\beta$. Also, the time taken by $Div\{a,n\}$ 
on entering with $a=\zeta$ is $2\frac{\zeta}{n}$.

To summarize, the time taken to reach the exit node of the $Inc~c$ 
component is $< 2 \beta$, on entering with $y=\beta$. 
Also,  the component $Div\{a,n\}$ divides the value in clock $a$ by $n$. 
 If $a=\zeta$ on entering $\Div{a}{n}$, then upon exit, its value is $a=\frac{\zeta}{n}$.  The time taken to reach the exit $ex_2$ is $2 * (\frac{\zeta}{n})$. The time taken to reach the node $\ddot \smile$ in $\Div{a}{n}$ is $< n+1$ (due to $n$ calls to $M_u$ component). 

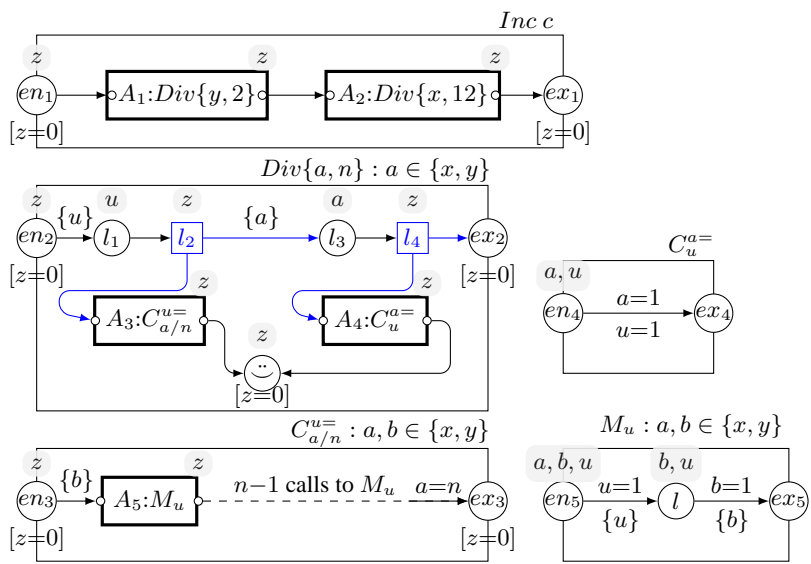
\begin{figure}
\begin{center}
 \begin{tikzpicture}[node distance=4cm] 

\draw(0, 0) rectangle (7,1.5);
\draw (6.5, 1.7) node {$\incc$};

  \node[loc](n1) at (0, 0.7) {$en_{1}$};  
  \node[inv]() [below of =n1, node distance=5mm]{$\figinv{z{=}0}$};
  \node[rate]() [above of =n1,node distance=5mm] {$\figrate{z}$};  
  
  \node[boxloc](a1) at (2, 0.7) {$\begin{array}{c} A_1{:}\Div{y}{2}\end{array}$};
  \node[port](a1n1) [left of=a1,node distance = 10mm] {}; 
  \node[port](a1x1) [right of=a1,node distance = 10mm] {};
  \node[rate]() [above of =a1x1,node distance=5mm] {$\figrate{z}$};  
  
  \node[boxloc](a2) at (5, 0.7) {$\begin{array}{c} A_2{:}\Div{x}{12}\end{array}$};
  \node[port](a2n1) [left of=a2,node distance = 11mm] {}; 
  \node[port](a2x1) [right of=a2,node distance = 11mm] {};
  \node[rate]() [above of =a2x1,node distance=5mm] {$\figrate{z}$};  
  
  \node[loc](x1) at (7, 0.7) {$ex_{1}$};
  \node[inv]() [below of =x1, node distance=5mm]{$\figinv{z{=}0}$};
%   \node[rate]() [above of =x1,node distance=5mm] {$\figrate{z}$};  
  
  \draw[trans] (n1)--(a1n1);
  \draw[trans] (a1x1) -- (a2n1);
  \draw[trans] (a2x1) -- (x1);
%%%%%%%%%%%%%%%%%%%%%%%%%%%%%%%%%%%%%%%%%%%%%%%%%%%%%%%%%%%%%%%%%%%%%%%%%%%%%%%%%%%%
%%%%%%%%%%%%%%%%%%%%%    div(a,n)  %%%%%%%%%%%%%%%%%%%%%%%%%%%%%%%%%%%%%%%%%%%%%%%%%

  \draw(0, -3.5) rectangle (6,-0.5);
  \draw (4.5, -0.3) node {$\Div{a}{n}:a\in\set{x,y}$};

  \node[loc](n2) at (0, -1.2) {$en_{2}$};  
  \node[inv]() [below of =n2,node distance=5mm] {$\figinv{z{=}0}$};  
  \node[rate]() [above of =n2,node distance=5mm] {$\figrate{z}$};  
  
  \node[loc](l1) at (1, -1.2) {$l_1$};    
  \node[rate]() [above of = l1, node distance=5mm] {$\figrate{u}$};
  
  \node[oloc](l2) at (2, -1.2) {$l_2$}; 
  \node[rate]() [above of =l2,node distance=5mm] {$\figrate{z}$};  
 %\node[rate]() [above of = l2, node distance=5mm] {$\neg z$};
  
  \node[boxloc](a3) at (1.5, -2.3) {$\begin{array}{c} A_3{:}\chkDiv{u}{a}{n}\end{array}$};

  \node[port](a3n1) [left of=a3,node distance = 7.2mm] {}; 
  \node[port](a3x1) [right of=a3,node distance = 7.2mm] {}; 
  \node[rate]() [above of =a3x1,node distance=5mm] {$\figrate{z}$};  
  
  \node[loc](l3) at (4, -1.2) {$l_3$};  
  \node[rate]() [above of =l3,node distance=5mm] {$\figrate{a}$};  
  
  \node[oloc](l4) at (5, -1.2) {$l_4$}; 
  \node[rate]() [above of =l4,node distance=5mm] {$\figrate{z}$};  
  
  \node[boxloc](a4) at (4.5, -2.3) {$\begin{array}{c} A_4{:}\chkEQn{a}{u}\end{array}$};
  \node[port](a4n1) [left of=a4,node distance = 6.9mm] {}; 
   \node[port](a4x1) [right of=a4,node distance = 6.9mm] {};
   \node[rate]() [above of =a4x1,node distance=5mm] {$\figrate{z}$};  
   
   \node[loc] (x2) at (6,-1.2) {$ex_2$};
     \node[inv]() [below of =x2,node distance=5mm] {$\figinv{z{=}0}$};  
%      \node[rate]() [above of =x2,node distance=5mm] {$\figrate{z}$};  
  
   \node[loc] (x2s) at (3,-3) {$\ddot \smile$};
     \node[inv]() [below of =x2s,node distance=3mm] {$\figinv{z{=}0}$};  
     \node[rate]() [above of =x2s,node distance=5mm] {$\figrate{z}$};  
  
   \draw[trans] (n2)--(l1)  node [midway, above]{$\set{u}$};
   \draw[trans] (l1) -- (l2);
   \draw[otrans] (l2)--(2,-1.9)--(0.3,-1.9)--(0.3,-2.2)--(a3n1);
   \draw[otrans] (l2)--node [midway, above]{\textcolor{black}{$\set{a}$}}(l3);
   \draw[trans] (l3) -- (l4);
   \draw[otrans] (l4)--(5,-1.9) -- (3.4,-1.9)--(3.4,-2.2)--(a4n1);
   \draw[otrans] (l4)--(x2);
%    \draw[trans] (a6x1)--(x2s);
%    \draw[trans] (a4x1)-- +(0.5,0)-- +(0.5,-0.8) -- +(2.5,-0.8)-- (x2s);
   \draw[trans] (a3x1) -- (2.5,-2.3)-- (2.5,-3) -- (x2s);
   \draw[trans] (a4x1) -- (5.5,-2.3) -- (5.5,-3) -- (x2s);

%%%%%%%%%%%%%%%%%%%%%%%%%%%%%%%%%%%%%%%%%%%%%%%%%%%%%%%%%%%%%%%%%%%%%%%%%%%%%
%%%%%%%%%%%%%%%%%%%%%%%%%%%% \chkEQn{a}{u} %%%%%%%%%%%%%%%%%%%%%%%%%%%%%%%%%%

  \draw(7, -3) rectangle (9,-1.5);
  \draw (8.7, -1.3) node {$\chkEQn{a}{u}$};
  
  \node[loc] (n4) at (7,-2.2) {$en_4$};
  \node[rate]() [above of =n4,node distance=5mm] {$\figrate{a,u}$};

  \node[loc] (x4) at (9,-2.2) {$ex_4$};
  
  \draw[trans] (n4)--(x4) node[midway,above]{$a{=}1$} node[midway,below]{$u{=}1$};

%%%%%%%%%%%%%%%%%%%%%%%%%%%%%%%%%%%%%%%%%%%%%%%%%%%%%%%%%%%%%%%%%%%%%%%%%%%%%
%%%%%%%%%%%%%%%%%%%%%%%%%%%% \chkDiv{b}{a}{n} %%%%%%%%%%%%%%%%%%%%%%%%%%%%%%%%

  \draw(0, -5.5) rectangle (6,-4);
  \draw (4.7, -3.8) node {$\chkDiv{u}{a}{n}:a,b\in\set{x,y}$};

  \node[loc](n3) at (0, -4.7) {$en_{3}$};  
  \node[inv]() [below of =n3,node distance=5mm] {$\figinv{z{=}0}$};  
  \node[rate]() [above of =n3,node distance=5mm] {$\figrate{z}$};    
  
  \node[boxloc](a5) at (1.5, -4.7) {$\begin{array}{c} A_5{:}M_u\end{array}$};
  \node[port](a5n1) [left of=a5,node distance = 6.5mm] {}; 
   \node[port](a5x1) [right of=a5,node distance =6.5mm] {};  
    \node[rate]() [above of =a5x1,node distance=5mm] {$\figrate{z}$};    
  
  \node (a) at (3.7,-4.7) {$\begin{array}{l} \mbox{$n{-}1$ calls to $M_u$} \\ \mbox{   } \end{array}$};
  
   \node[loc] (x3) at (6,-4.7) {$ex_3$};
   \node[inv]() [below of =x3,node distance=5mm] {$\figinv{z{=}0}$};  
%     \node[rate]() [above of =x3,node distance=5mm] {$\figrate{z}$};    
  
   \draw[trans] (n3)--(a5n1)node[midway,above] {$\set{b}$};
   \draw[trans,dashed] (a5x1) -- (x3);
   \draw[trans] (a) -- (x3) node[midway,above] {$a{=}n$};

%%%%%%%%%%%%%%%%%%%%%%%%%%%%%%%%%%%%%%%%%%%%%%%%%%%%%%%%%%%%%%%%%%%%%%%%%%%%%
%%%%%%%%%%%%%%%%%%%%%%%%%%%% M_u %%%%%%%%%%%%%%%%%%%%%%%%%%%%%%%%%%

  \draw(7, -5.5) rectangle (10,-4);
  \draw (8.7, -3.7) node {$M_u:a,b\in\set{x,y}$};
  
  \node[loc] (n5) at (7,-4.7) {$en_5$};
  \node[rate]() [above of =n5,node distance=5mm] {$\figrate{a,b,u}$};  
 % \node[rate]() [above of = n5, node distance=5mm] {$\neg z$};
  
  \node[loc] (l) at ( 8.5,-4.7) {$l$};
  \node[rate]() [above of =l,node distance=5mm] {$\figrate{b,u}$};  
 % \node[rate]() [above of = l, node distance=5mm] {$\neg z$};
  
  \node[loc] (x5) at (10,-4.7) {$ex_5$};
  
  \draw[trans] (n5)--(l) node[midway,above]{$u{=}1$}node[midway,below]{$\set{u}$};
  \draw[trans] (l)--(x5) node[midway,above]{$b{=}1$}node[midway,below]{$\set{b}$};

\end{tikzpicture}

 \caption{Games on Glitchfree-RSA with 4 stopwatches : Increment $c$. Note that the variables that tick in a location are indicated above it. Due to semantics of RHA, no time elapses in the call ports and exit nodes and hence no variables ticking is not mentioned for these locations.}
 \label{fig_3sw_rsg_inc}
 
\end{center}
\end{figure}

\noindent{\bf{Simulate Zero check instruction:}}
Here again, we illustrate the changes as compared to the zero check done in Lemma \ref{3clk}. 
Figure \ref{fig_3sw_rsg_zc} describes the zero check module.
As in the case of Figure \ref{fig_3clk_rtg_zc}, on entering $en_1$ with $x=\alpha,y=\beta,z=0$, 
we divide $y$ by 2 and $x$ by 6, to record the $(k+1)$th instruction in $x,y$. 
These modules are already discussed in the increment instruction above. 
We only discuss the module $Mul\{a,n\}$ here. This is similar to the $Div\{a,n\}$ module seen above.
If we enter $Mul\{a,n\}$ with $a=\zeta$, 
at location $l_1$, a time $t=\zeta.n$ should be spent. This makes $u=\zeta.n$, the values 
of $a,z$ are unchanged. \tort can verify that $t=\zeta.n$ using $C^{u=}_{n*a}$.
It can be seen that the $Mul\{a,n\}$ module here is similar to the $Mul\{a,n\}$ 
module in Figure \ref{fig_3clk_rtg_zc}. 

\input{./figs/fig-tb-rhg-4-zc.tex}

\noindent{\bf{Complete RSA:}} As in the case of Lemma \ref{3clk}, the complete RSA is constructed by connecting components according to the machine instructions. The time elapses in the components are exactly the same as those in Lemma \ref{3clk}. Thus the total time duration for machine simulation is $<4$. Along the same lines, we can also prove that \ach has a strategy to reach HALT or $\ddot \smile$ iff the machine halts.
\end{proof}

\begin{lemma}
 The time bounded reachability game problem is undecidable for unrestricted recursive stopwatch automata with at least 3 stopwatches. 
\end{lemma}
This follows from Lemma \ref{3clk}.

\subsection{Reachability games on RSA}
Reachability problem in recursive stopwatch automata with a single player is studied in Section \ref{sec:undec-rha}. The problem is undecidable for unrestricted recursive stopwatch automata with atleast two stopwatches. Further, it is undecidable for the glitchfree variant with atleast 3 stopwatches. The details of these results in Sections \ref{app:2sw} and \ref{app:3sw}. Due to these, following results in two player games on RSA are easy to see.

\begin{lemma}
\label{3sw-GF}
The reachability game problem is undecidable for glitchfree recursive stopwatch automata with at least 3 stopwatches. 
\end{lemma}

\begin{lemma}
 \label{2sw-UR}
 The reachability game problem is undecidable for unrestricted recursive stopwatch automata with at least 2 stopwatches. 
\end{lemma}

\section{Decidability with one player : Bounded Context RHA using only pass-by-reference}
\label{sec:dec-ref}
We mainly discuss the results of Theorem \ref{thm:dec} here; 

\subsection{Hybrid automata : Time bounded reachability \cite{BG13}}
\label{sec:ha_cnt}
% \todo{M: TODO : need to replace location with vertex for RHA explanation below.}
Time bounded reachability was shown to be decidable for hybrid automata
with no negative rates and no diagonal constraints \cite{BG13}. 
The main idea here is that if there is a run $\rho$ between two configurations $(q_1,\nu_1)$ 
and $(q_2,\nu_2)$ in a hybrid automata $H$ such that 
$\duration{\rho} \leq T$ (called $T-$time bounded run), then there exists a {\it contracted} run $\rho'$
between the same configurations, such that $\duration{\rho'} \leq T$, length of $\rho'$ is atmost $C$, a constant
 exponential in $H$ and linear in $T$, and  is dependent on $rmax$ (maximal rate in $H$) and $cmax$ (largest constant in the constraints of $H$). The construction of $\rho'$ from $\rho$ relies on a {\it contraction} operator. 
 This operator identifies positions $i<j$ in $\rho$, such that all locations between $i$ and $j$ are visited before $i$ in $\rho$ and locations $l_i = l_j$ and $e_{i+1}=e_{j+1}$ the outgoing edges from $l_i$ and $l_j$ respectively. The operator then deletes all the locations $i+1, \dots, j$ and adds their time to the other occurrences before $i$. It then connects $l_i \stackrel{e_{j+1}}{\rightarrow}{l_{j+1}}$ with sum of time delays accompanying $e_{i+1}$ and $e_{j+1}$. 
This operator is used as many times as required until a fixpoint is reached.
  Care should be taken to ensure that the \textit{contracted run} is a valid run :  
it should satisfy the constraints. To ensure this, the run $\rho$ is first carefully partitioned into exponentially many pieces, 
so that contracting the pieces and concatenating them yields a valid run.

%We shall now describe briefly, how the contraction is carefully applied to partitions of $\rho$.
%so as to respect the constraints and yield a valid run. We know that $\duration{\rho} \leq T$. 
Firstly, to help track whether the valuations resulting from contraction satisfy constraints, 
the region information is stored in the locations to form another hybrid automaton $R(H)$.
Given $cmax$, the set of regions is $\{(a-1,a), [a,a] | a \in \{1, \cdots cmax\}\} \cup \{\mathbf{0^=},\mathbf{0^+},(cmax,+\infty)\}$. It differs from the classical region notion due to lack of fractional part ordering (no diagonal constraints) and special treatment of valuations which are 0. 
$R(H)$ checks whether a variable $x$ never changes from 0 before the next transition, or if it becomes $> 0$ before the next transition. This helps bound the number of sub-runs that are constructed later, and prevents the contraction operator from merging locations 
where x remains 0 with those where x becomes $>0$. The construction ensures that $H$ admits a run between two states  of duration T iff $R(H)$ admits a run between the same
states and for the same time $T$.

As the rest of the automaton is untouched, the equivalent of run $\rho$ in $R(H)$ is a run
same as $\rho$, but having region information along with locations. Let us continue to call the run in $R(H)$ as $\rho$. 
$\rho$ is called a \textit{type-0 run}. $\rho$ is chopped into fragments of duration  $\leq \frac{1}{rmax}$,
 each of which is called a \textit{type-1} run.  There will be atmost $T . rmax +1$ type-1 runs. 
Additionally, as $rmax$ is the maximal rate of growth of any variable, a variable changes its region atmost 
3 times such that, when starting in $(b,b+1)$ region, growing through $[b+1,b+1], (b+1,b+2)$, gets reset and stays in $[0,1)$. 
Each type-1 run is further split into type-2 runs based on region changes which is atmost 
3 times per variable. 
Thus each type-1 run is split into atmost $3. |\variables|$ type-2 runs. 
Respecting region changes ensures that constraints continue to be satisfied post contraction. 
Type-2 runs are again split into type-3 runs based on the first and last reset of a variable. 
This is to enable concatenation of consecutive contracted fragments by ensuring the valuations in the start configuration and end configuration of each fragment are compatible with their 
neighbors. Each type-2 run is split into atmost $2 . |\variables| + 1$ type-3 runs. 
The contraction is applied to type-3 runs, removing second occurrences of loops. 
Hence, each contracted type-3 run will be atmost 
$|Loc'|^2+1$ long (Lemma 7 of \cite{BG13}), where $Loc'$ is the set of locations of $R(H)$. 
Note that $|Loc'| = |Loc| . (2.cmax+1)^{|\variables|}$ where $Loc$ is the set of locations of $H$.
After concatenating these contracted type-3 runs, we get contracted type-2 runs with the same start and end states.
These contracted type-2 runs are then concatenated to obtain a run $\rho'$ of 
that has the same start and end states as $\rho$, $\duration{\rho'} = \duration{\rho}$ 
and $|\rho'|\leq C = 24 . (T . rmax +1 ) . |\variables| ^2 . |Loc|^2 . (2 . cmax + 1)^{2 . |\variables|}$. To solve time-bounded reachability, we  nondeterministically guess a run of length at most $C$, and solve an  LP to check if there are time delays and valuations for each step to make the run feasible.

\subsection{Bounded-Context RHA with pass-by-reference only mechanism : Time bounded reachability}
Along the lines of contraction operator, we define a {\it context-sensitive} contraction operator $cnt$ for a run in the bounded context RHA. 
As seen in Section \ref{sec:ha_cnt}, we convert the bounded context RHA $H$ into $R(H)$, where we 
remember the respective regions along with the vertices of $H$. In the rest of this discussion, 
when we say $H$, we mean $R(H)$.  

The contraction operator in \cite{BG13} matches locations in the run while we match the (context, location) pairs of the configurations in the run. The context matching ensures that we do not alter the sequence of recursive calls made in the contracted run, thus maintaining validity w.r.t recursion. The second occurence is then deleted and the time delays are added to the first occurence of the loop. Let us denote the (context, location) pair to be used for matching as $cl = (\sseq{\kappa},q)$.
Ignoring the valuations, we denote a context as $\kappa \in B^*$ since all the variables are passed by reference and hence need not be stored in the context. Henceforth, we shall denote a run as $\rho = (\sseq{\kappa_0},q_0,\nu_0) \stackrel{t_1,e_1}{\rightarrow}   (\sseq{\kappa_1},q_1,\nu_1) \stackrel{t_2,e_2}{\rightarrow}  \cdots (\sseq{\kappa_{n-1}},q_{n-1},\nu_{n-1}) \stackrel{t_n,e_n}{\rightarrow} (\sseq{\kappa_n},q_n,\nu_n)$ where $e_i$ is the discrete transition enabled after the time delay $t_i$ in the vertex $q_{i-1}$.

\begin{definition}[Context-sensitive contraction $cnt$]
\label{def-cnt}
Consider a run $\rho =(\sseq{\kappa_0},q_0,\nu_0) \stackrel{t_1,e_1}{\rightarrow}   (\sseq{\kappa_1},q_1,\nu_1) \stackrel{t_2,e_2}{\rightarrow}  \cdots (\sseq{\kappa_{n-1}},q_{n-1},\nu_{n-1}) \stackrel{t_n,e_n}{\rightarrow} (\sseq{\kappa_n},q_n,\nu_n)$.  Assume there are two positions $0 \leq i < j <n $ and a function $ h : \{i+1, \cdots, j\} \rightarrow \{0, \cdots, i-1\}$ such that 
(i) $(\sseq{\kappa_i},q_i) = (\sseq{\kappa_j},q_j)$ and (ii) for all $ i < p < j : (\sseq{\kappa_p},q_p) = (\sseq{\kappa_{h(p)}},q_{h(p)})$. 

Then $cnt(\rho) = (\sseq{\kappa_0'},q_0',\nu_0') \stackrel{t_1',e_1'}{\rightarrow}   (\sseq{\kappa_1'},q_1',\nu_1') \stackrel{t_2',e_2'}{\rightarrow}  \cdots (\sseq{\kappa_{m-1}'},q_{m-1}',\nu_{m-1}') \stackrel{t_m',e_m'}{\rightarrow} (\sseq{\kappa_m'},q_m',\nu_m')$ where 
\begin{enumerate}
\item $m = n -(j-i)$
\item for all $ 0 \leq p < i,  (\sseq{\kappa_p'},q_p') = (\sseq{\kappa_p},q_p)$
\item for all $ 1 \leq p < i, e_p' = e_p$ and $ t_p' = t_p + \Sigma_{k \in h^{-1}(p-1)}t_{k+1}$
\item $e'_{i+1} = e_{j+1}$ and $t'_{i+1}=t_{i+1}+t_{j+1}$
\item for all $i+1 <p \leq m, (\sseq{\kappa_p'},q_p') = (\sseq{\kappa_{p+j-i}},q_{p+j-i})$
\end{enumerate}
\end{definition}

Given a run $\rho$, $cnt^0(\rho)=\rho$, $cnt^1(\rho)=cnt(\rho)$, $cnt^i(\rho)=cnt(cnt^{i-1}(\rho))$. The fixpoint $cnt^*(\rho)=cnt^n(\rho)$ such that $cnt^n(\rho)=cnt^{n-1}(\rho)$. We shall prove in the following lemmas that the length of $cnt^*(\rho)$ is independent of $\rho$.

\begin{lemma}
 \label{lem:cnt-len-type3}
 Given a type-3 run $\rho$ in the bounded context RHA $H$, $|cnt^*(\rho)| \leq (\alpha .|Q| . (2.cmax+1)^{|\variables|})^2 +1$, where 
  $\alpha=\sum\limits_{i=1}^K n^i$, $K$ is the bound on the context length, and $n$ is the number of boxes in $H$.
\end{lemma}
\begin{proof}
 Contraction of \cite{BG13}, matches the locations in a type-3 run. Thus the size of a contracted type-3 run is $|Loc|^2+1$ (Lemma 7 of \cite{BG13}) where $Loc$ is the set of locations in the region hybrid automata. However, we match (context,location) pairs in our context-sensitive contraction. 

Suppose $\rho$ is a type-3 run. Let $\rho' = cnt^*(\rho)$ have $M$ unique (context, location) pairs. 
Highlighting the first occurrences of these $M$ unique pairs and ignoring the valuations, we have in $\rho'$, 
$\underline{(\sseq{\kappa_1},q_1)} w_1 \underline{(\sseq{\kappa_2},q_2)} w_2 \dots  \underline{(\sseq{\kappa_{M-1}},q_{M-1})} w_{M-1} \underline{(\sseq{\kappa_M},q_M)} w_M$
where $w_i$ are strings over (context,location) pairs, which does not have any first occurrence of a 
(context,location) pair.  Clearly,  there are $M$ first occurences of (context, location) 
pairs $((\sseq{\kappa_i},q_i)$ for $1 \leq i \leq M$. Let a {\it portion} be a part of $\rho'$ between two such first occurences $(\sseq{\kappa_{i-1}},q_{i-1})$ and 
$(\sseq{\kappa_i},q_i)$, $2 \leq i \leq M+1$. 
  In a portion, contraction cannot be applied anymore. If it could be, then $cnt^*(\rho)$ is not a fixpoint.
    There could be a $cl$ pair such that 
    its first occurrence is at index $i$, and second occurrence is at index $j$, $i< j$; that is, 
    $cl_i=(\sseq{\kappa_i},q_i)=(\sseq{\kappa_j},q_j)=cl_j$, and the index $j$ is part of 
    a later portion ($cl_i$ is thus underlined, but $cl_j$ is not). We cannot contract the pairs at indices $i,j$, since 
         {\it all pairs} between $cl_i$ and $cl_j$ would not occur prior to $cl_i$ (if they occur, then 
         $\rho' \neq cnt^*(\rho))$. 
                  Thus the number of unique pairs in a portion could be more than 1 and can be atmost $M$. Thus the maximum length of $\rho' \leq M^2+1$. 

There are atmost $\alpha = \sum\limits_{i=1}^K n^i$ different contexts of size atmost $K$ with any sequence of the $n$ boxes including a box being called more than once. We know that $M$ is the number of unique (context, location) pairs. 
Clearly, $M \leq \alpha . |Q| . (2.cmax+1)^{|\variables|} $, where $Q=\bigcup_{i=1}^nQ_i$ is the union of the set of vertices $Q_i$ of all the $n$ components of the RHA, $|Q|.(2.cmax+1)^{|\variables|}$ is the number of vertices in the region RHA and $K$ is the context bound. Thus, the length of a type-3 contracted run is $\leq (\alpha . |Q|. (2.cmax+1)^{|\variables|} )^2+1$. Thus proved.
\end{proof}

Let us illustrate with an example why a contracted type-3 run $|cnt^*(\rho)|$ could be of length $>\alpha .|Q|. (2.cmax+1)^{|\variables|}$. Let $CL$ be a set of unique (context, location) pairs ($|CL| \leq \alpha . |Q|. (2.cmax+1)^{|\variables|}$). Assume $CL = \set{a,b,c,d,e,f}$. Let us abuse notation of a run for a short while and depict it to be only a sequence of pairs ignoring the valuations. Now let $cnt^*(\rho) = \rho'= \underline{a} \rightarrow \underline{b} \rightarrow \underline{c}\rightarrow \underline{d}\rightarrow a\rightarrow \underline{e} \rightarrow b\rightarrow \underline{f}$. Here the portions are 
$\epsilon$ (between $a$ and $b$, between $b$ and $c$, and between $c$ and $d$) and 
$a$ (between $d$ and $e$) and $b$ (between $e$ and $f$). 
%$a$, $b$, $c$, $d\rightarrow a$, $e\rightarrow b$ and $f$.
 Note that each of these portions themselves can not be contracted any further. Additionally although there are two occurences of $a$ (position 0 and 4) itself, the pairs between the first and second occurence of $a$ (these pairs are $b,c,d$) do not appear prior to $a$ at position 0. Thus contraction can not be applied to $\rho'$. Thus $|cnt^*(\rho)|$ could be $>\alpha . |Q|.(2.cmax+1)^{|\variables|}$. However, each portion itself could be atmost $\alpha . |Q|.(2.cmax+1)^{|\variables|}$ (number of unique pairs). Thus $|cnt^*(\rho)| \leq (\alpha . |Q|.(2.cmax+1)^{|\variables|})^2+1$.

\begin{lemma}
 \label{lem:cnt-len}
 Given a run $\rho$ in the bounded context RHA $H$, $|cnt^*(\rho)| \leq 24 (T.rmax +1 ) |\variables|^2  (\alpha |Q|)^2. (2 .cmax + 1)^{2 |\variables|}$.
\end{lemma}
\begin{proof}
Recall the splitting of a given run prior to contraction detailed in Section \ref{sec:ha_cnt}. The given run $\rho$ (type-0) yields $(T.rmax+1)$ type-1 runs each of which is further split into $(3.|\variables|)$ type-2 runs. Each type-2 run is split into $2.|\variables|+1$ type-3 runs. Each contracted type-3 run is atmost $(\alpha. |Q|.(2.cmax+1)^{|\variables|})^2+1$ long (Lemma \ref{lem:cnt-len-type3} above). Thus length of each contracted type-2 run is

$\begin{array}{l}
\leq [2.|\variables|+1].[(\alpha. |Q|.(2.cmax+1)^{|\variables|})^2+1] \\
\leq 2 . (|\variables|+1). ((\alpha. |Q|.(2.cmax+1)^{|\variables|})^2+1) \\
\leq 2 . (2.|\variables|). (2.(\alpha. |Q|.(2.cmax+1)^{|\variables|})^2)  
= 8 . |\variables| .(\alpha. |Q|.(2.cmax+1)^{|\variables|})^2 \\
\end{array}
$
\newline
SThus the length of $cnt^*(\rho) \leq [(T.rmax+1) . (3.|\variables|)]~.~[8 . |\variables| .(\alpha. |Q|.(2.cmax+1)^{|\variables|})^2] $\\ 
l $= 24 (T.rmax+1).|\variables|^2.(\alpha. |Q|)^2.(2.cmax+1)^{2.|\variables|}$. 

%  Splitting of $\rho$ into type-1, type-2 and type-3 runs, yields $24 (T . rmax +1 )  |\variables| ^2  (2 . cmax + 1)^{2 |\variables|}$ type-3 runs \cite{BG13}. From Lemma \ref{lem:cnt-len-type3} above, the length of a contracted type-3 run is atmost $(\alpha|Q|)^2+1$. Thus proved.
\end{proof}

\begin{lemma}
 \label{lem:cnt-valid}
 Given a run $\rho$ in the bounded context RHA $H$, $cnt^*(\rho)$ is a valid run in $H$. 
\end{lemma}
\begin{proof}
To prove that the contracted run $\rho'=cnt^*(\rho)$ is valid in the given bounded context RHA $H$, we need to ensure two conditions : 
\begin{itemize}
 \item the constraints appearing along the transitions of $\rho'$ are still satisfied and 
 \item the sequence of boxes (and mapping of call, return, entry, exit vertices) is valid w.r.t recursive calls in the given RHA. This means call port and appropriate entry node should be consecutive in the contracted run, exit node-return ports are matched and the contexts in configurations of $\rho'$ should be valid successors of the preceeding contexts.   
\end{itemize}
  
The first condition is satisfied as we consider a variant of RHA where all the variables are always 
passed by reference. Thus the context has no valuations but only a sequence of boxes. The constraints 
are guaranteed to be satisfied as the run is same as a hybrid automata run if the context is ignored. 
Thus the precautions (in carefully splitting from type-0 to type-3 runs) taken in \cite{BG13} for hybrid automata suffice with regards to constraints. 
 
The second condition is satisfied due to context-sensitive contraction where in the context is also matched 
in the loop detection. Due to this, a context in the contracted run will be a valid successor of the preceeding 
 context. We shall prove by contradiction that there exists no invalid pair of consecutive configurations in the contracted run $\rho'$. There are several ways in which a pair of configurations can be invalid predecessor/successor w.r.t recursion: 
 \begin{itemize}
  \item the call port and entry node are mismatched (either of them is missing or matched to another box's entry node)
  \item the return node and exit port are mismatched
  \item the consecutive contexts are incorrect/invalid 
 (the sequences of boxes in the two contexts are such that it is not possible in the RHA semantics to get one sequence from another via a valid RHA move.) 
 \end{itemize}

Let $\rho$ denote a run in the RHA, and let $\rho'$ be its contraction.
For ease of explanation, lets call the successor of a configuration $c$ in $\rho$ as $succ(c)$ and its predecessor as $pred(c)$.  
 Suppose there exists a pair of consecutive configurations 
 in $\rho'$ which are invalid w.r.t recursive call.  Abusing notation, we henceforth consider the context as $\kappa \in B^*$, ignoring valuations, as all variables are always passed by reference and hence need not be stored in the context.
 
Let us assume that the contracted run $\rho'$ has a pair of consecutive configurations $c' \stackrel{t,e}{\rightarrow} d'$ which are invalid as the call port configuration $c'$ is not succeeded by the appropriate entry node configuration in the contracted run i.e; $c' = (\sseq{\kappa},(b,en),\nu')$ and $d' \neq (\sseq{\kappa,b},en,\nu')$
($t=0$ by RHA semantics). 
  Consider configurations $c = (\sseq{\kappa},(b,en),\nu)$ and $succ(c)= (\sseq{\kappa,b},en,\nu)$ in the given run $\rho$ such that $c'$ corresponds to $c$. As $d' \neq (\sseq{\kappa,b},en,\nu')$, the configuration $succ(c)$ was deleted during contraction. Thus it must be the case that in $\rho$, $c$ was at position $i$ while the repeated (context,location) pairs were from position $i+1$ to $j$ and $d$ (corresponding to $d'$ in $\rho'$) was at position $j+1$. But, the configuration at $i+1$ is $succ(c)$ and this was matched with a configuration, say $e$ occuring prior to $c$ in $\rho$ : recall the contraction operator deletes 
repeating occurrences of (location, context) pairs; to delete $succ(c)$ at position $i+1$, 
we have to match the  (location, context) pair of $succ(c)$ at position $i+1$ with that of some 
$e$, occurring at a position $m < i$.  
   Due to the semantics of RHA, $pred(e)$ has (context, location) pair $(\sseq{\kappa},(b,en))$ which is the same as $c$. Thus even $c$ would be matched to $pred(e)$ and deleted. This contradicts our assumption that $c'$ (equivalent of $c$) exists in $\rho'$. 

In essence, the call port-entry node configurations always appear consecutive to each other in a given run $\rho$. Thus, matching a call-port configuration to a configuration $c$ will invariably match the corresponding entry-node configuration to the $succ(c)$ which will also be the same entry-node configuration. Similarly, exit-node and corresponding return-port configurations always appear consecutive to each other.

Now, lets consider another pair of invalid consecutive configurations $c_1' \stackrel{t,e}{\rightarrow} c_2'$ in $\rho'$ such that $c_1' = (\sseq{b_1},q_1,\nu_1')$ and $c_2' = (\sseq{b_1b_2b_3},q_2,\nu_2')$. Clearly, such a sequence is invalid under the RHA semantics (even if $q_1$ is a call port). During context-sensitive contraction, we match the (context,location) pairs and do not alter the contexts. Now consider two configurations $c_1$ and $c_2$ in $\rho$ which correspond to $c_1'$ and $c_2'$ respectively. Hence  $c_1 = (\sseq{b_1},q_1,\nu_1)$ and $c_2 = (\sseq{b_1b_2b_3},q_2,\nu_2)$ (valuations would differ in the two paths due to contraction). Obviously $c_2$ is not successor of $c_1$ in $\rho$, as $\rho$ is the given run (hence valid) and hence cant have such a pair of consecutive configurations. For $c_1' \stackrel{t,e}{\rightarrow} c_2'$ in $\rho'$, it must be the case that in $\rho$, $c_1$ is the $i$th configuration and the repeated (context,location) pairs are from $i+1$ to $j$ and $c_2$ is the 
$j+1$ configuration. 
Let $pred(c_2)$ be the predecessor of $c_2$ in $\rho$, i.e; the configuration appearing at position $j$. Contraction deletes the configurations from $i+1$ to $j$, after matching the (location, context) pairs of positions $i,j$, and hence $c_2'$ ends up as successor of $c_1'$. 
By definition \ref{def-cnt} of contraction, since we verify $(l_i,context_i) = (l_j,context_j)$, we know that the (context, location) pair of $c_1$ is the same as that of $pred(c_2)$ (recall $c_2$ occurs at position $j+1$, hence $pred(c_2)$ is at position $j$).
  But then the context of $pred(c_2)$ is $\sseq{b_1}$ which can not be followed by $\sseq{b_1b_2b_3}$ of $c_2$ in $\rho$. This contradicts our assumption that $c_2'$ succeeds $c_1'$ in $\rho'$.

Similar proof by contradiction can be given for any pair of configurations invalid w.r.t recursion.
\end{proof}

\begin{theorem}
 \label{thm:dec-ref}
 Time bounded reachability is decidable for bounded context RHA using only pass-by-reference mechanism.
\end{theorem}
\begin{proof}
 Contraction of a given run $\rho$ yields a smaller run $\rho'$ whose length is independent of $\rho$. From Lemma \ref{lem:cnt-len}, we know that $|\rho'| \leq C= 24 (T .rmax +1 )  |\variables| ^2  (\alpha |Q|)^2  (2 .cmax + 1)^{2  |\variables|}$. Additionally, from Lemma \ref{lem:cnt-valid}, $\rho'$ is a valid run in the RHA. Thus a non-deterministic algorithm (as in the case of Hybrid automata \cite{BG13}) can be used to guess a run of length atmost $C$ and then solve an  LP to check if there are time values and valuations for each step that make such a run feasible. 
\end{proof}

\subsection{Unbounded context RHA with pass-by-reference only mechanism}
We show that our adaptation of contraction does not work on RHAs with unbounded context, as  in figure \ref{fig:ref-no-cnt-app}.
We can not apply context-sensitive contraction as the context 
might grow unboundedly and matching pairs may not be found within a type-3 run. 

% In figure \ref{fig:ref-no-cnt-app}, consider a run $\rho$ 
% where $B_1$ calls $B_2$ and $B_2$ calls $B_1$ until $x=1$
% in $B_2$, after which the context shrinks to the first instance of $B_1$. In another run $\rho'$,
% $B_1$ calls $B_2$ once,  $x$ grows to be $1$ and $B_2$ returns to $B_1$. The context size is very small while 
% achieving the same effect - same start and end configurations and same duration. 

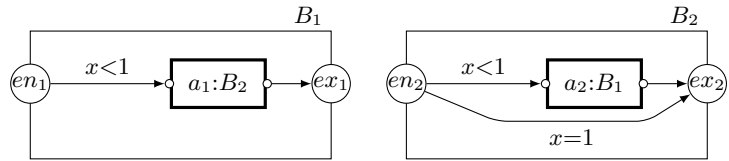
\begin{figure}[h]
%\vspace{-2em}
%\scalebox{0.65}{
\begin{center}
\begin{tikzpicture}[node distance=4cm] 
  
  \draw(0, -0.2) rectangle (4,1.5);
  \draw (3.7, 1.7) node {$B_1$};

  \node[loc](n1) at (0, 0.8) {$en_{1}$};  
  
  \node[boxloc](f1) at (2.5, 0.8) {$\begin{array}{c}~a_1{:}B_2~\\ \end{array}$};
  \node[port](f1n1) [left of=f1,node distance = 6.5mm] {}; 
  \node[port](f1x1) [right of=f1,node distance = 6.5mm] {};

  \node[loc](x1) at (4, 0.8) {$ex_{1}$};

   \draw[trans] (n1)-- (f1n1) node[midway,above]{$ x{<}1$};
   \draw[trans] (f1x1)--(x1) {};   
   %%%%%%%%%%%%%%%%%%%%%%%%%%%%%%%%%%%%%%%%%%%%%%%%%%%%%%%%%%%%%%%%%%%%%%%%%%%%%%%%%%%%%%%%%%%%%%%%%%%%%
   
   \draw(5, -0.2) rectangle (9,1.5);
  \draw (8.7, 1.7) node {$B_2$};

  \node[loc](n2) at (5, 0.8) {$en_{2}$};  
  
  \node[boxloc](f2) at (7.5, 0.8) {$\begin{array}{c}~a_2{:}B_1~\\ \end{array}$};
  \node[port](f2n1) [left of=f2,node distance = 6.5mm] {}; 
  \node[port](f2x1) [right of=f2,node distance = 6.5mm] {};

  \node[loc](x2) at (9, 0.8) {$ex_{2}$};

   \draw[trans] (n2)-- (f2n1) node[midway,above]{$ x{<}1$};
   \draw[trans] (f2x1)--(x2) {};   
   \draw[trans] (n2) -- (6.2, 0.3) -- (8.2,0.3)node [midway, below] {$x{=}1$} -- (x2);
   
   %%%%%%%%%%%%%%%%%%%%%%%%%%%%%%%%%%%%%%%%%%%%%%%%%%%%%%%%%%%%%%%%%%%%%%%%%%%%%%%%%%%%%%%%%%%%%%%%
\end{tikzpicture}
\end{center}
\caption{RHA with unbounded context}
%\vspace{-2em}
\label{fig:ref-no-cnt-app}
\end{figure}

Consider the run  $\rho = (\sseq{\varepsilon},en_1,0) \stackrel{0.1}{\rightarrow} 
  (\sseq{\varepsilon},(a_1,en_2),0.1) \stackrel{0}{\rightarrow} 
  (\sseq{a_1},en_2,0.1) \stackrel{0.1}{\rightarrow}$\\ 
  $(\sseq{a_1},(a_2,en_1),0.2) \stackrel{0}{\rightarrow} 
  (\sseq{a_1a_2},en_1,0.2) \stackrel{0.1}{\rightarrow} 
  (\sseq{a_1a_2},(a_1,en_2),0.3) \stackrel{0}{\rightarrow}\\ 
  (\sseq{a_1a_2a_1},en_2,0.3) \stackrel{0.1}{\rightarrow} 
  (\sseq{a_1a_2a_1},(a_2,en_1),0.4) \stackrel{0}{\rightarrow} 
  (\sseq{a_1a_2a_1a_2},en_1,0.4) \stackrel{0.1}{\rightarrow}\\
  (\sseq{a_1a_2a_1a_2},(a_1,en_2),0.5) \stackrel{0}{\rightarrow}
  (\sseq{a_1a_2a_1a_2a_1},en_2,0.5) \stackrel{0.5}{\rightarrow}
  (\sseq{a_1a_2a_1a_2a_1},ex_2,1) \stackrel{0}{\rightarrow}
  (\sseq{a_1a_2a_1a_2},(a_1,ex_2),1) \stackrel{0}{\rightarrow}
  (\sseq{a_1a_2a_1a_2},ex_1,1) \stackrel{0}{\rightarrow} 
  (\sseq{a_1a_2a_1},(a_2,ex_1),1) \stackrel{0}{\rightarrow}\\
  (\sseq{a_1a_2a_1},ex_2,1) \stackrel{0}{\rightarrow} 
  (\sseq{a_1a_2},(a_1,ex_2),1) \stackrel{0}{\rightarrow} 
  (\sseq{a_1a_2},ex_1,1) \stackrel{0}{\rightarrow}
  (\sseq{a_1},(a_2,ex_1),1) \stackrel{0}{\rightarrow}
  (\sseq{a_1},ex_2,1) \stackrel{0}{\rightarrow}  
  (\sseq{\varepsilon},(a_1,ex_2),1) \stackrel{0}{\rightarrow} 
  (\sseq{\varepsilon},ex_1,1)$. $\duration{\rho} = 1$. 
  
  Now consider another run \\
  $\rho' = (\sseq{\varepsilon},en_1,0) \stackrel{0.1}{\rightarrow} 
  (\sseq{\varepsilon},(a_1,en_2),0.1) \stackrel{0}{\rightarrow} 
  (\sseq{a_1},en_2,0.1) \stackrel{0.9}{\rightarrow} 
  (\sseq{a_1},ex_2,1) \stackrel{0}{\rightarrow}
  (\sseq{\varepsilon},(a_1,en_2),1) \stackrel{0}{\rightarrow}
  (\sseq{\varepsilon},ex_1,1)$. 
 \\
  Note that the start and end configurations of both the runs are the same and \\$\duration{\rho'} = 1$. However, we can not apply context-sensitive contraction in this case. There can be several other runs like $\rho$ which can have unbounded contexts but have the same effect as $\rho'$. To be able to obtain $\rho'$ from $\rho$, we would need to apply contraction to the contexts too and shorten them in a sensible manner.
 Our context-sensitive contraction studied earlier does not alter the contexts and hence does not readily extend to this class  of RHA. We conjecture this to be decidable with a double layered contraction - one altering the contexts and the other altering the valuations (as seen in \cite{BG13}).

\section{Decidability with one player : Glitch-free RHA with 2 stopwatches}
\label{sec:dec-rha}
% \subsection{Decidability of 2 stopwatch RHA}
\subsection{Region Abstraction of Hybrid Automata with 2 Stopwatch Variables}
\label{dec}
We first show that the reachability problem is decidable for  stopwatch automata (hybrid automata with only stopwatches)
with 2 stopwatch variables. 

\begin{definition}[Singular Hybrid Automata]
   A Singular Hybrid automaton is a tuple 
   $\HA = (Q, Q_0, \Sigma, X, \Delta, I, F)$ where 
  \begin{itemize}
  \item 
    $Q$ is a finite set of control \emph{modes} including a distinguished
    initial set of control modes $Q_0 \subseteq Q$, 
  \item 
    $\Sigma$ is a finite set of \emph{actions},
  \item 
    $V$ is an (ordered) set of \emph{variables}, 
  \item
    $\Delta \subseteq Q \times \rect(V) \times \Sigma \times 2^{X}
    \times Q$ is the \emph{transition relation}, 
  \item 
    $I: Q \to \rect(V)$  is the mode-invariant function, and
    \item 
    $F: Q \to \Rat^{|V|}$ is the mode-dependent \emph{flow function}
    characterizing the rate of each variable in each mode.
  \end{itemize}
\end{definition}

Recollect that, $\rect(V)$ is the set of rectangular constraints over $V$.

Let $\HA$ be a singular hybrid automata with two stopwatch variables $V=\{x,y\}$ and as both variables are stopwatches, $F: Q \to \set{0,1}^2$. Let $c_{max}$ be the maximum constant used in any of the guards of $\HA$. For simplicity, we assume that the hybrid automata does not have location invariants.

\subsubsection*{Regions and Region Automaton} 
                We consider a finite partitioning $\mathcal{R}$ of $\mathbb{R}^2$.
  For each valuation $\nu =(\nu(x),\nu(y))\in \mathbb{R}^2$, the unique element of
   $\mathcal{R}$ that contains $\nu$ is called a region, denoted $[\nu]$. We define the successors of a region $R$, 
   $Succ_{(r_x,r_y)}(R) \subseteq  \mathcal{R}$, in the following natural way:
 For $r_x,r_y \in \{0,1\}$, 
 $$R' \in Succ_{(r_x,r_y)}(R)~\mbox{if}~\exists \nu \in R, \exists t \in \mathbb{R}~\mbox{such that}~ 
[\nu+(r_x,r_y)t] =R'$$
Denote by $\nu+(r_x,r_y)t$, the valuation $(\nu(x)+r_xt,\nu(y)+r_yt)$.
We say that such a finite partition is a {\it set of regions} whenever the following condition holds:
$$R' \in Succ_{(r_x,r_y)}(R)~\mbox{iff}~\forall \nu \in R, \exists t \in  \mathbb{R}~
\mbox{such that}~[\nu+(r_x,r_y)t] = R'$$
 
The only kind of updates we consider are those where we reset  variables to 0. 
A reset $res$ maps a region $R$ to the region $res(R)$ obtained 
from $R$ by assigning value 0 to all variables which were reset to 0.  
The set of regions $\mathcal{R}$ is compatible with resets $res$ if whenever a valuation 
$\nu' \in R'$ is reachable from a valuation $\nu \in R$ after a reset, 
then $R'$ is reachable from any $\nu \in R$ by the same reset. 
% Let $res(R)$ denote the set of valuations obtained after using a reset operation 
% on all valuations of $R$. 
Formally, we have 
 $$R' \in res(R) \rightarrow \forall \nu \in R, \exists \nu' \in R'~\mbox{such that}~\nu' \in res(\nu)$$ 

The guards $\varphi$ considered in $\HA$ are boolean combinations of  
 $x \bowtie c$ where $x \in V$ and 
$c \in \mathbb{N}$ and $\bowtie \in \{<,>,\leq, \geq,=\}$.  
A region $R$ is compatible with $\varphi$ iff 
for all valuations $\nu \in R$, either $\nu \models \varphi$
or $\nu \models \neg \varphi$.

We first construct a set of regions for 2 stopwatch automata 
that are compatible with resets and guards. 
For $z \in \{x,y\}$, we define the set of intervals 
$$\mathcal{I}_z= \{[c] \mid 0 \leq c \leq c_{max}\} \cup \{(c,c+1) \mid 0 \leq c < c_{max}\} \cup\{(c_{max}, \infty)\}$$
Define $\alpha=((I_x,I_y),\prec)$, where $\prec$ is a total preorder 
on $V_0=\{ x \in V \mid I_x$ is an interval of the form $(c,c+1)\}$. The region associated with $\alpha$ denoted  
$R_{\alpha}$ is the set of valuations 
$$\{\nu \in \mathbb{R}^2 \mid \nu(x) \in I_x, \nu(y) \in I_y~\mbox{and}~[(x,y \in V_0, x \prec y)
 \leftrightarrow (frac(\nu(x)) \leq frac(\nu(y))]\}$$ The finite set $\mathcal{R}$ of all such regions $R_{\alpha}$ forms a partition of 
$\mathbb{R}^2$. 

\begin{lemma}
$\mathcal{R}$ as defined above, is a set of regions.
\end{lemma}
\begin{proof}
In the sequel, we show that $\mathcal{R}$ is  a set of regions. 
Consider $\alpha=((I_x,I_y),\prec)$. If $I_x=((c_{max}, \infty),(c_{max}, \infty))$, then 
for all $\nu \in R_{\alpha}$, for all $t \in \mathbb{R}$, $\nu+(r_x,r_y)t  \in R_{\alpha}$ for $r_x,r_y \in \{0,1\}$. 
Hence $Succ_{(r_x,r_y)}(R_{\alpha})=R_{\alpha}$.
If $Succ_{(r_x,r_y)}(R_{\alpha})\neq R_{\alpha}$, 
then there is atleast one another region 
in $Succ_{(r_x,r_y)}(R_{\alpha})$ different from $R_{\alpha}$. 
Let $C_{\alpha}$ denote the region that is  closest to region 
   to $R_{\alpha}$. Such a closest region is such that 
 $C_{\alpha} \in Succ_{(r_x,r_y)}(R_{\alpha})$, and for all $\nu \in R_{\alpha}$, for all $t \in \mathbb{R}$, if 
$\nu+(r_x,r_y)t \notin R_{\alpha}$, then $\exists t' \leq t$ such that 
$\nu+(r_x,r_y)t' \in C_{\alpha}$. 
Such a region $C_{\alpha}=((I'_x,I'_y),\prec')$ is characterized as follows: 
Let $Z=\{z \in V  \mid I_z$ is of the form $[c]\}$.
\begin{enumerate}
\item If $Z \neq \emptyset$ and $r_x=r_y=1$. Then 
\begin{itemize}
\item \[ I'_z=\left\{ \begin{array}{ll}
I_z & \mbox{if $z \notin Z$},\\
(c,c+1) & \mbox{if $z \in Z$ and $0 \leq c < c_{max}$} \\
(c_{max}, \infty) & \mbox{if $z \in Z$ and $I_z=[c_{max}]$}
\end{array} \right.\]
\item  $x \prec' y$ if  $I_x=[c]$ with $0 \leq c < c_{max}$ and $I'_y$ 
is of the form $(d,d+1)$.
\end{itemize}
\item If $Z \neq \emptyset$ and atleast one of $r_x,r_y$ is 0. Then
\begin{itemize}
\item \[ I'_z=\left\{ \begin{array}{ll}
I_z  & \mbox{if  $r_z=0$},\\
\mbox{$[c+1]$} & \mbox{ if $z \notin Z$ and $r_z=1$}\\
(c,c+1) & \mbox{if $z \in Z,r_z=1$ and $0 \leq c < c_{max}$} \\
(c_{max}, \infty) & \mbox{if $z \in Z, r_z=1$ and $I_x=[c_{max}]$}
\end{array} \right.\]
\item $x \prec' y$ if $r_x=1$ and $x \in Z, y \notin Z$.
 If ($r_x=0$ and $x \in Z$) %other guy came to integer point
 or ($r_x=1$ and $x \notin Z$)  %both are now integers
 or if $x, y \in Z$, then $V'_0=\emptyset$. 
\end{itemize}

\item If $Z=\emptyset$ and $r_x=r_y=1$. Let $M$ denote the set of variables with the maximum fractional part,
whose interval is of the form $(c,c+1)$ for $0 \leq c < c_{max}$. 
Then 
\begin{itemize}
\item \[ I'_z=\left\{ \begin{array}{ll}
I_z & \mbox{if $z \notin M$},\\
\mbox{$[c+1]$} & \mbox{if $z \in M$ and $I_z=(c,c+1)$ with $0 \leq c < c_{max}$} 
\end{array} \right.\]
\item One variable moves to an integer value, or both variables are in $(c_{max}, \infty)$. Hence $V'_0=\emptyset$.
\end{itemize}
\item If $Z=\emptyset$ and atleast one of $r_x,r_y$ is 0. 
Then 
\begin{itemize}
\item 
\[ I'_z=\left\{ \begin{array}{ll}
I_z & \mbox{if $r_z=0$},\\
\mbox{$[c+1]$} & \mbox{if $z \in M, r_z=1$ and $I_z=(c,c+1)$ with $0 \leq c < c_{max}$}\\
\mbox{$[c+1]$}  & \mbox{if $z \notin M$ and $r_z=1$ and  $I_z=(c,c+1)$ with $0 \leq c < c_{max}$}
\end{array} \right.\]
\item $x \prec' y$ is same as $x \prec y$ when $r_x=r_y=0$. Otherwise, 
one of the variables gets an integer value, and hence $V'_0=\emptyset$.
\end{itemize}
\end{enumerate}
We now claim that 
$$\forall \nu \in \alpha, \exists t \in \mathbb{R}~\mbox{such that}~\nu+t \in C_{\alpha}$$ 
Let $\nu$ be a valuation in $\alpha$. Then let $frac(\nu(x))$ denote the fractional 
part of $\nu(x)$. Similarly for $\nu(y)$.
\begin{enumerate}
\item If $Z \neq \emptyset$ and $r_x=r_y=1$. 
Let $\tau=min\{1-frac(\nu(z)) \mid I_z$ is of the form $(c,c+1)\}$. Then $\nu+(1,1)\frac{\tau}{2}$ is in the region 
$C_{\alpha}$.
\item If $Z \neq \emptyset$ and  atleast one of $r_x,r_y$ is 0.
\begin{itemize} 
\item If $x \in Z, y \notin Z$ and $r_x=1$, then pick $\tau=frac(\nu(y))$. 
Then $\nu+(1,0)\frac{\tau}{2}$ is in the region $C_{\alpha}$. 
\item If $r_x=0$ and $x \in Z$, then pick $\tau=1-frac(\nu(y))$. Then 
$\nu+(0,1)\tau$ is in the region $C_{\alpha}$.
\item If $r_x=1$ and $x \notin Z$, then  pick $\tau=1-frac(\nu(x))$.
$\nu+(1,0)\tau$ is in the region $C_{\alpha}$. 
\item If $x, y \in Z$, and $r_x=1$, then pick $\tau=0.5$. Then 
$\nu+(1,0)\tau$ is in the region $C_{\alpha}$.
\end{itemize}
\item If $Z=\emptyset$ and $r_x=r_y=1$. 
\begin{itemize}
\item Pick the variable $z \in M$. Let $\tau=1-frac(\nu(z))$. Then $\nu+(1,1)\tau$ is in the region 
$C_{\alpha}$.
\end{itemize}
\item If $Z=\emptyset$ and atleast one of $r_x,r_y$ is 0. 
\begin{itemize}
\item If $r_x=1$ and $r_y=0$. Pick $\tau=1-frac(\nu(x))$. Then 
$\nu+(1,0)\tau$ is in the region $C_{\alpha}$.
\end{itemize}
\end{enumerate}
Thus we obtain that $C_{\alpha} \in Succ_{(r_x,r_y)}(R_{\alpha})$ 
is the closest successor of $R_{\alpha}$. 
Inducting on $C_{\alpha}$, we get the closest successor 
of $C_{\alpha}$, which is the  
successor of $R_{\alpha}$,  2 steps away, and so on. 
We write $R_{\alpha} \rightarrow_n R^n_{\alpha}$ 
if $R^n_{\alpha}$ is the $n$th closest successor of $R_{\alpha}$
with respect to some choice of rates $(r_x,r_y)$. 
 This clearly means that there is a sequence 
 of regions $R^0_{\alpha},R^1_{\alpha},R^2_{\alpha}, \dots,R^n_{\alpha}$ 
 such that $R^0_{\alpha}=R_{\alpha}$, 
 and $R^{i+1}_{\alpha}$ is the closest successor 
 of $R^i_{\alpha}$ for all $1 \leq i < n$.

In this way, we can find all successors $R'_{\alpha}$ 
of $R_{\alpha}$ such that  
$R'_{\alpha} \in Succ_{(r_x,r_y)}(R_{\alpha})$ 
iff for all $\nu \in R_{\alpha}$ 
there exists some $t \in \mathbb{R}$ such that $\nu+(r_x,r_y)t \in R'_{\alpha}$.
Hence, $\mathcal{R}$ is indeed a set of regions partitioning $\mathbb{R}^2$.
\qed \end{proof}

Given two valuations $\nu_1, \nu_2 \in R_{\alpha}$ for some region $R_{\alpha}$, we 
say that $\nu_1$ and $\nu_2$ are equivalent if they lie in the same region, i.e,   $[\nu_1]=[\nu_2]$.

\begin{lemma}
$\mathcal{R}$ is compatible with the guards $\varphi$ and with the resets $res$. 
\end{lemma}
\begin{proof}
\begin{enumerate}
\item Let $R' \in res(R)$. Consider $\nu_1, \nu_2 \in R$, i.e, $[\nu_1]=[\nu_2]$.
Clearly, $\nu_1(x)$ and $\nu_2(x)$ lie in the same interval; 
same with $\nu_1(y)$ and $\nu_2(y)$.  If the operation $res$ resets $x$, 
then $res(\nu_1)=(0, \nu_1(y))$ and $res(\nu_2)=(0, \nu_2(y))$. Since 
$\nu_1(y)$ and $\nu_2(y)$ are in the same interval, we have $[res(\nu_1)]=[res(\nu_2)]$.
Similar results are obtained when $y$ is reset, or when both $x,y$ are reset.
\item Let $[\nu_1]=[\nu_2]$ be valuations in the same region $R$.
Let $\varphi$ be a guard. 
The result can be proved by structural induction on $|\varphi|$. 
If $\varphi$ is atomic of the form $x \sim c$, 
clearly, $\nu_1 \models \varphi$ iff $\nu_2 \models \varphi$, since $\nu_1$ and $\nu_2$ are 
equivalent. 
Assume for guards of size $\leq n-1$. 
It can be seen that the inductive hypothesis can be easily extended 
to guards of size $n$. 
\end{enumerate}
Thus, $\mathcal{R}$ is a finite set of regions compatible with guards and resets,  
partitioning $\mathbb{R}^2$.
\qed \end{proof}
 
Hence, we can use the region abstraction for the above set of regions to obtain a region automaton $\RegHA$ 
capturing the untimed language of $\HA$. The set of states of such a region automaton is the set $Q \times \mathcal{R}$, where $Q$ is the set of modes of $\HA$. The initial location of $\RegHA$  
is $(q_0,(0,0))$ where $q_0$ is the initial mode of $\mathcal{H}$. The transitions of $\RegHA$ are defined as 
$(q,R) \rightarrow^{a} (q',R')$ iff 
there is a region $\hat{R}$ and a transition from $q$ to $q'$ on $(\varphi,a,res)$ 
in $\mathcal{H}$ such that 
\begin{itemize}
\item $\hat{R} \in Succ_{(r_x,r_y)}(R)$. Here $r_x,r_y$ are the rates of variables $x,y$ at the state $q$ of $\mathcal{H}$,
\item For all $\nu \in \hat{R}$,  $\nu \models \varphi$, and 
\item $res(\hat{R})=R'$
\end{itemize}
 The final states of the region automaton are the states $(f,R)$ such that $f$ is a final state of $\mathcal{H}$.
 It can be seen that the language accepted by this region automaton is indeed 
  the untimed counterpart of $L(\HA)$.
  We thus have, the following result.

\begin{theorem}
\label{2sw-dec}
The reachability problem for hybrid automata with two stopwatch variables is decidable.
\end{theorem}
The decidability result above extends when we consider hybrid automata with location invariants as well.

\subsection{Region Abstraction for Glitchfree RHA with two stopwatch variables}
Given an $RHA$ $\mathcal{H}$ with two stopwatch variables
  $x,y$, 
 we define the regional equivalence relation 
$\Upsilon_R \subseteq S_{\mathcal H} \times S_{\mathcal H}$ in the following way:
For configurations $s=(\langle \kappa \rangle, q, \nu)$ and $s'=(\langle \kappa' \rangle, q', \nu')$, 
we have $(s,s') \in \Upsilon_R$, or equivalently, $[s]=[s']$ 
    if $q=q'$, $[\nu]=[\nu']$ and $[\kappa]=[\kappa']$ such that $\kappa=(b_1,\nu_1)(b_2,\nu_2)\dots(b_n,\nu_n)$ and 
    $\kappa'=(b'_1,\nu'_1)(b'_2,\nu'_2)\dots(b'_n,\nu'_n)$ are such that $b_=b'_i$ and 
    $[\nu_i]=[\nu'_i]$.
    
    A relation $B \subseteq S_{\mathcal H} \times S_{\mathcal H}$ defined over the set of configurations of a recursive stopwatch automaton is called a {\it time abstract bisimulation} 
     if for every pair of configurations $s_1,s_2 \in  S_{\mathcal H}$ such that
$(s_1,s_2) \in B$, for every timed action $(t,a) \in A_{\mathcal H}$ such that 
 $X_{\mathcal H}(s_1, (t,a))=s'_1$, there exists a 
 timed action $(t',a) \in A_{\mathcal H}$ such that 
 $X_{\mathcal H}(s_2, (t',a))=s'_2$, and $(s'_1,s'_2) \in B$. 
\begin{lemma}
\label{2sw-rha-dec}
Regional equivalence relation for  2 stopwatch glitch-free recursive automata is a time abstract bisimulation. 
\end{lemma}
\begin{proof}
Let us fix two configurations $s=(\langle \kappa \rangle, q, \nu)$ and $s'=(\langle \kappa' \rangle, q', \nu')$ such that 
$[s]=[s']$ and a timed action $(t,a) \in A_{\mathcal {H}}$ such that 
$X_{\mathcal {H}}(s, (t,a))=s_a=(\langle \kappa_a \rangle, q_a, \nu_a)$. We have to find a 
$(t',a)$ such that $X_{\mathcal {H}}(s', (t',a))=s'_a=((\langle \kappa'_a \rangle, q'_a, \nu'_a)$
 such that $[s_a]=[s'_a]$. There are three cases:
 \begin{enumerate}
 \item The state $q$ is a call port. That is, $q=(b,en) \in \call$. In this case, $t=0$, the context $\langle \kappa_a \rangle
 =\langle \kappa, (b,\nu)\rangle$, $q_a=en$ and $\nu_a=\nu$. Since $[s]=[s']$,
 we know $q'=(b,en)$ is also a call port. For $t'=0$, and
 $\langle \kappa'_a \rangle
 =\langle \kappa', (b,\nu')\rangle$, $q'_a=en$ and $\nu'_a=\nu_a$. It is clear that $[s_a]=[s'_a]$.
 \item The state $q$ is an exit node. 
 That is, $q=ex \in \Ex$. Let $\langle \kappa \rangle=\langle \kappa_*, (b, \nu_*)\rangle$ and let 
 $(b,ex) \in \return$. 
 In this case, $t=0$, the context $\langle \kappa_a \rangle= \langle \kappa_* \rangle$ and 
 $q_a=(b,ex)$ and $\nu_a=\nu[P(b):=\nu_*]$.
  Now let $\langle \kappa' \rangle=\langle  \kappa'_*, (b, \nu'_*)\rangle$. Again, since 
  $[s]=[s']$, we have $q=q'=ex$, $[\kappa_*] =[\kappa'_*] $ and hence $t'=0$, 
    $\langle \kappa'_a \rangle= \langle \kappa'_* \rangle$ 
    and $\nu'_a=\nu'[P(b):=\nu'_*]$. We have to show that $[\nu_a]=[\nu'_a]$.
    \begin{itemize}
    \item $P(b)=V$. In thus case, $\nu_a=\nu_*$ and $\nu'_a=\nu'_*$. Since we know 
    that $[\nu_*]=[\nu'_*]$, we obtain  $[\nu_a]=[\nu'_a]$.
    \item $P(b)=\emptyset$. In this case, $\nu_a=\nu$ and $\nu'_a=\nu'$ and since 
    $[\nu]=[\nu']$, we obtain $[\nu_a]=[\nu'_a]$.
        \end{itemize}
        \item If state $q$ is of any other kind, 
then the result follows  by the region equivalence of 2 stopwatch automata (Theorem \ref{2sw-dec}).        
 \end{enumerate}
 The proof is now complete.
\qed \end{proof}

Lemma \ref{2sw-rha-dec} allows us to extend the concept of regions abstraction 
to two stopwatch glitch free recursive automata.  
 \subsubsection*{Region Abstraction for 2 StopWatch Glitchfree RHA}
 Let $\mathcal{H}=(V, (\mathcal{H}_1,\dots,\mathcal{H}_1))$ be a glitch-free two stopwatch 
 RHA, where each $\mathcal{H}_i$ is a tuple 
 $(N_i, \En_i, \Ex_i, B_i,Y_i, A_i, X_i, P_i, Inv_i, E_i, J_i,F_i)$. The region abstraction 
 of $\mathcal{H}$ is a finite RSM $\mathcal{H}^{RG}=(\mathcal{H}_1^{RG},\mathcal{H}_2^{RG}, \dots, \mathcal{H}_k^{RG})$ where for each $1 \leq i \leq k$, component 
 $\mathcal{H}_i^{RG}=(N_i^{RG},\En_i^{RG},\Ex_i^{RG},B_i^{RG},Y_i^{RG},A_i^{RG},X_i^{RG})$ consists of:
 \begin{itemize}
 \item a finite set of $N_i^{RG} \subseteq (N_i \times \mathcal{R})$ of nodes such that 
 $(n,R) \in N_i^{RG}$ if $R \models Inv(n)$.\\
      Moreover, 
 $N_i^{RG}$ includes the set of entry nodes $\En_i^{RG} \subseteq \En_i \times \mathcal{R}$ 
 and exit nodes $\Ex_i^{RG} \subseteq \Ex_i \times \mathcal{R}$;
 \item a finite set $B^{RG}_i=B_i \times \mathcal{R}$ of boxes;
 \item boxes-to-components mapping $Y_i^{RG}: B_i^{RG} \rightarrow \{1,2,\dots,k\}$ is such that 
 $Y_i^{RG}(b,R)=Y_i(b)$. To each $(b,R) \in B_i^{RG}$, we associate a set of call ports 
 $\call^{RG}(b,R)$ and a set of return ports $\return^{RG}(b,R)$:
 \begin{itemize}
 \item $\call^{RG}(b,R)=\{(((b,R),en),R') \mid R' \in \mathcal{R}$ and $en \in \En_{Y_i(b)}\}$, and 
 \item $\return^{RG}(b,R)=\{(((b,R),ex),R') \mid R' \in \mathcal{R}$ and $ex \in \Ex_{Y_i(b)}\}$
  \end{itemize}
Let $\call^{RG}_i$ and $\return^{RG}_i$ be the set of call and return ports of component 
$\mathcal{H}_i^{RG}$. We write $Q_i^{RG}=N_i^{RG} \cup \call_i^{RG} \cup \return_i^{RG}$ for the vertices 
of the component $\mathcal{H}_i^{RG}$.
\item $A_i^{RG} \subseteq \mathbb{N} \times A_i$ is the set of actions such that 
if $(h,a) \in A_i^{RG}$, ($h$ is the number of region hops before taking $a$), then 
$h \leq 4^{c_{max}}$, where $c_{max}$ is the maximum constant appearing  in the 
guards;
\item a transition function $X_i^{RG}: Q_i^{RG} \times A_i^{RG} \rightarrow Q_i^{RG}$ with 
the natural condition that call ports and exit nodes do not have any outgoing transitions. 
 Also, for $q,q' \in Q_i^{RG}$,
 $(h,a) \in A_i^{RG}$, we have that $q'=X_i^{RG}(q,(h,a))$ if one of the following is true:
 \begin{itemize}
 \item $q=(n,R) \in N_i^{RG}$, there is a region $R_a$ such that 
 $R \rightarrow_h R_a$, $R_a \models E_i(n,a)$, and 
 \begin{itemize}
 \item If $q'=(n',R')$, then $R'=R_a[J_i(a):=0]$ and $X_i(n,a)=n'$.
 \item If $q'=(((b,R'),en),R'')$, then $R=R'=R''=R_a[J_i(a):=0]$ and $X_i(n,a)=(b,en)$
 \end{itemize}
 \item $q=(((b,R_{old}),ex),R_{now})$ is a return port of $\mathcal{H}_i^{RG}$. 
 Let $R=R_{old}$ if $P_i(b)=V$ and $R=R_{now}$ otherwise. 
 There exists a region $R_a$ such that $R \rightarrow_h R_a$ and 
 $R_a \models E_i((b,ex),a)$ and 
 \begin{itemize}
 \item If $q'=(n',R')$, then $R'=R_a[J_i(a):=0]$ and $X_i(n,a)=n'$
 \item If $q'=(((b,R'),en),R'')$, then $R'=R''=R_a[J_i(a):=0]$ and $X_i(n,a)=(b,en)$.
  \end{itemize}
    \end{itemize}
 \end{itemize}
The following lemma is a direct consequence of  Lemma \ref{2sw-rha-dec} and the region abstraction 
for 2 stopwatch glitchfree RHAs.
\begin{lemma}
Reachability (termination) problems and games on glitch-free two stopwatch RHA can be reduced  to
solving reachability (termination) problems and games, respectively, on the corresponding region abstraction 
$\mathcal{H}^{RG}$.
\end{lemma}
\subsection{Computational Complexity}
 The complexity for 2 stopwatch glitch free RHAs is the same as those 
 of 2 clock glitch free RTAs.

\section{Conclusion}
The main result of this paper is that time-bounded reachability problem for
recursive timed automata is undecidable for automata with five or more clocks. 
We also showed that for recursive hybrid automata the reachability problem turns
undecidable even for glitch-free variant with  three stopwatches, and the
corresponding time-bounded problem is undecidable for automata with $14$
stopwatches. 
Using the similar proof techniques we have also studied reachability games on
recursive hybrid automata, and  showed that time-bounded reachability games are
undecidable over recursive timed automata with three clocks. 
Similarly, for glitch-free recursive hybrid automata with three stopwatches
time-bounded reachability games are undecidable. 

\bibliographystyle{plain}
\bibliography{papers}

% \newpage
% \appendix
% \centerline{\bf \Large Appendix}
%   \input{appendix.tex}

\end{document}